\documentclass[12pt]{article}
\usepackage{enumerate}
\usepackage{natbib}
\usepackage{url} % not crucial - just used below for the URL 

\usepackage{amsmath}
\usepackage{amssymb}
\usepackage{amsthm}
\usepackage{amsfonts}
\usepackage[pdftex]{graphicx}
\usepackage{psfrag,epsf}
\usepackage{enumerate}
\usepackage{subcaption}
\usepackage{setspace}
\usepackage[colorlinks,citecolor=blue,urlcolor=blue]{hyperref}
\usepackage[title,titletoc]{appendix}

\usepackage[labelfont=bf]{caption}

\allowdisplaybreaks
\usepackage[ruled,linesnumbered]{algorithm2e}
\usepackage{algorithmic}
\usepackage{float}

\usepackage{graphics}
\usepackage{adjustbox}

\DeclareMathOperator{\tr}{tr}

\numberwithin{equation}{section}

\newtheorem{thm}{Theorem}
\newtheorem{lem}{Lemma}
\newtheorem{prop}{Proposition}

\newcommand{\blind}{0}

\begin{document}

\def\spacingset#1{\renewcommand{\baselinestretch}%
{#1}\small\normalsize} \spacingset{1}

%%%%%%%%%%%%%%%%%%%%%%%%%%%%%%%%%%%%%%%%%%%%%%%%%%%%%%%%%%%%%%%%%%%%%%%%%%%%%%

\if0\blind
{
  \title{\bf Spatial autoregressive model with measurement error in covariates\thanks{The methods developed in this paper were earlier circulated with the title "Causal Network Influence with Latent Homophily and Measurement Error: An Application to Therapeutic Community" in a preprint (version 1) on \url{https://arxiv.org/pdf/2203.14223v1} on 27th March 2022.}}
  \author{Subhadeep Paul \\
    Department of Statistics, The Ohio State University\\
    and \\
    Shanjukta Nath \\
    Department of Agricultural and Applied Economics, University of Georgia}
    \date{}
  \maketitle
} \fi

\if1\blind
{
  \bigskip
  \bigskip
  \bigskip
  \begin{center}
    {\LARGE\bf Spatial autoregressive model with measurement error in covariates}
\end{center}
\date{}
  \medskip
} \fi

\bigskip
\begin{abstract}
The Spatial AutoRegressive model (SAR) is commonly used in studies involving spatial and network data to estimate the spatial or network peer influence and the effects of covariates on the response, taking into account the dependence among units. While the model can be efficiently estimated with a  Quasi maximum likelihood approach (QMLE), the detrimental effect of covariate measurement error on the QMLE and how to remedy it is currently unknown. If covariates are measured with error, then the QMLE may not have the $\sqrt{n}$ convergence and may even be inconsistent even when a node is influenced by only a limited number of other nodes or spatial units. We develop a measurement error-corrected ML estimator (ME-QMLE) for the parameters of the SAR model when covariates are measured with error. The ME-QMLE possesses statistical consistency and asymptotic normality properties and we derive its limiting covariance. We consider two types of applications. The first is when the true covariate is imprecisely measured with replicated measurements or cannot be measured directly, and a proxy is observed instead. The second one involves including latent homophily factors estimated with error from the network for estimating peer influence. Our numerical results verify the bias correction property of the estimator and the accuracy of the standard error estimates in finite samples. We illustrate the method on two real datasets; i) peer influence in GPA for middle school students in New Jersey and ii) county-level death rates from the COVID-19 pandemic.
\end{abstract}

\noindent%
\vfill

\newpage
\spacingset{1.6} % DON'T change the spacing!

   \section{Introduction}

    \label{sec:intro}

In studies involving spatial and network data, it is often of interest to estimate the extent to which the outcomes of spatially close (spatial influence) or network-connected neighbors influence (peer influence) the outcomes of individuals. At the same time, one may wish to control for the effect of spatial or network dependence when estimating the effects of covariates on the response. Both of these goals can be accomplished with the Spatial autoregressive model (SAR) \citep{ord1975estimation}, which is a regression model that predicts an individual $i$'s response with a weighted average of responses of its neighboring units as well as additional covariates.  The SAR model has been extensively applied to spatial statistics and spatial econometrics \citep{anselin1998introduction,anselin1988spatial,kelejian1999generalized,li2007beyond,cressie2015statistics}.

Recently, this model has also been adopted to network-linked studies for estimating network influence, primarily when a single time point measurement of the network and outcome is available \citep{bramoulle2009identification,lee2007identification,lin2010identifying,lee2010specification,lee2004asymptotic,zhu2017network,leenders2002modeling}. Moreover, the SAR model was extended to multivariate responses in \cite{zhu2020multivariate}. The estimation methods for parameters of the model have been explored in \cite{ord1975estimation,kelejian1999generalized,lee2004asymptotic,lesage1997bayesian,smirnov2001fast}.

In the SAR model, since we have the same response variable on the left and the right-hand side of the equation, the usual least squares estimator for linear regression is not appropriate to estimate the parameters of this model. In a seminal paper, \cite{ord1975estimation} proposed a maximum likelihood estimator under the assumption of normally distributed errors whose theoretical properties were rigorously studied in \cite{lee2004asymptotic}. In particular, \cite{lee2004asymptotic} showed that under a few assumptions on the network adjacency or spatial weight matrix, the Quasi maximum likelihood estimator (QMLE), which relaxes the assumption of normally distributed errors with a distribution with finite fourth moments, is consistent and asymptotically converges to a normal distribution. The major conditions on the weight matrix translate to one unit or node being influenced by only a limited number of other nodes or units. Other estimators for the spatial influence parameter include the generalized method of moments approach of \cite{kelejian1999generalized}, the approximate maximum likelihood estimator of \cite{li2007beyond}, a generalized two stage least squares estimator with instrumental variables \cite{lee2004asymptotic}, and a Bayesian estimation method in \cite{lesage1997bayesian}. These estimators are developed under the assumption that all covariates are measured without error and they do not study the implications of introducing measurement error in covariates in the SAR framework.

There is extensive literature on correcting for measurement error in covariates in many linear and non-linear models spanning multiple decades \citep{carroll2006measurement,buonaccorsi2010measurement,nakamura1990corrected,stefanski1987conditional}. The corrected score function approach in \cite{nakamura1990corrected}, corrects the bias on the likelihood and the score functions due to measurement error and derives a corrected MLE. Recently, measurement error correction methods have been employed in the context of high dimensional linear regression \citep{datta2017cocolasso,sorensen2015measurement}, generalized linear model \citep{sorensen2018covariate} and matrix variate logistic regression \citep{fang2021matrix}.

Covariate measurement errors in the context of spatial regression (but not using the SAR model) have been considered in \cite{huque2014impact}. In the context of SAR models, \cite{suesse2018estimation} considers the scenario where the response is measured with measurement errors. More recently, \cite{eralp2023maximum} shows that the parameter estimate of the covariate is biased if the covariate (only one) is measured with error. However, \cite{eralp2023maximum} does not propose a method for correction and does not study the effect on the spatial or peer influence parameter. In another recent study, \cite{luo2022estimation} proposes an instrumental variable-based three-stage least squares (3SLS) method for estimating the SAR model when covariates are observed with error. 
 
\paragraph{Motivating application problems:} There is a lack of methodology for bias correction due to measurement errors in covariates in the SAR models. We have two types of motivation for studying this problem. First, we consider the situation where we have a proxy variable available to us in place of the actual target covariate, and we can measure the covariance of the random deviation of the proxy from the true covariate from either a subset of the original data or a related dataset that contains both variables. This motivation includes situations where we cannot measure the target covariate directly and accurately, but can obtain a number of replicated error-prone measurements.
The second application is more specific to SAR models and is in the context of using the model for estimating peer influence in networked data. While the SAR model has been widely used for estimating peer influence (\cite{bramoulle2009identification,lin2010identifying,goldsmith2013social}), a major difficulty in causal identification of peer influence is confounding due to homophily (\cite{shalizi2011homophily,goldsmith2013social}). In observational social network data, individuals typically select their friends or peers on the basis of similarities in characteristics, a phenomenon known as homophily (\cite{mcpherson2001birds}). Such characteristics are typically unobserved to the researcher, yet they may affect both the selection of peers and an outcome of interest. These variables therefore act as confounders hindering causal estimation of peer influence. A possible solution is to assume that these latent characteristics are manifested in the observed friendship or peer selections. Therefore we model the network with a latent variable model and estimate the latent homophily factors using the observed adjacency matrix similar to recent proposals in (\citep{natha2022identifying,mcfowland2021estimating}) in the context of longitudinal outcomes. The latent factors are estimated from the observed network with error. Using estimated factors instead of true unknown ones introduces bias into the estimated parameters. Therefore we propose methodology in this article to correct for this bias.

\paragraph{Contributions:}
In this paper, we develop a Measurement Error-corrected Quasi Maximum Likelihood Estimator for the SAR model (ME-QMLE). We consider the scenario where multiple covariates are subject to measurement error. We show that our bias-corrected estimator is consistent, and we derive an asymptotic distribution of the estimator under some assumptions related to the measurement error and the spatial weight or network matrix using the general theory of M-estimators. The limiting covariance matrix is different and involves additional terms than the one without measurement error in \cite{lee2004asymptotic}.  Our simulation studies verify the bias correction property of this new estimator in finite samples. We observe the estimator's performance is better when the ratio of the measurement error variance to the variance of the target covariates is lower and gradually deteriorates as this ratio increases. We also verify the performance of our closed form asymptotic covariance matrix estimator in simulations. When the covariance matrix of the measurement error is not known and must be estimated from a validation data, we further provide asymptotic covariance matrix of our estimators incorporating the error in estimating the covariance matrix. We demonstrate the new methods on two real datasets, focussing on the two types of application problems discussed previously.

There are several technical challenges for developing the methodology that we solve here. First, while the likelihood bias correction method is an adaption of the method in \cite{nakamura1990corrected} from the linear model to the context of SAR mode,  the theory is more involved due to the need for convergences to hold uniformly for all $\rho$ (spatial or network influence parameter). Second,
given the complex interaction of the vectors of measurement errors with the error variables from the outcome model, our proofs involve several new results in addition to the proofs in \cite{lee2004asymptotic}. For example, in the proof of consistency, we have additional terms involving the projection matrix $M_{2n}$ which is itself stochastic, while the projection matrix $M_n$ in \cite{lee2004asymptotic} was deterministic. In the asymptotic normality result, the limiting covariance matrix involves unconditional expectations (with respect to both model error term and measurement error term) of the negative Hessian matrix, $I(\theta_0, X)$ and the unconditional variance of the score vector $\Sigma(\theta_0,X)$. This form is different from what one would obtain if there were no measurement errors. We provide estimator of this asymptotic covariance matrix that can be calculated from data. Third, we provide results in the case when the measurement error covariance matrix is not known and must be estimated from validation data, which will itself lead to additional uncertainty. Finally, for the challenging problem of identifying endogenous peer effects from observational networked data, we provide a methodology for estimating the SAR model augmented with latent homophily variables which are common to the SAR model and the network model.

\section{Spatial autoregressive model with measurement error}
\label{methods}

We consider a set of $n$ entities or individuals denoted by $i\in\{1,2,..,n\}$ for whom we observe a social network or information on their spatial locations. We denote the symmetric spatial weight matrix or the network adjacency matrix (possibly weighted) as $A$. The diagonal elements of $A$ are assumed to be $0$.  We observe an $n$ dimensional vector $Y$ of univariate responses at the vertices of the network. Define the diagonal matrix of rowsums (degrees) $D$ such that $D_i = \sum_j A_{ij}$. The Laplacian matrix or the row-normalized adjacency matrix is defined as $L = D^{-1}A$.  We further observe an $n \times p$ matrix $X$ of measurements of $p$ dimensional covariates at each node.

The Spatial Auto-regressive model (SAR) on $n$ units is defined as follows.
\begin{equation}
Y=\rho  L Y + X \delta +  V,
\label{model}
\end{equation}

where, $V_i$'s are i.i.d. random variables with mean 0 and variance $\sigma^2 < \infty$ and $\rho$ is the spatial or network influence parameter. While we do not assume any distribution for the error terms, we do assume that the error is additive and has a finite variance, and later, we will also assume a boundedness condition on higher-order moments. For any node $i$, the variable $(LY)_i = \sum_j L_{ij}Y_j$, measures a weighted average of the responses of the network connected neighbors of $i$. Therefore, the above model asserts that the outcome of $i$ is a function of a weighted average of outcomes of its network-connected or spatially contiguous neighbors and values of the covariates.

We further differentiate between the covariates observed at each node and assume that we observe two sets of covariates $U$ and $Z$ such that $X_{n\times p}=[U_{n\times d_{1}}, Z_{n\times d_{2}}]$. $U_i \in \mathbb{R}^{d_1}$ are measured with an additive measurement error. We can write 
\[
\tilde{U}_i = U_i + \xi_i, \quad cov(\xi_i) = (\Delta_i)_{d_1 \times d_1},
\]
where the covariance matrices $\Delta_i$s can differ between units of observation (heteroskedastic error) and the error vectors $\xi_i$'s are independent of each other.   
The second set of $d_2$ covariates $Z_i \in \mathbb{R}^{d_2}$ are assumed to be observed without error. 

The error covariance matrices are assumed to be either known or estimated from a validation dataset or auxiliary sample. In application problems, researchers often are able to collect data on both $\tilde{U}$ and $U$ for a small percentage of the samples. Then, such data can be used to estimate the error covariances as internal validation data. In other contexts, an external validation dataset where both $\tilde{U}$ and $U$ are available must be used to estimate the error covariances. In another type of application we consider involving network latent homophily, the variables in $U$
are latent variables that are both related to the response and the formation of the network ties. Naturally, latent factors can be estimated from the observed network ($\tilde{U}$) but are noisy versions of the population latent factors. With an estimate of the covariance matrix of the (additive) error associated with the estimation of the latent factors, one can then represent $\tilde{U}$ in the above setup.

Therefore, the true data-generating model (and the model we would like to fit) is 
\begin{equation}
\label{eq2}
Y=\rho_0 L Y + U \beta_0 + Z \gamma_0 +  V_0,
\end{equation}
while the model we are forced to fit due to not having access to $U$ is 
\begin{equation}
Y=\rho  L Y + \tilde{U} \beta + Z \gamma +   V.
\end{equation}

Define the parameter vector $\theta=\{\beta, \gamma,\rho,\sigma^2\}$. Also from equations (\ref{model}) and (\ref{eq2}) $\delta=\{\beta,\gamma\}$. Note parameters cannot be estimated through a usual regression model since we have $Y$ on both the left and the right-hand side of the equation. Further, the term $LY$ is clearly correlated with the error term $V$. We can still obtain the parameters through a quasi-maximum likelihood estimation approach. The MLE is ``quasi" since we will derive it using the likelihood function that corresponds to  $V_i$ being a normal distribution, even though we place no such assumption on $V_i$. The log-likelihood function assuming $V_i$ follows a normal distribution is given by \cite{ord1975estimation,lee2004asymptotic}
\begin{equation}
    \label{loglikelihood}
l(\theta,Y,\tilde{U},Z) = -\frac{n}{2} \log ( 2 \pi \sigma^2) -\frac{1}{2\sigma^2}[((I-\rho L)Y - Z \gamma - \tilde{U}\beta)^T((I-\rho L)Y- Z \gamma -\tilde{U}\beta)] + \log |I-\rho L|.
\end{equation}

However, as we show in section \ref{section3}, this QMLE will be asymptotically inconsistent unless the measurement error vanishes. We construct a bias-corrected estimator (ME-QMLE) where the central idea is to correct for bias using a corrected score function methodology similar to those employed in linear and generalized linear models in \cite{stefanski1985effects,stefanski1985covariate,schafer1987covariate,nakamura1990corrected,novick2002corrected}.

\section{SAR model with measurement error}
\label{section3}
In this section, we develop the ME-QMLE for the SAR model when a set of covariates is measured with error.  We derive a maximum likelihood estimator for the SAR model similar to \cite{ord1975estimation, lee2004asymptotic}, but with a modified score function to correct for the bias introduced by the measurement error in $U$. We consider $U_i$s and $Z_i$s as fixed (nonrandom) variables. Using the notation of \cite{nakamura1990corrected}, define $\mathbb E^+$ as the expectation with respect to the random variable $V$ and $\mathbb E^{*}$ as the expectation with respect to the random variable $\xi$. Let $\mathbb E=\mathbb E^+ \mathbb E^*$ denote the unconditional expectation.
We wish to find a corrected likelihood function $l^*(\theta,Y,\tilde{U},Z)$ such that $\mathbb E^*[l^*(\theta,\tilde{U},Z)] = l(\theta, Y,U, Z)$ \citep{nakamura1990corrected}. 
It will be convenient to define a few additional notations.  Define $\tilde{X} =[\tilde{U} \quad Z]$ and as previously defined $X =[U \quad Z]$. Further define $\eta=[\xi \quad 0_{d_2}]^T$, where $0_{d_2}$ is the $d_2$ dimensional $0$ vector. Then, the data-generating model can be rewritten as in Equation \ref{model} augmented with the measurement error model, $
     \tilde{X}_i = X_i + \eta_i, \quad  cov(\eta_i) = \Omega_i= \begin{pmatrix}
    \Delta_i & 0\\
    0 & 0
    \end{pmatrix}.$
    The fitted SAR model replaces $X$ with $\tilde{X}$ in the equation for $Y$. The parameter vector is $\theta = [\delta, \rho, \sigma^2]$. Define $S(\rho) = (I-\rho L)$, and $ \quad G (\rho) = L S (\rho)^{-1}$. Define 
    $\tilde{H} (\theta) = G(\rho)(\tilde{U}\beta+Z \gamma) = G(\rho) \tilde{X}\delta,$ while $H(\theta) = G(\rho) X\delta$. 
   Therefore, $
\tilde{V}(\theta)=(I-\rho L)Y-Z\gamma - \tilde{U}\beta =S(\rho)Y-\tilde{X} \delta,$ and similarly, $V(\theta) = S(\rho)Y-X\delta$.
The function $l(\theta,Y,U,Z)$ is the same log-likelihood function as defined in equation \ref{loglikelihood}, but we replaced $\tilde{U}$ with $U$. 
Now, applying the expectation operator $\mathbb E^{*}$ on the log-likelihood function,
\begin{align*}
\mathbb E^*[l(&\theta,Y,\tilde{U},Z)]=-\frac{n}{2}\log(2\pi \sigma^{2})-\frac{1}{2\sigma^2}\mathbb E^{*}\bigg[(S(\rho)Y - Z \gamma - ({U}+\xi)\beta)^T(S(\rho)Y- Z \gamma -({U}+\xi)\beta)\bigg]\\&+ \log |I-\rho L|.\\
&=-\frac{n}{2}\log(2\pi \sigma^{2})-\frac{1}{2\sigma^2}\bigg[V(\theta)^TV(\theta)\bigg]+\log |S(\rho)| + \frac{1}{2\sigma^2}2\mathbb E^{*}\bigg[(\xi\beta)^T\bigg]V(\theta)-\frac{1}{2\sigma^2}\beta^{T}\mathbb E^{*}\bigg[\xi^{T}\xi\bigg] \beta.\\
&= l(\theta,Y,U,Z) - \frac{1}{2\sigma^2}\beta^T (\sum_i \Delta_i) \beta,
\end{align*}
Then, the ``corrected" log-likelihood function is given by
\[
l^{*}(\theta,Y,\tilde{U},Z) = -\frac{n}{2} \log ( 2 \pi \sigma^2) -\frac{1}{2\sigma^2}[ \tilde{V}(\theta)^T\tilde{V}(\theta) -  \beta^T(\sum_i \Delta_i) \beta] + \log |S(\rho)|
\]
We define the solutions to the corrected score equations $
\nabla_{\theta}l^{*}(\theta,Y,\tilde{U},Z) = 0 
$
as the bias-corrected maximum likelihood estimators.  The following quantities will be useful:
\begin{align*}
   M_2 & = I -\begin{pmatrix}
\tilde{U} &
Z
\end{pmatrix}\begin{pmatrix}
\tilde{U}^T\tilde{U} - \sum_i \Delta_i & \tilde{U}^TZ \\
Z^T\tilde{U}& Z^TZ
    \end{pmatrix}^{-1}\begin{pmatrix}
\tilde{U} \\
Z
\end{pmatrix} = I- \tilde{X}(\tilde{X}^T\tilde{X}-\sum_i \Omega_i)^{-1}\tilde{X}^T.\\
 K  & = \tilde{X}(\tilde{X}^T\tilde{X} - (\sum_i \Omega_i))^{-1} (\sum_i \Omega_i) (\tilde{X}^T\tilde{X} - (\sum_i \Omega_i))^{-1}\tilde{X}^T.
\end{align*}
The corrected score function with respect to $\beta,\gamma$ is
\begin{align*}
    \nabla_{\beta} l^* (\theta,\cdot) &= -\frac{1}{\sigma^2} [-\tilde{U}^{T}((I-\rho L)Y-Z\gamma) + \tilde{U}^T\tilde{U}\beta  - \sum_i \Delta_i \beta],\\
\nabla_{\gamma} l^* (\theta,\cdot)&=-\frac{1}{\sigma^2}[-Z^T((I-\rho L)Y-\tilde{U}\beta)+Z^{T} Z\gamma]
\end{align*}
Equating this function to 0 yields the solution
\[
\begin{pmatrix}
\hat{\beta} \\
\hat{\gamma}
\end{pmatrix} = 
\begin{pmatrix}
\tilde{U}^T\tilde{U} - \sum_i\Delta_i & \tilde{U}^TZ \\
Z^T\tilde{U}& Z^TZ
    \end{pmatrix}^{-1}\begin{pmatrix}
\tilde{U} \\
Z
\end{pmatrix} (I - \rho L)Y = (\tilde{X}^T\tilde{X}-\sum_i \Omega_i)^{-1}\tilde{X}^TS(\rho)Y.
\]
Using this solution for $\beta,\gamma$, we have the estimate for $\sigma^2$ as 
\begin{align*}
\hat{\sigma}^2  = \frac{1}{n}\{((I-\rho  L)Y - \tilde{U}\hat{\beta}-Z\hat{\gamma})^T((I-\rho  L)Y - \tilde{U}\hat{\beta}-Z\hat{\gamma})-  \hat{\beta}^T (\sum_i\Delta_{i})\hat{\beta}\}
\end{align*}
The first part can be simplified by noticing,
\begin{align*}
    (I-\rho L)Y - \tilde{U}\hat{\beta}-Z\hat{\gamma}  = S(\rho)Y - \tilde{X} \hat{\delta} = M_2 S(\rho)Y.
\end{align*}
Therefore, the first term becomes $
Y^TS(\rho)^T M_2^T M_2 S(\rho)Y$.
Note in this case, the matrix $M_2$  is not idempotent in contrast to the case of regular MLE without measurement error correction.
For the second part first notice $\hat{\beta}^T (\sum_i \Delta_i) \hat{\beta} = \hat{\delta}^T (\sum_i \Omega_i) \hat{\delta},$ and
\begin{align*}
    \hat{\delta}^T (\sum_i \Omega_i) \hat{\delta} & =  Y^TS(\rho)^T\tilde{X}(\tilde{X}^T\tilde{X} - (\sum_i \Omega_i))^{-1} (\sum_i \Omega_i) (\tilde{X}^T\tilde{X} - (\sum_i \Omega_i))^{-1}\tilde{X}^T S(\rho)Y \\
    & = Y^T S(\rho)^T K S(\rho)Y 
\end{align*}
However, with a little algebra, we note (in Appendix \ref{MTM}) $
M_2^TM_2 - K  =  I - \tilde{X}(\tilde{X}^T\tilde{X} - (\sum_i \Omega_i))^{-1} \tilde{X}^T = M_2$.
Therefore, by combining the two terms, we have a convenient expression,
\[
\hat{\sigma}^2 = \frac{1}{n}\{Y^TS(\rho)^TM_2S(\rho)Y\}.
\]
The parameter $\rho$ can be estimated by minimizing the negative of the corrected concentrated log-likelihood function obtained by replacing $\beta$, $\gamma$, and $\sigma^2$ by their estimates as follows:
 \[
 l^*(\rho)= \frac{n}{2}(-\frac {2}{n} \log |1- \rho L|) + \log ( \hat{\sigma}^2)).
 \]
If $A$ is symmetric, the minimization can be performed by writing the determinant as the product of the real eigenvalues and then using a Newton-Raphson algorithm similar to \cite{ord1975estimation}.
The first and the second derivative of $l^*(\rho)$ with respect to $\rho$ is given in the Appendix \ref{rhoderiv}.
For non-symmetric $A$ matrix, we directly perform the optimization (e.g., using R's optimize() function). The method is summarized in Algorithm \ref{alg:estimation}.

\begin{algorithm}[tb]
  \caption{Measurement Bias corrected QMLE in SAR model}
  \label{alg:estimation}
  
  \begin{algorithmic}
    \STATE {\bfseries Input:} Adjacency matrix $A$, response $Y$, observed covariates $\tilde{U}, Z$, matrices $\Delta_i$
    \STATE {\bfseries Result:} Model parameters ($\hat{\beta},\hat{\gamma},\hat{\rho},\hat{\sigma}^2$)
  \end{algorithmic}
  
  \begin{algorithmic}[1]
    \STATE $ \tilde{X} = \begin{pmatrix}
\tilde{U} &
Z
\end{pmatrix}, \quad \Omega_i= \begin{pmatrix}
    \Delta_i & 0\\
    0 & 0
    \end{pmatrix} $
    \STATE $L= D^{-1}A$ where $D_{ii} =  \sum_j A_{ij}$ and $D_{ij}=0$ for $j \neq i$
    \STATE $M_2  =  I- \tilde{X}(\tilde{X}^T\tilde{X}-\sum_i \Omega_i)^{-1}\tilde{X}^T$
\STATE $
\sigma^2(\rho) = \{Y^TS(\rho)^TM_2S(\rho)Y\}/n.$
\STATE 
$\hat{\rho}$ $\leftarrow$ minimize $ f(\rho)= \frac{n}{2}(-\frac {2}{n} \log |1- \rho L|) + \log (\sigma^2)).$
    \STATE $
\begin{pmatrix}
\hat{\beta} \\
\hat{\gamma}
\end{pmatrix} = 
(\tilde{X}^T\tilde{X}-\sum_i \Omega_i)^{-1}\tilde{X}^T (I - \hat{\rho} L)Y$
\STATE $\hat{\sigma}^2 = \sigma^2(\hat{\rho})$
    \STATE \textbf{return} $[\hat{\beta},\hat{\gamma},\hat{\rho},\hat{\sigma}^2]$
  \end{algorithmic}
\end{algorithm}

\subsection{Asymptotic theory and inference on parameters}

We study the asymptotic properties of ME-QMLE and derive its asymptotic standard error. We augment the approach in \cite{lee2004asymptotic} in the context of measurement error in variables to derive the consistency property and the limiting distribution of the measurement error bias-corrected estimator ME-QMLE. The proofs for both the consistency and the asymptotic distribution rely on techniques from the theory of M-estimators. Let the vector $\theta_0 = [\delta_0,\rho_0, \sigma^2_0]^T$ be the true parameter vector. We consider an asymptotic setting where $n \to \infty$. Therefore, in the following, we will add a subscript of $n$ to all quantities to denote sequences of quantities. The bias-corrected log-likelihood function is $l^*_n(\theta)$. The gradient vector $\nabla_{\theta}l_n^*(\theta)$ and Hessian matrix $\nabla^2_{\theta}l_n^*(\theta)$ of the log likelihood function is given in Appendix \ref{gradhess}.

For a matrix $S$, we define the notations $\|S\|_{\infty} =\max_{i} \sum_j |S_{ij}|$ as the maximum row (absolute) sum norm, $\|S\|_{1} =\max_{j} \sum_i |S_{ij}|$, as the maximum (absolute) column sum norm, and $\|S\|_2$, as the spectral norm of the matrix. We further define $\tr(S)$ as the trace of the matrix. For a matrix $S$, the term ``bounded in row sum norm" (respectively column sum norm) is used to mean $\|S\|_{\infty} < c$, (respectively $\|S\|_{1} < c$) where $c$ is not dependent on $n$. We also define a few additional notations, namely, $S_n = S_n(\rho_0)$, $G_n = G_n(\rho_0) = L_nS_n^{-1}$, and $H_n=G_nX_n\delta_0$.

\textbf{Assumptions}: We need most of the assumptions in \cite{lee2004asymptotic} and a few additional assumptions related to the measurement error. 
\begin{enumerate}
    \item Error terms $(V_n)_i, i =1, \ldots, n$, are iid with mean 0, variance $\sigma^2 <\infty$, and $E[(V_n)_i^{4+c}]$ for some $c>0$ exists. 
    \item The elements of the adjacency matrix $A_n$ are non-negative and uniformly bounded. The degrees $(d_n)_i =O(h_n)$ uniformly for all $i$ and consequently, the elements of the matrix $L_n$ are uniformly bounded by $O(1/h_n)$ for some sequence $h_n \to \infty$. Further $h_n=o(n)$, i.e., the degrees grow with $n$ but at a rate slower than $n$.
            \item $S_n(\rho)^{-1}$ is bounded in both row and column sums uniformly for all $\rho \in R$, where the parameter space $R$ is compact. The true parameter $\rho_0$ is in the interior of $R$. The sequence of matrices $L_n$ is uniformly bounded in both row and column sum norms.
     \item For each $i$, $(\tilde{X}_n)_i = (X_n)_i + (\eta_n)_i$, with $(\eta_n)_i$ are iid $d$-dimensional bounded random variables from a distribution with mean 0, covariance matrix $\Omega_i$, for all $n$ and the $(\Omega)_i$ are assumed to be known and finite. Therefore, $\lim_{n \to \infty} \frac{\sum_i \Omega_i}{n}$ is finite. The error terms $(V_n)_i$ and $(\eta_n)_i$ are independent of each other.

    \item  We assume the elements of $X_n$ are bounded constants as a function of $n$, and consequently, $\lim_{n \to \infty} \frac{1}{n}X_n^TX_n $ is finite. We further assume this limit is non-singular. The $\lim_{n \to \infty} \frac{1}{n} (X_n, H_n)^T(X_n ,H_n) $ exists and is non-singular.
    \item We assume the elements of the matrix $(\frac{1}{n} \tilde{X}_n^T \tilde{X}_n -\frac{\sum_i \Omega_i}{n})^{-1}$ are bounded.
\end{enumerate}

\begin{thm}
Under Assumptions A1-A6, the ME-QML estimator $\hat{\theta}_n$ are consistent estimators of the true population parameters, i.e., $\hat{\theta}_n \overset{p}{\to} \theta_0$.
\end{thm}

\textit{Discussion of assumptions:} Assumptions A4 and A6 are our main assumptions related to the measurement error. For each observation, the vector-valued measurement errors are bounded random vectors generated iid from a distribution with mean 0 and finite (known) covariance matrix. Similar to \cite{lee2004asymptotic}, we do not put distributional assumptions on the outcome model's error terms $V_n$s and only require those random variables to have finite fourth moments. However, for the error vectors from the measurement error model, we assume the vectors are bounded random variables. The assumption A6 is a technical condition needed since the projector matrix $M_{2n}$ is a stochastic matrix and will be involved in linear and quadratic forms that require uniform convergences. Similar assumptions on matrices involved in quadratic forms have appeared, for example, in \cite{dobriban2018high}. The assumption on the growth of the network adjacency matrix (A2) is mild and commonly assumed in the network literature. For example, the condition is satisfied if the degrees grow at the rate of $O(\log n)$, which has been assumed (in the population adjacency matrix and shown to hold with high probability in the sample) in the context of stochastic block model networks in \cite{lei2015consistency,paul2020spectral}. 

\textit{Short proof outline: }
The complete proof of this theorem is given in Appendix \ref{proofthm1}. Here, we describe a brief outline of the main arguments. We will need the following two lemmas, which are proved in the Appendix \ref{lemmaproof}.

\begin{lem}
    Let $\tilde{X}_n$ be a $n \times p$ random matrix with  uniformly bounded entries for all $i,j$. Further assume  the elements of the matrix $ (\frac{1}{n} \tilde{X}_n^T \tilde{X}_n -\frac{\sum_i \Omega_i}{n})^{-1})$ are also uniformly bounded for all $i,j$. Then for the projectors $M_{2n}$ and $ I_n- M_{2n}$ where $M_{2n} = I_n - \tilde{X}_n(\tilde{X}_n^T\tilde{X}_n - \sum_i \Omega_i)^{-1}\tilde{X}_n^T $, both the row and column absolute sums are bounded uniformly, i.e., $ \sum_j |(M_{2n})_{ij}| <\infty$ and all $i$ and $ \sum_i |(M_{2n})_{ij}| <\infty$ and all $j$.
    \label{m2bound}
\end{lem}

\begin{lem}
    Under the assumptions A1-A6, we have the following results. (1) $S_n(\rho)S_n^{-1}$ is uniformly bounded in row and column sums. (2) The elements of the vectors $X_n \delta_0$ and $G_n X_n \delta_0$ are uniformly bounded in $n$, where $G_{n}=L_{n}S_{n}^{-1}$.
    (3) $M_{2n} S_n(\rho)S_n^{-1}$ and $(S_n^{-1})^{T}S_n^T(\rho)M_{2n}S_n(\rho)S_n^{-1}$ are uniformly bounded in row and column sums.
%\end{lemma}
\label{uniformbound}
\end{lem}

Recall the corrected loglikelihood function maximized to obtain the estimators is
\[
l_n^{*}(\theta, \tilde{X}_n) = -\frac{n}{2} \log ( 2 \pi \sigma^2) -\frac{1}{2\sigma^2}[ \tilde{V}_n(\theta)^T \tilde{V}_n(\theta) -  \delta^T(\sum_i \Omega_i) \delta] + \log |S_n(\rho)|.
\]
Now define
\[
Q_n(\rho, X_n) = \max_{\delta, \sigma^2} \mathbb E[l_n^*(\theta,\tilde{X}_n)] = \max_{\delta, \sigma^2} \mathbb E^+ \mathbb E^* [l_n^*(\theta,\tilde{X}_n)],
\]
as the concentrated unconditional expectation of the corrected log likelihood function. Solving the optimization problem leads to solutions for $\delta, \sigma^2$ for a given $\rho$,
\begin{align}
    \delta(\rho) & =(X_n^TX_n)^{-1}X_n^TS_n(\rho) \mathbb E[Y_n] =  (X_n^TX_n)^{-1}X_n^TS_n(\rho) S_n^{-1}X_n \delta_0 \nonumber \\
    \sigma^2(\rho) & = \frac{1}{n}\left ((\rho_0-\rho)^2(G_n X_n\delta_0)^TM_{1n}(G_nX_n
    \delta_0\right)  + \frac{\sigma^2_0}{n} \tr \left((S_n^{-1})^{T}S_n(\rho)^TS_n(\rho)S_n^{-1} \right), 
    \label{sigmasq}
\end{align}
where $M_{1n} = (I_n - X_n (X_n^TX_n)^{-1}X_n^T)$, and $G_n=L_nS_n^{-1}$.
The arguments in \cite{lee2004asymptotic} prove the uniqueness of $\rho_0$ as the global maximizer of $Q_n(\rho,X_{n})$ in the compact parameter space $R$. 
However, we need to show the uniform convergence of the concentrated corrected log-likelihood $l_n^*(\rho, \tilde{X}_n)$ to $Q_n(\rho,X_n)$ over the compact parameter space $R$, which requires further analysis.
We start by noting,
\[
\frac{1}{n} (l^*_n(\rho, \tilde{X}_n) - Q_n(\rho, X_n)) = -\frac{1}{2}(\log (\hat{\sigma}_n^2(\rho)) - \log (\sigma^{2}(\rho))),
\]
where $\sigma^2 (\rho)$ is as defined in Equation \eqref{sigmasq} and
\begin{equation}
   \hat{\sigma}^2_n(\rho) = \frac{1}{n} Y_n^T S_n^T(\rho) M_{2n} S_n(\rho) Y_n,
\end{equation}
with $ 
M_{2n}= I_n - \tilde{X}_n(\tilde{X}_n^T\tilde{X}_n - \Omega)^{-1}
\tilde{X}^T_n.$
We can expand $\hat{\sigma}_n^2(\rho)$ to obtain
\begin{align*}
   \hat{\sigma}^2_n(\rho) &  = \frac{1}{n} (S_n^{-1}X_n \delta_0 + S_n^{-1}V_n)^T S_n^T(\rho) M_{2n} S_n(\rho) (S_n^{-1}X_n \delta_0 + S_n^{-1}V_n)\\
   & = (X_n\delta_0 + (\rho_0-\rho)G_n X_n\delta_0 + S_n(\rho) S_n^{-1}V_n)^T M_{2n} (X_n\delta_0 + (\rho_0-\rho)G_n X_n\delta_0 + S_n(\rho) S_n^{-1}V_n) \\
  & =  \frac{1}{n} \bigg ((\rho_0-\rho)^2(G_nX_n\delta_0)'M_{2n}(G_nX_n\delta_0 )   + (X_n \delta_0)^T M_{2n} (X_n \delta_0) \\
  & \quad + V_n^T(S_n^{-1})^{T} S^T_n(\rho) M_{2n}S_n (\rho) S_n^{-1}V_n  +  2 (\rho_0 -\rho) (G_nX_n \delta_0)^T M_{2n}S_n(\rho) S_n^{-1}V_n \\
  & \quad  + 2 (\rho_0 -\rho) (X_n\delta_0)^T M_{2n}(G_nX_n\delta_0) + 2(X_n\delta_0)^T M_{2n}S_n(\rho)S_n^{-1}V_n \bigg ).
\end{align*}

The decomposition of $\hat{\sigma}_n^2(\rho)$ produces 6 terms. Several of these terms are additional terms as compared to the results of \cite{lee2004asymptotic} because of the added uncertainty introduced by the measurement error model. We show that under the assumptions of A1-A6, the first term converges to the first term in Equation \eqref{sigmasq} uniformly in $\rho$, while the terms 2, 4, 5, 6 uniformly converge to 0. The convergences in terms 1, 2, and 5 do not involve $\rho$ and hence follow from the weak law of large numbers. The terms 4 and 6 require delicate analysis as we need to show uniform convergences over all $\rho$. Under the assumptions, we establish $\mathbb{E} [\sup_{\rho \in R} | (M_{2n} S_n(\rho)S_n^{-1}V_n)_i (X_n\delta_0)_i | ]< \infty$ for all $i$, and consequently applying Uniform LLN (Theorem 9.2 in \cite{keener2010theoretical}), we obtain the uniform convergence of term 6. The uniform convergence for term 4 is obtained similarly.  The 3rd term is a quadratic form in the error vector $V_n$ with a matrix $B_n(\rho,\tilde{X}_{n}) = (S_n^{-1})^{T} S^T_n(\rho) M_{2n}S_n (\rho) S_n^{-1}$ which itself is stochastic. In contrast, the commonly encountered quadratic forms are of the form $y^T A y$, with $A$ being a nonrandom matrix. We show this term uniformly converges to the second term in Equation \eqref{sigmasq} for all $\rho \in R$. 
Once the convergence of $\hat{\rho}_n$ to $\rho_0$ has been established, the convergence of $\hat{\delta}_n, \hat{\sigma}_n^2$ can be easily proved using weak law of large numbers results.

\noindent\textbf{Inconsistency and extent of bias in uncorrected estimators}: Without correcting for bias the parameter estimate $\hat{\delta}_n$ has the following asymptotic expansion
\begin{align*}
    \hat{\delta}_n & = \left(\frac{1}{n}\tilde{X}_n^T\tilde{X}_n\right)^{-1} \left(\frac{1}{n}\tilde{X}_n^T(I + (\hat{\rho}_n -\rho_0)L_nS_n^{-1})(X_n \delta_0 + V_n) \right) \\
    & = \left(\frac{1}{n}X_n^TX_n +\frac{\sum_i \Omega_i}{n} + o_p(1)\right)^{-1} \left(\frac{1}{n}X_n^TX_n\delta_0 + \frac{1}{n}(\hat{\rho}_n -\rho_0)X_n^TL_nS_n^{-1}X_n \delta_0 +o_p(1) \right).
\end{align*}
as $n \to \infty$, and where $\hat{\rho}_n$ is the estimate for $\rho$ and the other population quantities have the same meanings as defined earlier. Under the assumed conditions $X_n^TL_nS_n^{-1}X_n\delta_{0}= O(1)$. Clearly $\hat{\delta}_n$ does not converge to $\delta_0$ unless both $\hat{\rho}_n -\rho_0 = o_p(1)$ and $\frac{\sum_i \Omega_i}{n} = o(1)$. The second condition is akin to the vanishing of the measurement error. We also see the extent of the bias is dependent on the covariance of the measurement error. 

The next result establishes that the bias-corrected estimator is asymptotically normal. The result also gives us the asymptotic variance-covariance matrix of the bias-corrected estimator, which can be used to compute the standard errors of the estimates.

\begin{thm}
If assumptions A1-A6 hold, then
\[
\sqrt{n}(\hat{\theta}_n - \theta_0) \overset{D}{\to} N(0, I(\theta_0,X)^{-1}\Sigma(\theta_0,X)I(\theta_0,X)^{-1}),
\]
where $I(\theta_0,X)$ is the unconditional expectation of the negative Hessian matrix given by
\[
I(\theta_0,X) = \lim_{n \to \infty} \mathbb E^+\mathbb E^*  [-\frac{1}{n}\nabla_{\theta}^2 l^{*}_n(\theta_0, \tilde{X}_n)]
\]
which is 
\begin{align*}
     \lim_{n \to \infty} \frac{1}{n}\begin{pmatrix}
    \frac{1}{\sigma_0^2}(X_n^T X_n)  & \frac{1}{\sigma_0^2}X_n^TH_n& 0\\
\frac{1}{\sigma_0^2}H_n^TX_n  &  \frac{1}{\sigma_0^2} H_n^TH_n +  \tr (G_n^{T} G_n+G_{n}G_{n}) & \frac{1}{\sigma_0^2} \tr (G_n) \\
    0  & \frac{1}{\sigma_0^2} \tr (G_n) & \frac{n}{2\sigma_0^4},
    \end{pmatrix}
\end{align*}
where $S_n=(I-\rho_0L)$, $G_n=L_nS_n^{-1}$, and $H_n = G_nX_{n}\delta_{0}$. $\Sigma(\theta_0,X)$ is the unconditional variance of $\frac{1}{\sqrt{n}}\nabla_{\theta} l_n^{*}(\theta_0,\tilde{X}_{n})$ evaluated at $\theta_0$:  
\[\Sigma(\theta_0,X)=\lim_{n \to \infty}\mathbb E^+\mathbb E^* [(\frac{1}{\sqrt{n}}\nabla_{\theta} l_n^{*}(\theta_0,\tilde{X}_n)(\frac{1}{\sqrt{n}}\nabla_{\theta} l_n^{*}(\theta_0,\tilde{X}_n)^T].
\]
\label{asymptotic}
\end{thm}
The result proven in the Appendix uses the theory of M-estimators. The limiting variance matrix involves two terms. The term $I(\theta_0,X)$ is the usual unconditional Fisher information matrix. The limiting variance would have been just $I(\theta_0,X)^{-1}$, if there were no measurement error as proved in \cite{lee2004asymptotic}. However, due to measurement error, the limiting variance is more involved and also contains the term $\Sigma(\theta_0,X)$, which is the unconditional variance of the scaled score vector, and $\Sigma(\theta_0,X) \neq I(\theta_0,X)$.

To apply Theorem \ref{asymptotic}, we need to estimate the matrices involved in the asymptotic covariance of the limiting distribution. However, we note that we do not have $X_n$ available to us. Therefore we estimate the matrix $I(\theta_0,X)$ by the following consistent estimator $I(\theta_0,\tilde{X}_n)$, which is a function of $\hat{X}_n$:
\begin{align*}
     \frac{1}{n}\begin{pmatrix}
    \frac{1}{\sigma_0^2}(\tilde{X}_n^T \tilde{X}_n-\sum_{i}\Omega_{i})  & \frac{1}{\sigma_0^2}(\tilde{X}_n^TG_{n}\tilde{X}_{n}-\sum_{i}G_{n,ii}\Omega_{i})\delta_{0} & 0\\
\frac{1}{\sigma_0^2}\delta_{0}^T(\tilde{X}_n^TG_{n}^T\tilde{X}_n-\sum_iG_{n,ii}\Omega_{i})  & \tilde{D}_{n} & \frac{1}{\sigma_0^2} \tr (G_n) \\
    0  & \frac{1}{\sigma_0^2} \tr (G_n) & \frac{n}{2\sigma_0^4},
    \end{pmatrix}
\end{align*}

where $\tilde{D}_{n}=\frac{1}{\sigma_0^2} \delta_{0}^T(\tilde{X}_{n}^{T}G_{n}^{T}G_{n}\tilde{X}_{n}-\sum_{i}(G_{n}^{T}G_{n})_{ii}\Omega_{i})\delta_{0} +  \tr (G_n^{T} G_n+G_{n}G_{n})$. The consistency of the above estimator for estimating $I(\theta_0,X)$ can be seen by noting the following results from the Weak Law of Large Numbers
\begin{align*}
\frac{1}{n} (\tilde{X}_{n}^T\tilde{X}_{n} -\sum_i \Omega_i) &\overset{p}{\to} \frac{1}{n} X_{n}^TX_{n} \\
\frac{1}{n}(\tilde{X}_n^TG_n\tilde{X}_n - \sum_{i}G_{n,ii}\Omega_{i}) & \overset{p}{\to} \frac{1}{n} X_n^TG_nX_n\\
\frac{1}{n}(\tilde{X}_{n}^{T}G_{n}^{T}G_{n}\tilde{X}_{n}-\sum_{i}(G_{n}^{T}G_{n})_{ii}\Omega_{i}) & \overset{p}{\to} \frac{1}{n}X_n^TG_n^TG_nX_n.
\end{align*}

Next, the matrix $\Sigma(\theta_0,X)$, which is the unconditional variance of the score vector at $\theta_0$, can be estimated using the following estimator $\Sigma(\theta_0, \tilde{X}_n)$,
 \[
\frac{1}{n} \sum_i (\nabla_{\theta} l_i^{*}(\theta_0, \tilde{X}_{n,i}))(\nabla_{\theta} l_i^{*}(\theta_0, \tilde{X}_{n,i})^T,
 \]
 where $\nabla_{\theta} l_i^{*}(\theta_0, \tilde{X}_{n,i})$ is one observation score vector for observation $i$ evaluated at $\theta_0$. While the observations are not independent, given the score vector is composed of linear and quadratic forms, the following is a reasonable estimate of the one observation score vector.  
\begin{align*}
\begin{pmatrix}
    \frac{1}{\sigma_0^2} [\tilde{X}_{n,i}(\tilde{V}_{n,i}^T(\theta_0)) + \Omega_i \delta_{0}]\\
    \frac{1}{\sigma_0^2}[(L_{n}Y_{n})_i\tilde{V}_{n,i}^T(\theta_0) ] -  (G_n)_{ii} \\
   -\frac{1}{2 \sigma_0^2} + \frac{1}{2\sigma_0^4} \{\tilde{V}_{n,i}(\theta_0)(\tilde{V}_{n,i}(\theta_0))^T-  \delta^T (\Omega_i)\delta\} 
\end{pmatrix}
\end{align*}
It is easy to check that the score vector (given in Appendix \ref{gradhess}), $\nabla_{\theta} l_n^{*}(\theta_0, \tilde{X}) = \sum_i \nabla_{\theta} l_i^{*}(\theta_0, \tilde{X}_i)$. For both of these matrices, we replace $\theta_0$ with the estimator $\hat{\theta}_n$ to obtain their plug-in estimators $\hat{I}(\hat{\theta}_n, \tilde{X}_n)$ and $\hat{\Sigma}(\hat{\theta}_n, \tilde{X}_n)$.

\subsection{Estimated covariance of measurement error}

The theory in the above section was developed under the assumption that the true covariance matrix of the measurements errors $\Omega_i$s are known. However, in practice that is rarely the case and typically the error covariance needs to be estimated from replication, validation, or proxy variables. Typically, in replication datasets, each observation is replicated 2 or 3 times. Therefore, we assume that the covariance matrix of the measurement errors are the same $\Omega_i=\Omega$ for all $i$, so that we can pull information across observations for estimating $\Omega$. However, in some application problems, it is possible to have a large number of replicates for each observation, and in that case, one can attempt to estimate $\Omega_i$s separately. Let $\hat{\Omega}$ be the estimated error covariance obtained from replicated measurements. Define $\omega=vec(\Omega) = [\Omega_{11},\Omega_{12}, \Omega_{13},\ldots, \Omega_{ij},\ldots \Omega_{pp}]$ as the vector of the  elements of the matrix $\Omega$. Then following \cite{carroll2006measurement}, suppose we have access to an estimator $\hat{\Omega}$ with limiting covariance matrix $C(\Omega)$, which is independent of $\tilde{X}_n$. 
The corrected score vector is given by 
\begin{align*}
 \nabla_{\theta} l^* = \begin{pmatrix}
       \frac{1}{\sigma_0^2}[X_n^T V_n + \eta_n^TV_n  - X_n^T\eta_n \delta_0 - \eta_n^T\eta_n\delta_0 + n\Omega \delta_0 ],\\
 \frac{1}{\sigma_0^2} [(G_nX_n\delta_0)^TV_n+ V_n^TG_nV_n -(G_nX_n\delta_0)^T\eta_n \delta_0 - V_n^TG_n \eta_n \delta_0 ] - \tr (G_n)\\
 -\frac{n}{2 \sigma_0^2} + \frac{1}{2\sigma_0^4} \{V_n^TV_n- V_n^T\eta_n\delta_0  +\delta_0^T\eta_n^T\eta_n\delta_0 - \delta_0^T (n \Omega) \delta_0\}.
 \end{pmatrix}   
\end{align*}

Then we compute     $D(\theta,\Omega) = \nabla_{\omega} \nabla_{\theta} l^*(\theta, Y,\tilde{U},Z)$ as the $(p+2) \times p^2$ matrix of derivatives of the score vector with respect to the elements of $\omega$. We provide an explicit formula for this matrix in the lemma below.
\begin{lem}
\label{estcov}
    Consider the settings of Theorem \ref{asymptotic} and assume that $\Omega$ is unknown but finite. Let $\hat{\Omega}$ be a  estimator of $\Omega$ which is independent of $\tilde{X}_n$ and $\hat{\theta}_n$ and has asymptotic covariance matrix $C(\Omega)$. Then
    \[
\sqrt{n}(\hat{\theta}_n - \theta_0) \overset{D}{\to} N(0, I(\theta_0,X)^{-1}(\Sigma(\theta_0,X) + D(\theta_0,\Omega)C(\Omega)D(\theta_0, \Omega)^T)I(\theta_0,X)^{-1}),
\]
where
\begin{align*}
    D(\theta,\Omega) =  \begin{pmatrix}
        \frac{n}{\sigma^2}I_p \otimes \delta^T \\
        0_{p^2}^T \\
        \frac{n}{2*\sigma^4}\delta^T \otimes \delta^T
    \end{pmatrix}_{(p+2) \times p^2},
\end{align*}
with the notation $\otimes$ denoting the Kronecker product. 
\end{lem}
Note the function $D(\theta_0, \Omega)$ only depends on $\theta_0$ and does not depend on $\tilde{X}_n$ or $\Omega$. Therefore we replace $\theta_0$ with $\hat{\theta}_n$ to obtain a plug-in estimator. Therefore, with the above lemma we can modify the estimator for the asymptotic variance of $\hat{\theta}_n$ as 
\[
\hat{I}^{-1}(\hat{\theta}_n,\tilde{X}_n) (\hat{\Sigma}(\hat{\theta}_n, \tilde{X}_n) + \hat{D}_n(\hat{\theta}_n, \hat{\Omega}_n) C_n(\hat{\Omega}_n)\hat{D}_n(\hat{\theta}_n, \hat{\Omega}_n)^T)\hat{I}^{-1}(\hat{\theta}_n,\tilde{X}_n). 
\]
We further note that the condition of the Lemma on $\hat{\Omega}_n$ being independent of $\tilde{X}_n$ is not too strong since that holds, e.g., when we have replicated measurements of a normally distributed $X$, and we use sample mean to estimate $\tilde{X}_n$ and sample covariance to estimate $\hat{\Omega}$.

\section{Application Problems}

\subsection{Covariates measured with error}
\label{dataexample}
We consider a few scenarios where some covariates are measured with error, and the proposed methods are applicable. The key quantity we need to apply the proposed method is an estimate of the error covariance matrix $\Delta$ that can be calculated from data (either from the same data or another related validation dataset).

The first situation arises when the true $U$ cannot be directly measured, but several replicate measurements $U^R$ are available \citep{carroll2006measurement}. Let $U^R_{ij} = U_i + \epsilon_{ij}$, for $j=\{1,\ldots, k_i\},$ with typically $k_i>1$ for all observations. Then we use the sample mean $\tilde{U}_{i} = \frac{1}{k_i}\sum_j U^R_{ij}$ over the replicate measurements as the proxy variable in the SAR model. The estimate for the error covariance matrix $\Delta$ can be obtained as shown in \cite{carroll2006measurement}
\[
\hat{\Delta} = \frac{1}{k} \frac{\sum_i \sum_j (U^R_{ij} - \tilde{U}_{i})(U^R_{ij} - \tilde{U}_{i})^T}{n (k -1)}.
\]
To use Lemma \ref{estcov} we can further estimate the variance of the above covariance matrix estimator as \cite{carroll2006measurement} $C_n(\hat{\Delta}) = \frac{1}{k}\frac{\sum_i d_i d_i^T}{(n(k-1))^2}$, with $d_i = vec((U^R_{ij} - \tilde{U}_{i})(U^R_{ij} - \tilde{U}_{i})^T) - (k -1)vec(\hat{\Delta})$.
A situation like this also arises in our data example. In our application problem in section \ref{coviddata}, we are interested in a target covariate that measures the mortality rate in the older population of a county. The target covariate is not directly observed, but for this purpose, we can use 3 measured covariates, which measure mortality rate in individuals of age group 65-74, 75-84, and 84+, respectively. In the notation defined above, $k=3$, and we obtain $\tilde{U}$ as mean of the 3 measurements while $\Delta$ is obtained using the formula given above.

We consider another scenario of covariates measured with error following the examples in \cite{blackwell2017unified}.  If a validation sample is available where both $U$ and $\tilde{U}$ are available (either internal, i.e., as a subset of the data, or from an external database) then we can compute the sample covariance of the error as $\hat{\Delta} =\frac{1}{n-1}\sum_i (U_i- \tilde{U}_i)(U_i- \tilde{U}_i)^T $. This scenario includes the case when the target covariate $U$ is observed only for a subset of observations while the error-prone covariate $\tilde{U}$ is available for all observations. 
In that case, we can compute $\hat{\Delta} =\frac{1}{n_{val}-1}\sum_i^{n_{val}} (U_i- \tilde{U}_i)(U_i- \tilde{U}_i)^T $, using the validation portion of the data as the covariance of the measurement error associated with using $\tilde{U}$ in place of $U$.
The proxy covariate need not be exactly an error-prone measurement of the true target covariate. The measurement error correction method will work if we can assume that they are related through a bias term that can also be estimated from data \cite{blackwell2017unified}. Let $
U_i^{\text{proxy}} = \beta + U_i + \epsilon_i$, with $E[\epsilon_i]=0$. Then, using the validation data, 
we can obtain an estimate of the parameter $\beta$ as $\hat{\beta}=\frac{1}{n_{val}}\sum_{i=1}^{n_{val}} (U_i^{\text{proxy}} -U_i)$. Therefore in our SAR model, we proceed with the error-prone measurement $\tilde{U}_i = U_i^{\text{proxy}} -  \hat{\beta}$. We note that $E[\tilde{U}_i] = \beta + U_i -\beta=0$.
The covariance matrix can be estimated from the sample covariance matrix of $\tilde{U}_i - U_i$ within the validation portion.

\subsection{Network homophily correction}
A second application comes from the relatively new literature on correcting for latent homophily while estimating peer influence in network-linked observational studies \citep{natha2022identifying,mcfowland2021estimating}. The authors in \cite{mcfowland2021estimating,natha2022identifying} put forth methods for estimating peer influence in longitudinal network-linked studies correcting for latent homophily when an outcome is measured on individuals in at least 2 time points. For this purpose, \cite{natha2022identifying} modeled the network with the random dot product graph model, which is a general latent variable model for network data. However, one needs to necessarily use estimated latent variables in place of true latent variables. The estimated latent variables can be thought of as erroneous versions of the true latent variables and hence will introduce measurement error bias. The authors then proposed to incorporate these latent variables in a linear regression model when longitudinal data is available. However, currently methodology does not exist for estimating peer influence in cross-sectional data adjusting for latent homophily. Our proposed method in this paper can be used for correcting such measurement error bias in the estimated latent homophily vector while applying in the SAR context. 

Consider a network equivalent of the SAR model in Equation \ref{eq2}
\begin{equation}
Y=\rho  L Y + U \beta + Z \gamma +  V,
\label{model1}
\end{equation}
where in addition to observed covariates $Z$, the model includes unobserved latent characteristics  $U$. We assume the network formation involves the same set of latent characteristics. In particular, we let $U$ to be the latent positions in the 
Random Dot Product Graph (RDPG) model \cite{athreya2017statistical,tang2018limit,rubin2017statistical,xie2023efficient}, such that $E[A] = U_{n \times d} U_{d \times n}^T$, with the constraint that $U_i^TU_j \in [0,1]$. We can think of this set of latent characteristics $U$ as latent homophily variables that affect both the outcome as well as network formation. Following the methodology in \cite{natha2022identifying}, we propose to estimate the latent variables $\hat{U}$ from the adjacency spectral embedding of the observed network adjacency matrix. From the RDPG model literature \cite{xie2023efficient}, we also obtain estimates of the covariance matrices of $\hat{U}$s, as 
\begin{equation*}
\hat{\Delta}_i = (\frac{1}{n}\sum_i \hat{U}_i\hat{U}_i^T)^{-1}(\sum_j \hat{U}_j^T\hat{U}_i(1-\hat{U}_j^T\hat{U}_i)\hat{U}_j \hat{U}_j^T)(\frac{1}{n}\sum_i \hat{U}_i\hat{U}_i^T)^{-1}.
%\label{Deltai}
\end{equation*}
With this covariance matrix of the error we apply the ME-QMLE to estimate the model in Equation \ref{model1}.

\section{Simulation Studies}
In this section, we perform simulation studies to assess the finite sample performance of the proposed methods.

\begin{figure}[!htbp] 
  \caption{Mean Bias (+/- 1.96*SE) in $\hat{\beta},\hat{\gamma}$}
  \centering
  \label{covariatebias}
  \begin{minipage}[b]{0.5\linewidth}
    \centering
    \includegraphics[width=\linewidth]{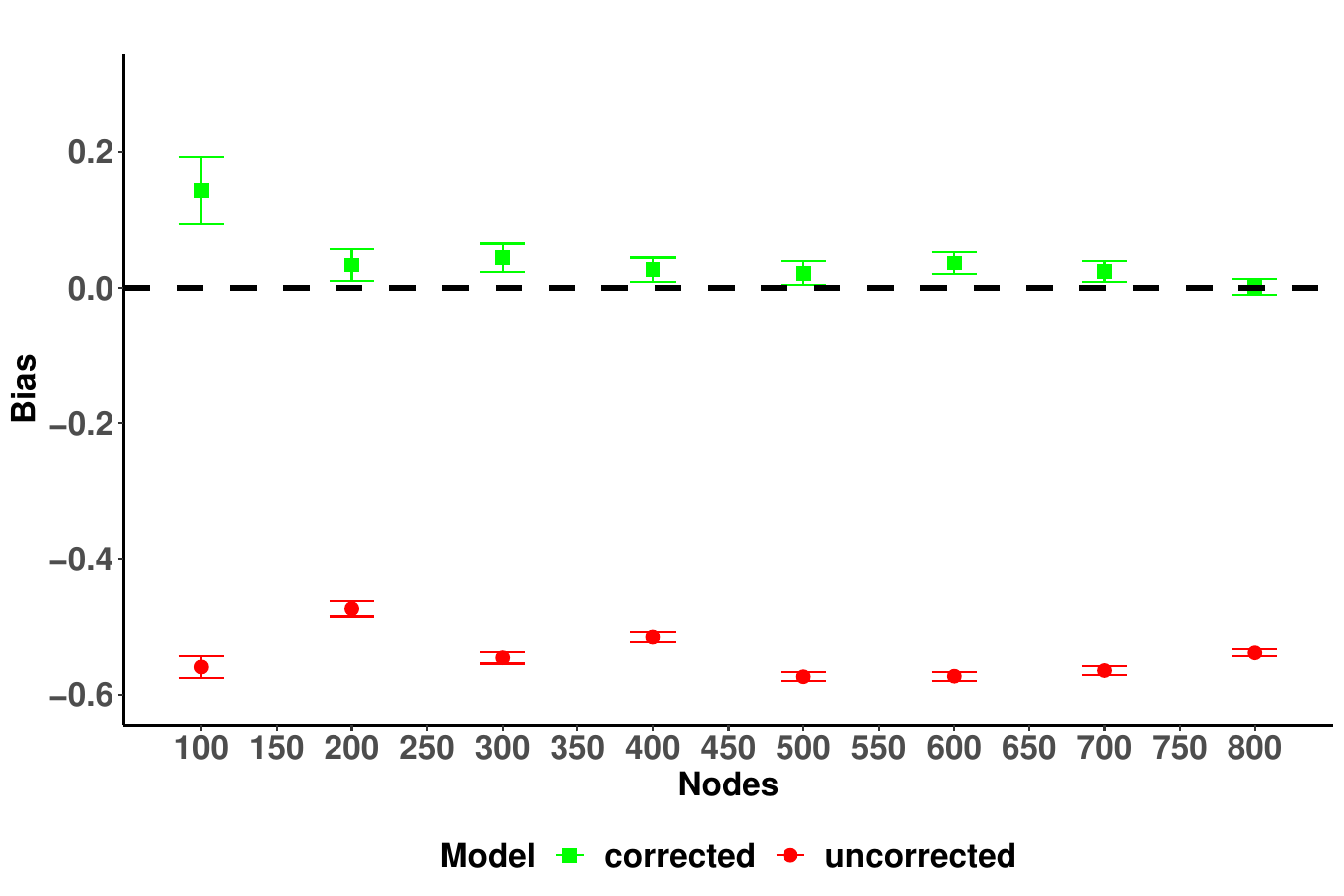} 
        \caption*{a. $\hat{\beta_{1}}$} 
  \end{minipage}%%
   \begin{minipage}[b]{0.5\linewidth}
    %\centering
    \includegraphics[width=\linewidth]{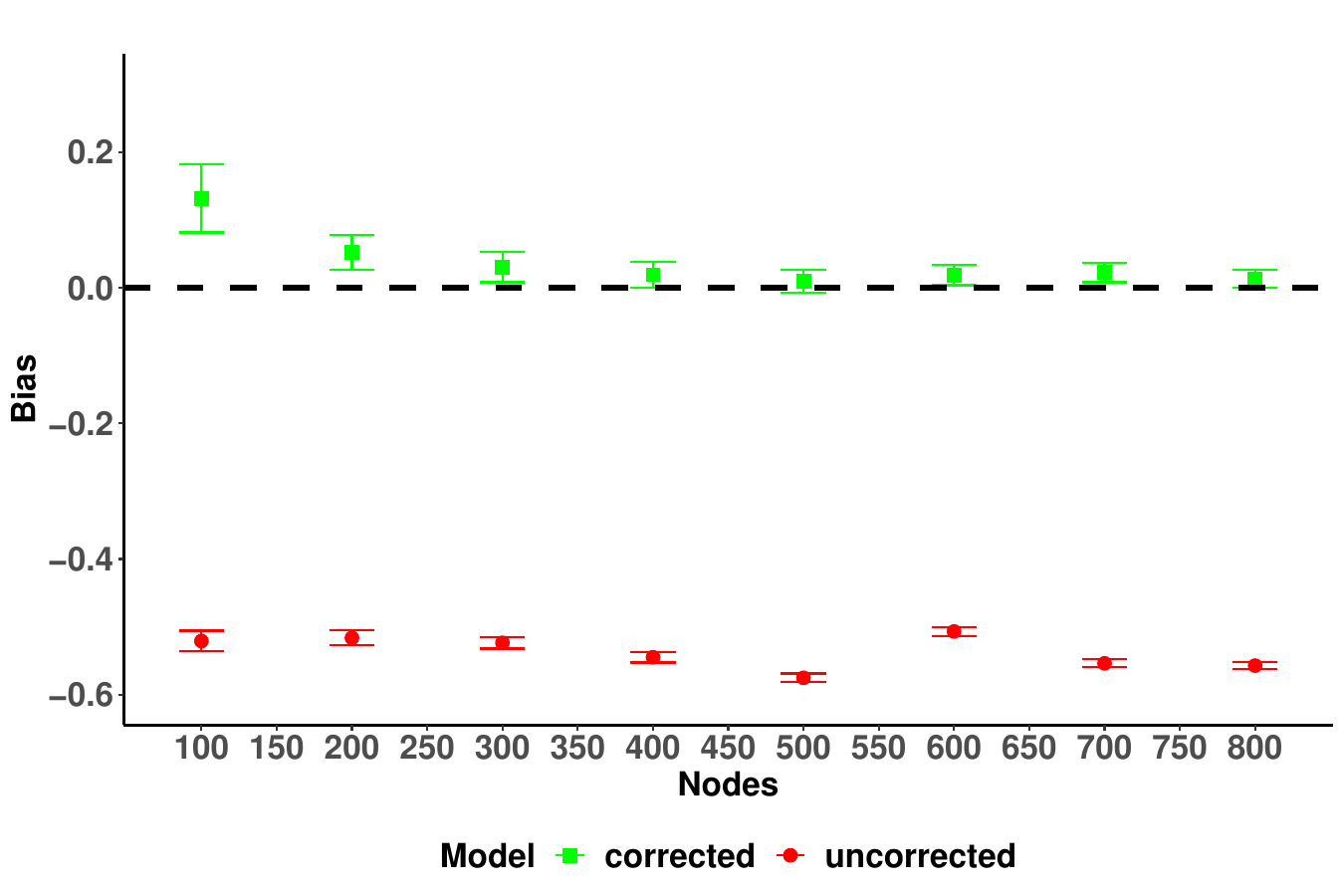}
    \caption*{b. $\hat{\beta_{2}}$} 
  \end{minipage}
  \begin{minipage}[b]{0.5\linewidth}
    \centering
    \includegraphics[width=\linewidth]{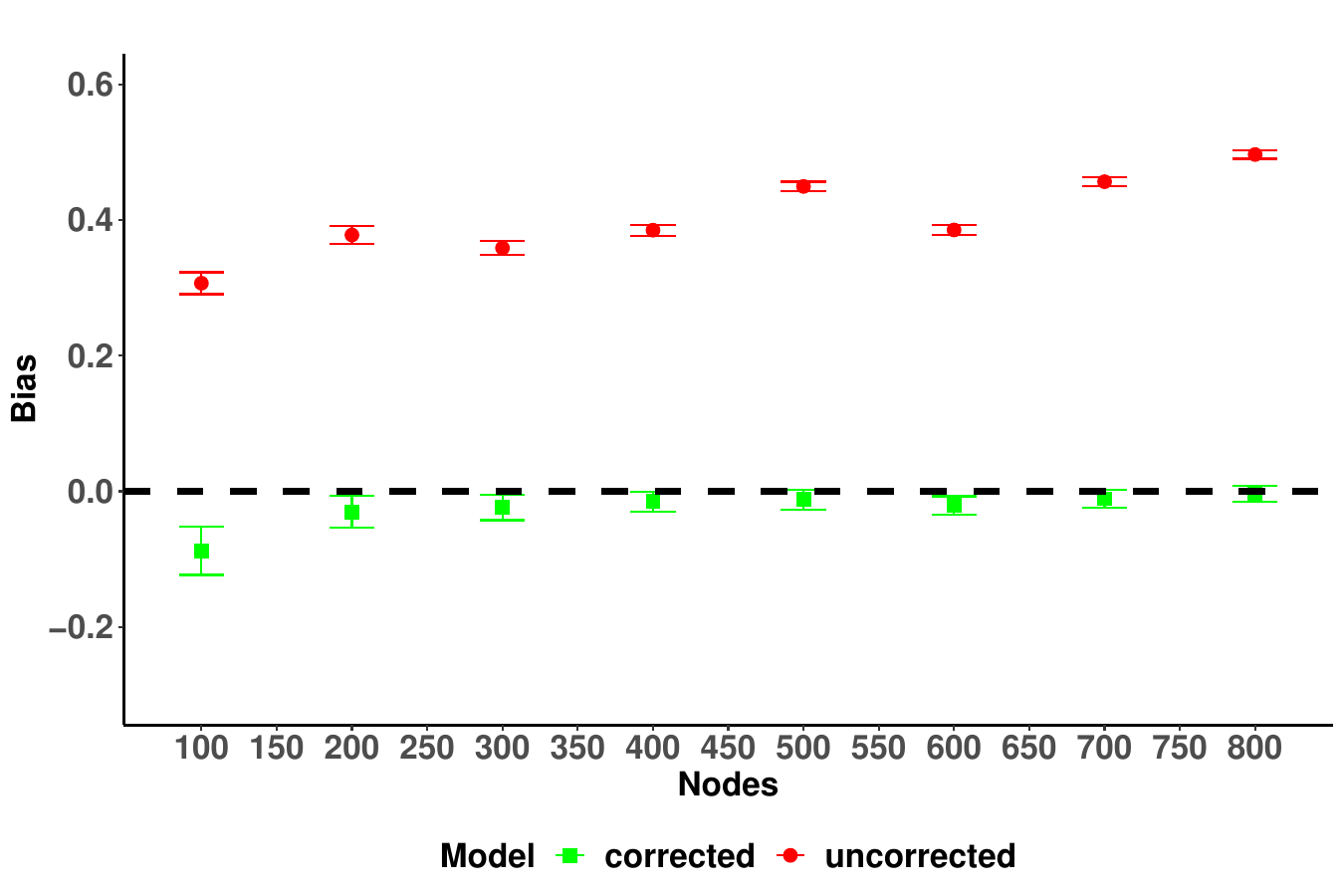} 
        \caption*{c. $\hat{\gamma_{1}}$} 
  \end{minipage}%%
   \begin{minipage}[b]{0.5\linewidth}
    %\centering
    \includegraphics[width=\linewidth]{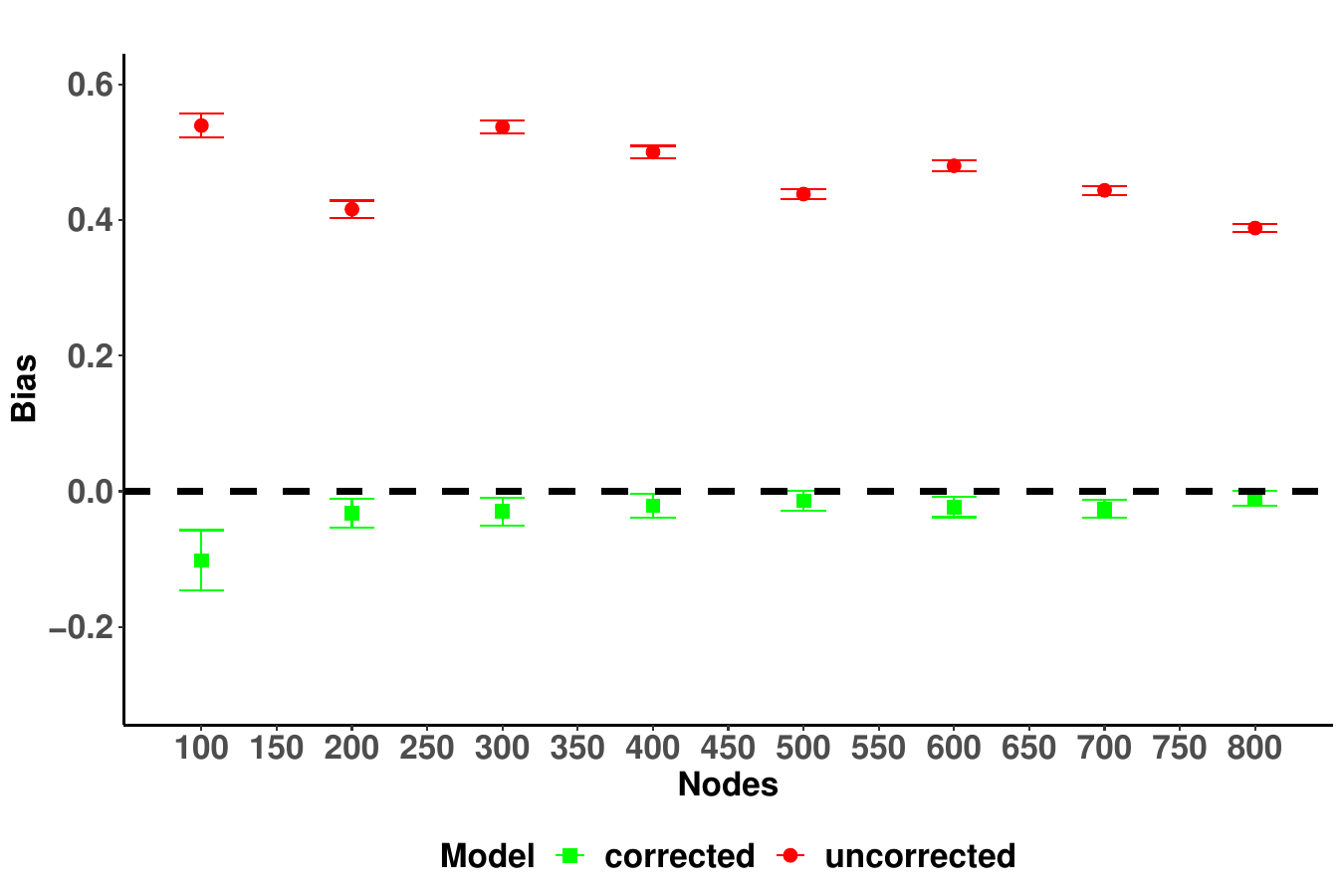}
    \caption*{d. $\hat{\gamma_{2}}$} 
  \end{minipage}
  \begin{minipage}{13.5 cm}{\footnotesize{Notes: These figures illustrate the average estimation bias and 1.96*SE for the mean bias over 300 simulations for the two error-prone covariates $\tilde{U}_1$ and $\tilde{U}_2$, and the two covariates measured without error $Z_1$ and $Z_2$.}}
\end{minipage} 
\end{figure}

\subsection{Covariate measurement error} We design a simulation study to illustrate how the proposed method can eliminate bias in parameter estimates due to measurement errors in some covariates. We consider 4 covariates $U=\{U_1, U_2\}, Z=\{Z_1, Z_2\}$ generated from a joint multivariate normal distribution $MVN(0,\Sigma_{X})$ with the covariance matrix $\Sigma_X$ having diagonal elements as $1.2$ and off-diagonal elements as $0.8$. In each replication, we generate an error-prone measurement of the first set of covariates $\tilde{U}$ by adding to $U$ an error generated from $MVN(0,\Sigma_{\xi})$ with $\Sigma_{\xi}$ having diagonal elements as $0.5$ and off-diagonal elements as $0.4$.

The observed network $A$ is generated from a stochastic block model with 4 communities with the $ 4\times 4$ block matrix of probabilities $\pi$ having $0.8$ as diagonal elements and $0.4$ as off-diagonal elements. The response $Y$ is then generated from a normal distribution according to the SAR model equation as follows:
\[
Y \sim N \left((I-\rho_0 L)^{-1}(U\beta_0 + Z \gamma_0), \sigma_0^2 (I-\rho_0 L)^{-1} (I-\rho_0 L)^{-1}\right).
\]
We set $\rho_0 = 0.4, \beta_0=(1,1), \gamma_0 = (1,1), \sigma_0 =1$. While the true data generating model uses the set of covariates $U,Z$, due to measurement error, we observe $\tilde{U},Z$ as our set of covariates. We vary the number of samples $n$ from $100$ to $800$ in increments of $100$ and perform 300 simulations for each $n$. We plot the average bias along with error bars representing 1.96*SE for average bias in estimating the coefficients associated with the two sets of covariates $U$ and $Z$ in Figure \ref{covariatebias}.
\begin{figure}[!htbp] 
  \caption{Mean Bias (+/- 1.96*SE) in $\hat{\beta},\hat{\gamma}$ varying the extent of error in $U$}
  \centering
  \label{incratio}
  \begin{minipage}[b]{0.5\linewidth}
    \centering
    \includegraphics[width=\linewidth]{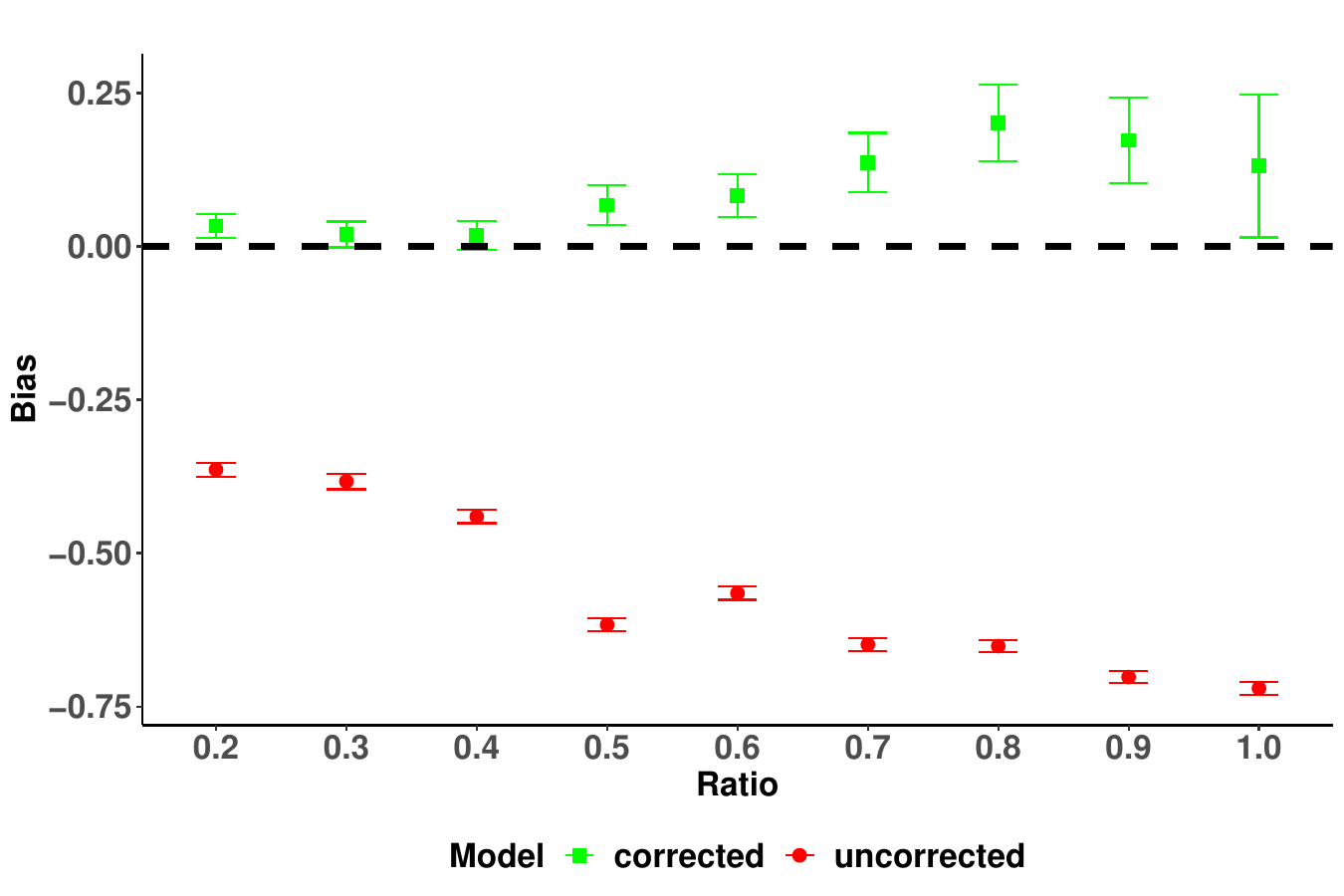} 
        \caption*{a. $\hat{\beta_{1}}$} 
  \end{minipage}%%
   \begin{minipage}[b]{0.5\linewidth}
    %\centering
    \includegraphics[width=\linewidth]{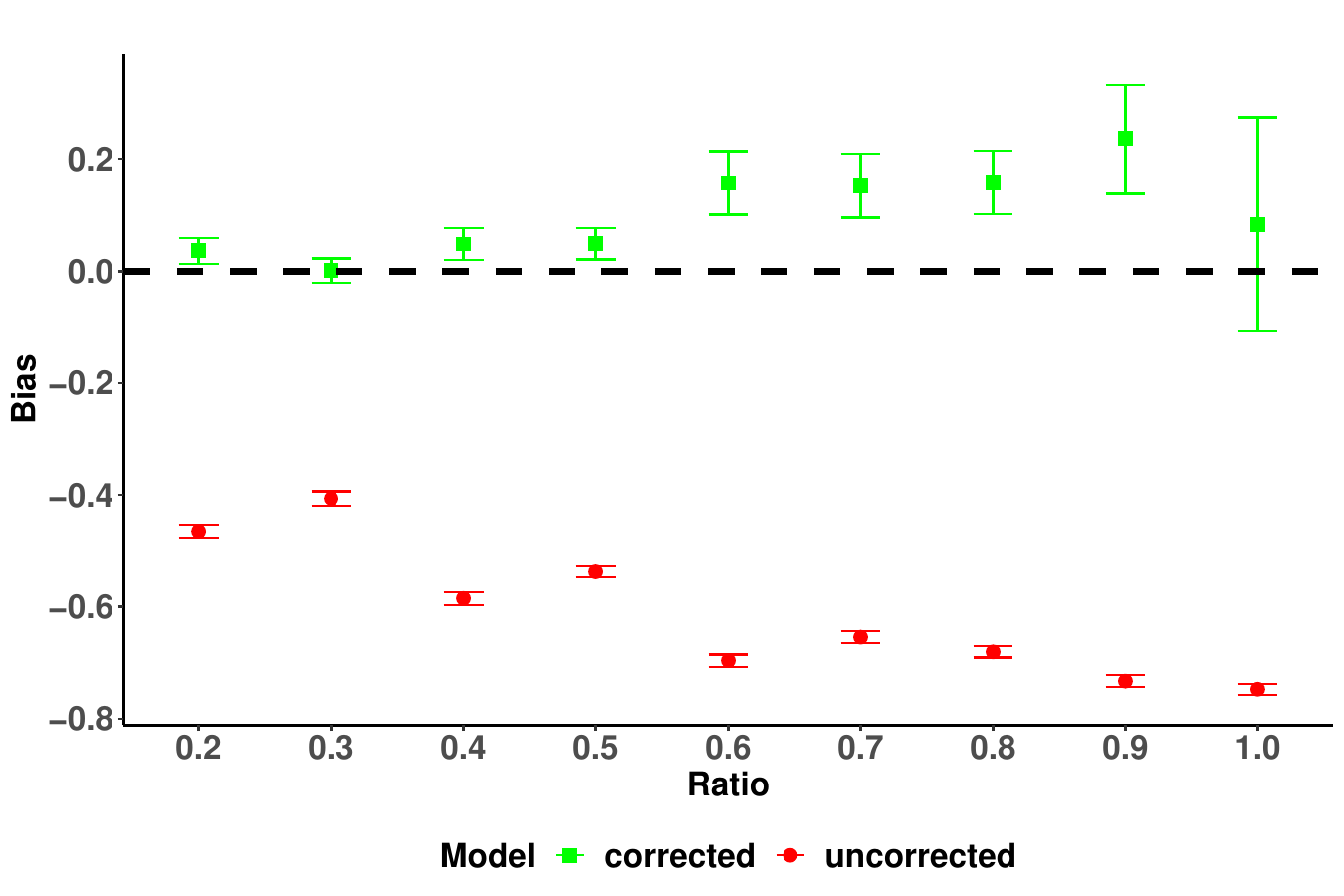}
    \caption*{b. $\hat{\beta_{2}}$} 
  \end{minipage}
  \begin{minipage}[b]{0.5\linewidth}
    \centering
    \includegraphics[width=\linewidth]{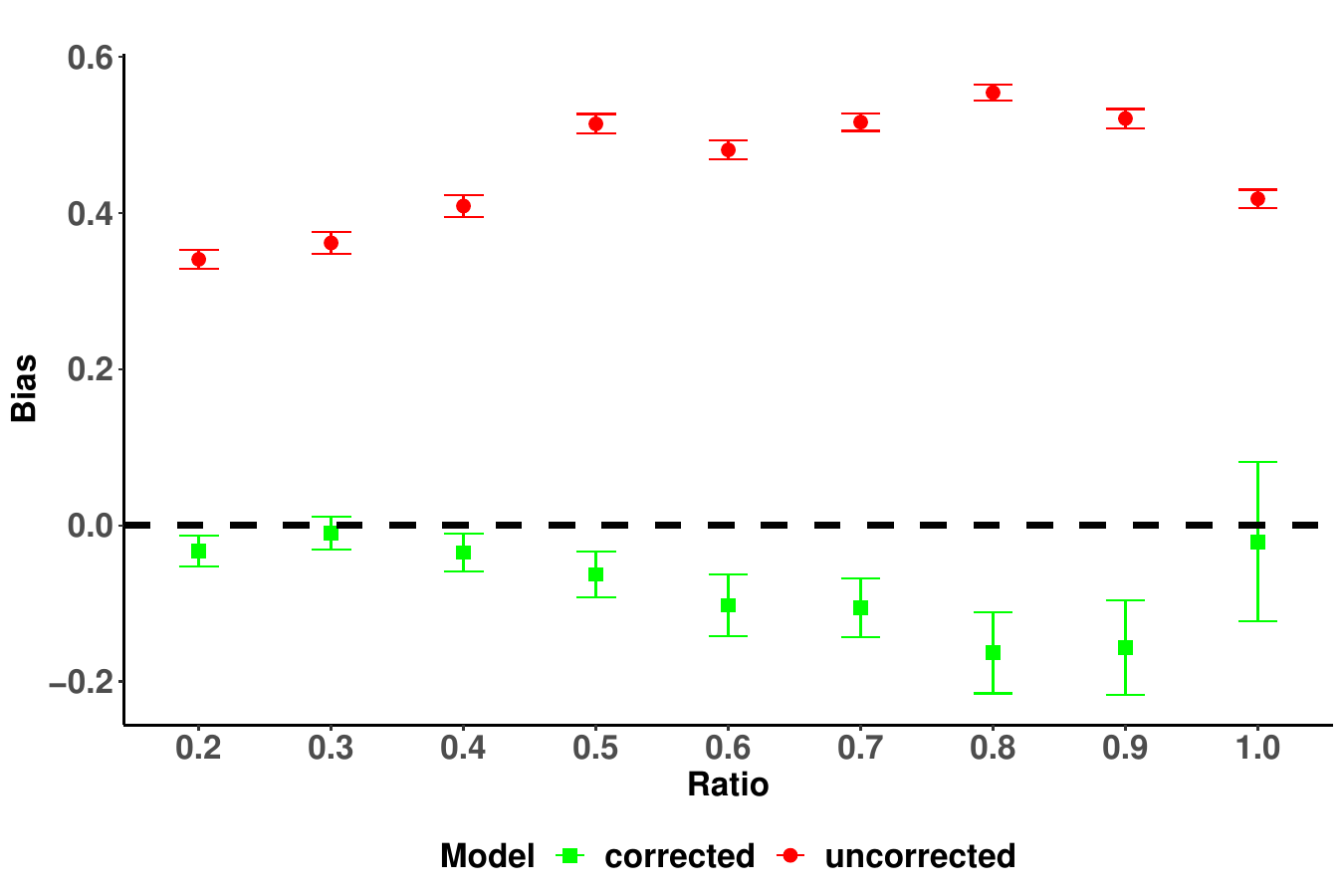} 
        \caption*{c. $\hat{\gamma_{1}}$} 
  \end{minipage}%%
   \begin{minipage}[b]{0.5\linewidth}
    %\centering
    \includegraphics[width=\linewidth]{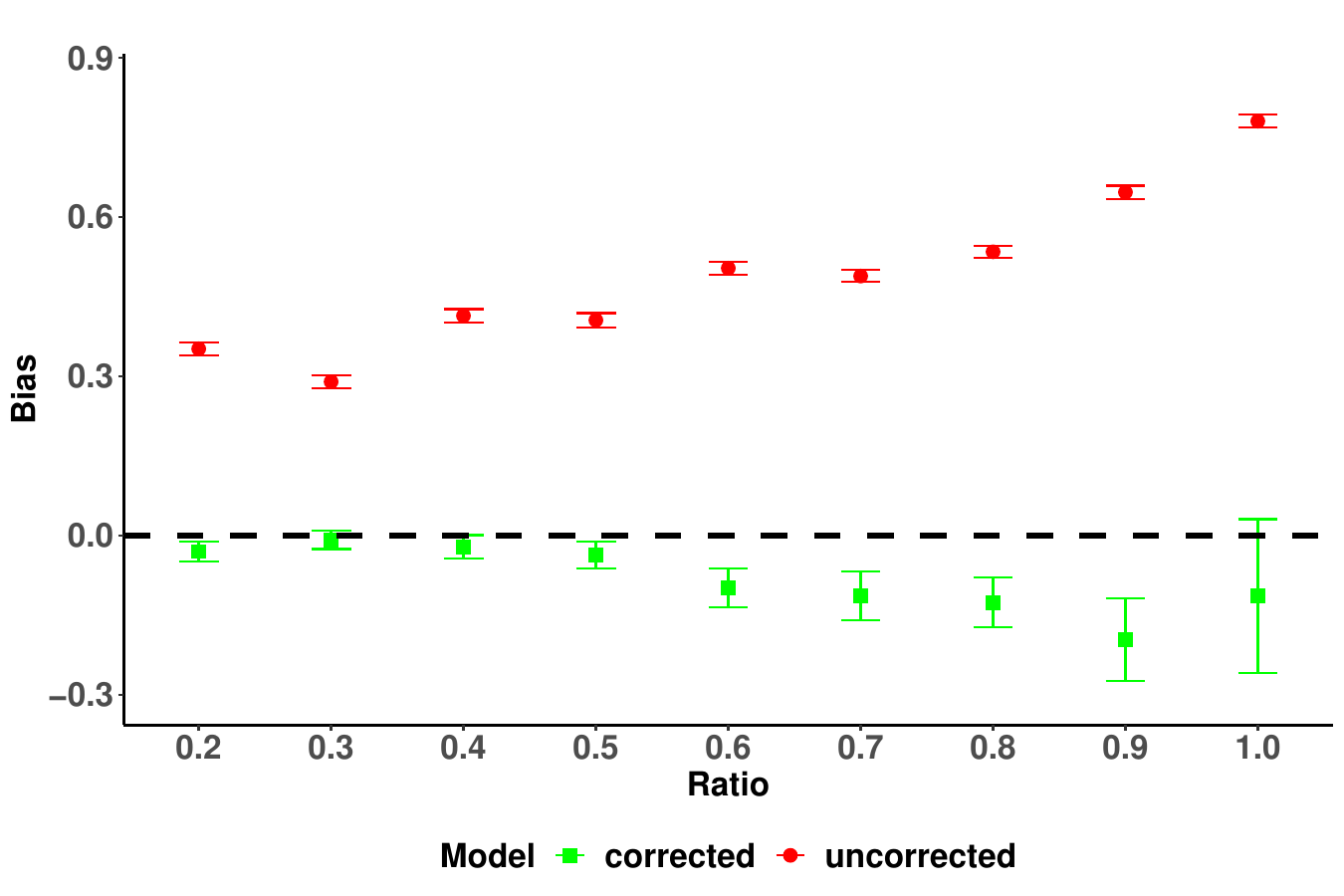}
    \caption*{d. $\hat{\gamma_{2}}$} 
  \end{minipage}
  \begin{minipage}{13.5 cm}{\footnotesize{Notes: These figures illustrate the average estimation bias and 1.96*SE for the mean bias over 300 simulations for the two error-prone covariates $\tilde{U}_1$ and $\tilde{U}_2$, and the two covariates measured without error $Z_1$ and $Z_2$. The number of nodes is kept fixed at 200 for all the simulation scenarios.}}
\end{minipage} 
\end{figure}
As expected from section 3.2, the parameter estimates for $U$ are negatively biased, and those for $Z$ are positively biased in the uncorrected method. The proposed correction methodology is successful in reducing bias rapidly and produces estimates with near zero bias, even for relatively small sample sizes.

The effectiveness of the proposed method is also dependent on the ``extent" of error in the measured proxy covariate from the true covariate, as seen in the asymptotic expansion in section 3.2.  We perform a second simulation to understand this phenomenon in finite samples. We set the number of nodes $n=200$. We keep $\Sigma_X$ and $\sigma_0$ as in the previous simulation. However, we make the matrix $\Sigma_{\xi} = \tau*\Sigma_0$, where $\Sigma_0$ has 1 in the diagonal and 0.8 in the off-diagonal, while $\tau$ is varied from $0.2$ to $1$ in increments of 0.1. We expect that for smaller values of $\tau$, since the extent of error in $\tilde{U}$ from $U$ is smaller relative to $\Sigma_X$, we will have a better performance of both the uncorrected and the corrected estimator. The performance of the estimator is presented in Figure \ref{incratio}. We see that the estimator performance in recovering the population parameter improves as we reduce the extent of the error in $\tilde{U}$. 
\begin{figure}[!htbp] 
  \caption{Comparison of Standard Error with Standard Deviation of Estimates}
  \centering
  \label{covariatese}
  \begin{minipage}[b]{0.5\linewidth}
    \centering
    \includegraphics[width=\linewidth]{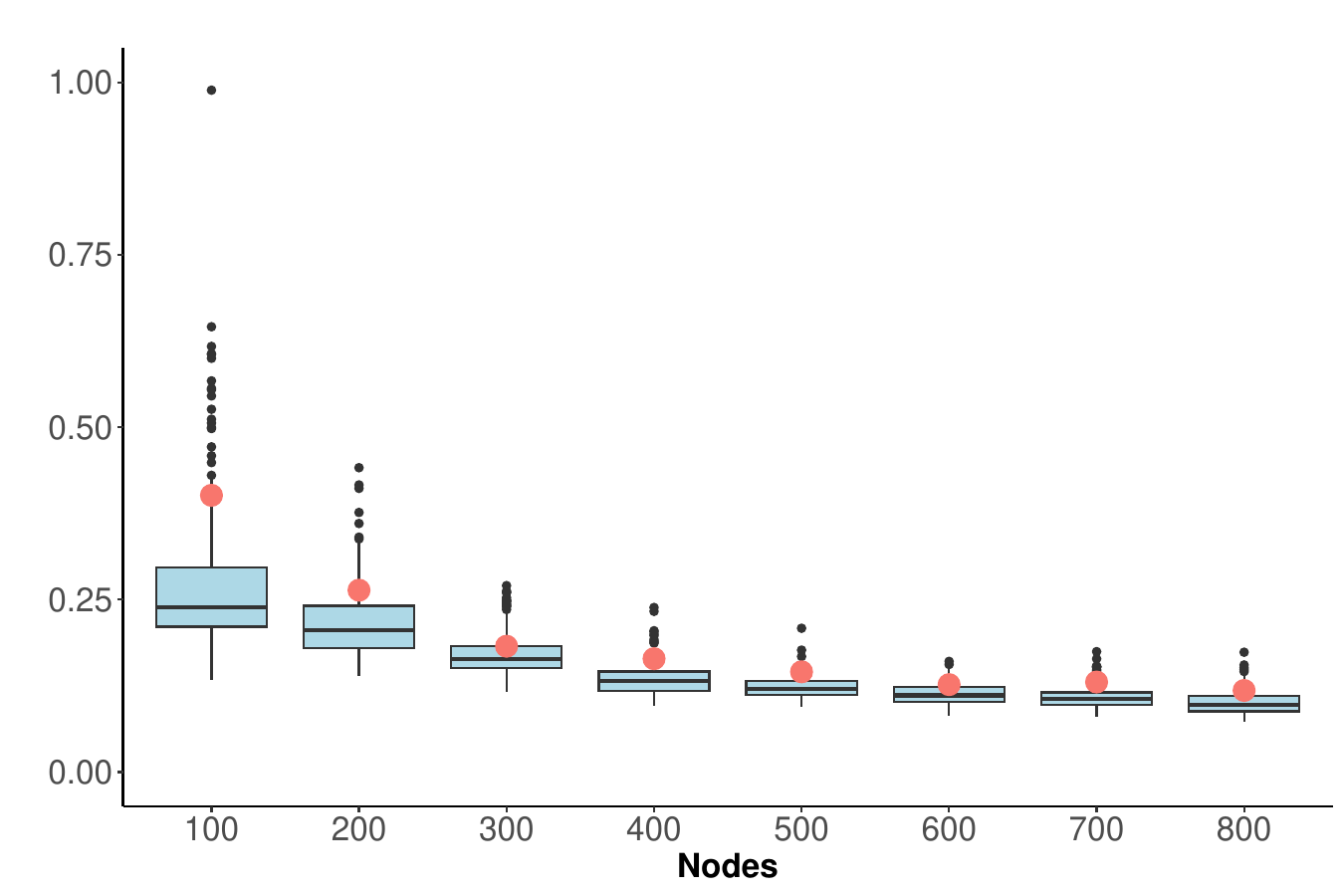} 
        \caption*{a. $\hat{\beta_{1}}$} 
  \end{minipage}%%
   \begin{minipage}[b]{0.5\linewidth}
    %\centering
    \includegraphics[width=\linewidth]{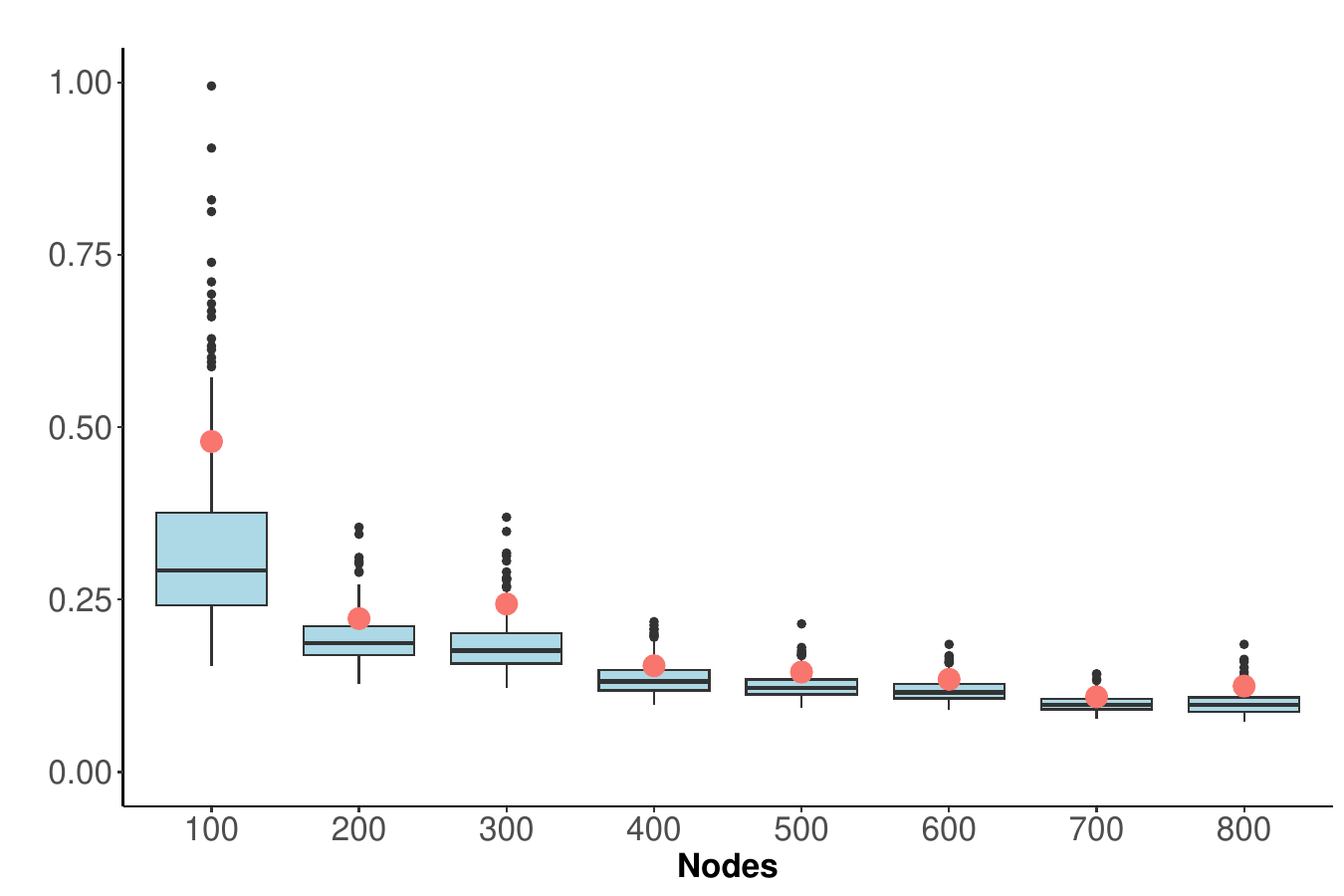}
    \caption*{b. $\hat{\beta_{2}}$} 
  \end{minipage}
  \begin{minipage}[b]{0.5\linewidth}
    \centering
    \includegraphics[width=\linewidth]{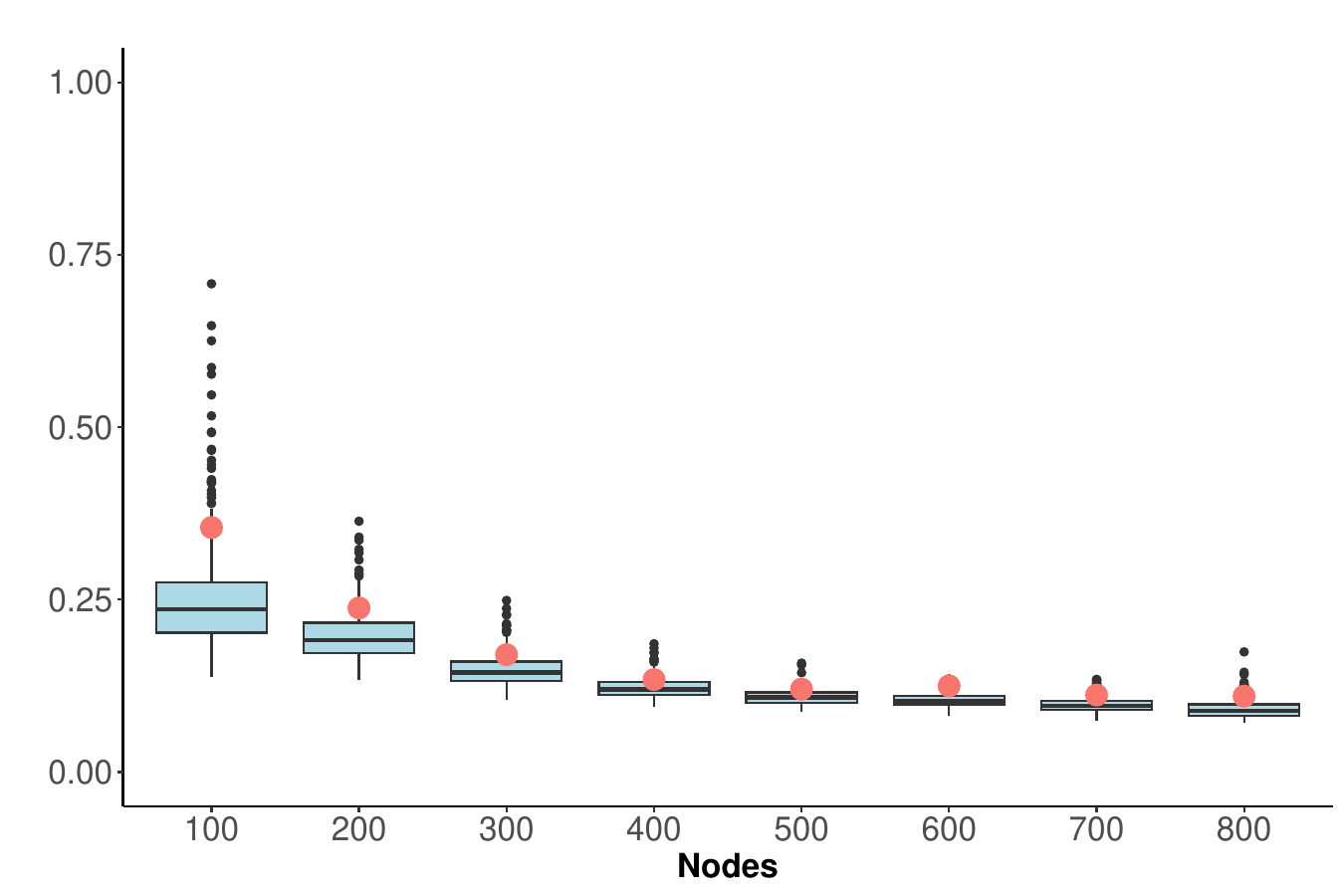} 
        \caption*{c. $\hat{\gamma_{1}}$} 
  \end{minipage}%%
   \begin{minipage}[b]{0.5\linewidth}
    %\centering
    \includegraphics[width=\linewidth]{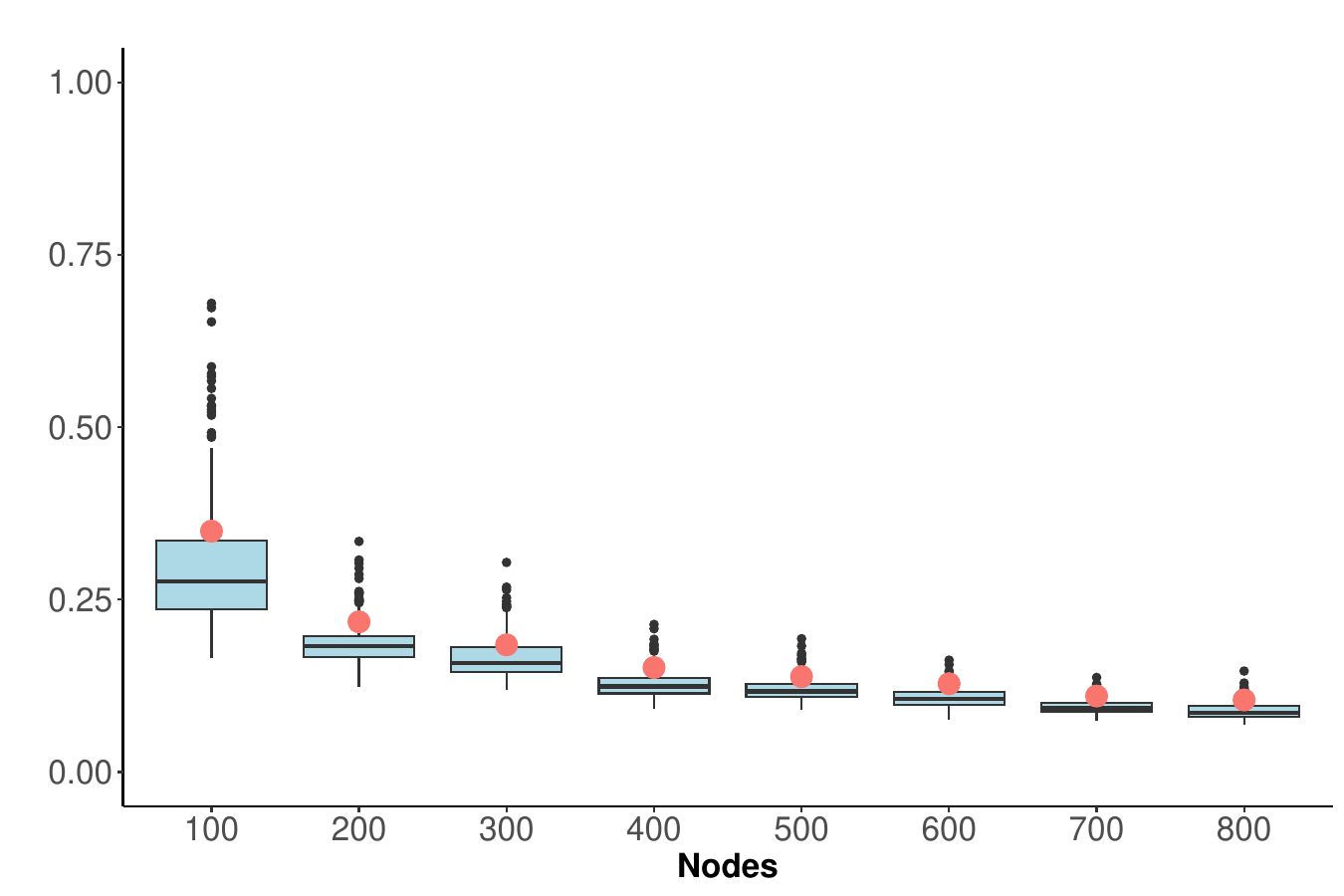}
    \caption*{d. $\hat{\gamma_{2}}$} 
  \end{minipage}
  \begin{minipage}{13.5 cm}{\footnotesize{Notes: These figures illustrate the comparison between standard error calculated using the close form solution derived in theorem \ref{asymptotic} with the empirical standard deviation (red points) of the parameter estimates over 300 simulations. }}
\end{minipage} 
\end{figure}
Finally, we verify the accuracy of our estimate of standard error from Theorem \ref{asymptotic}. For this purpose, Figure \ref{covariatese} displays box plots of the estimated standard errors over the 300 repetitions with the increasing nodes ($n$). We compare the boxplots with the estimated standard deviation of the estimates of parameters over the 300 repetitions (red dots). We notice that the estimated standard errors (SE) closely match the estimated sampling standard deviations with increasing $n$, validating the result in Theorem \ref{asymptotic}. On the other hand, Figure \ref{incratiose} shows this comparison as we increase the error in $U$. We see, especially for the cases with less error in $U$, that standard error estimates are close to the estimated sampling standard deviation. Comparing figures in Panel (a) and (b) shows that the estimated standard error better estimates the standard deviation of the estimates as we increase nodes from 200 to 500, similar to the pattern in Figure \ref{covariatese}. We provide these comparisons for $\beta_{1},\beta_{2}$ and $\gamma_{2}$ in the figure \ref{incratioseappendix} in the Appendix.

\begin{figure}[!htbp] 
  \caption{Comparing the standard error with standard deviation of the estimates}
  \centering
  \label{incratiose}
   \begin{minipage}[b]{0.5\linewidth}
    %\centering
    \includegraphics[width=\linewidth]{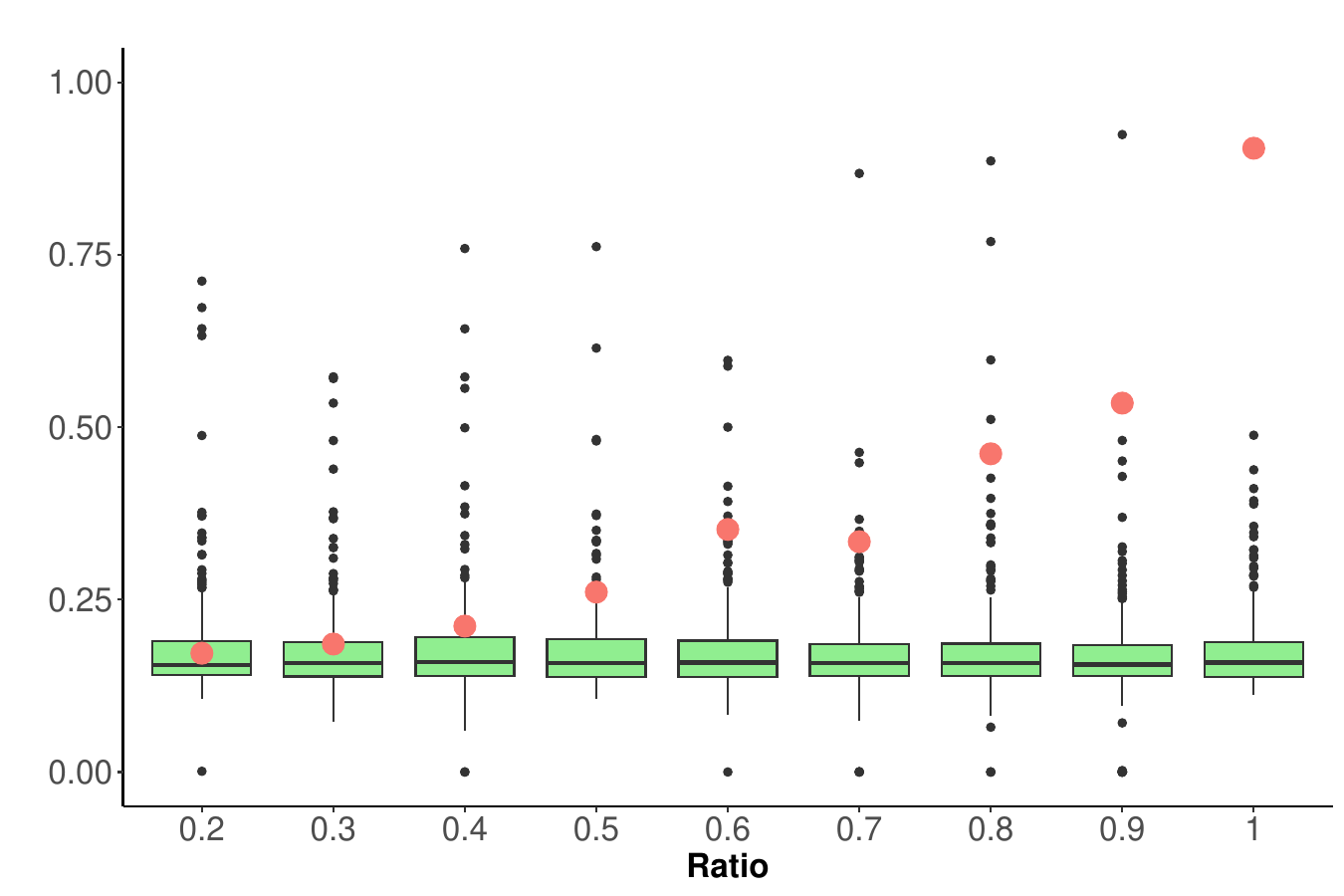}
    \caption*{a. $\hat{\gamma_{1}}$} 
  \end{minipage}%%
     \begin{minipage}[b]{0.5\linewidth}
    %\centering
    \includegraphics[width=\linewidth]{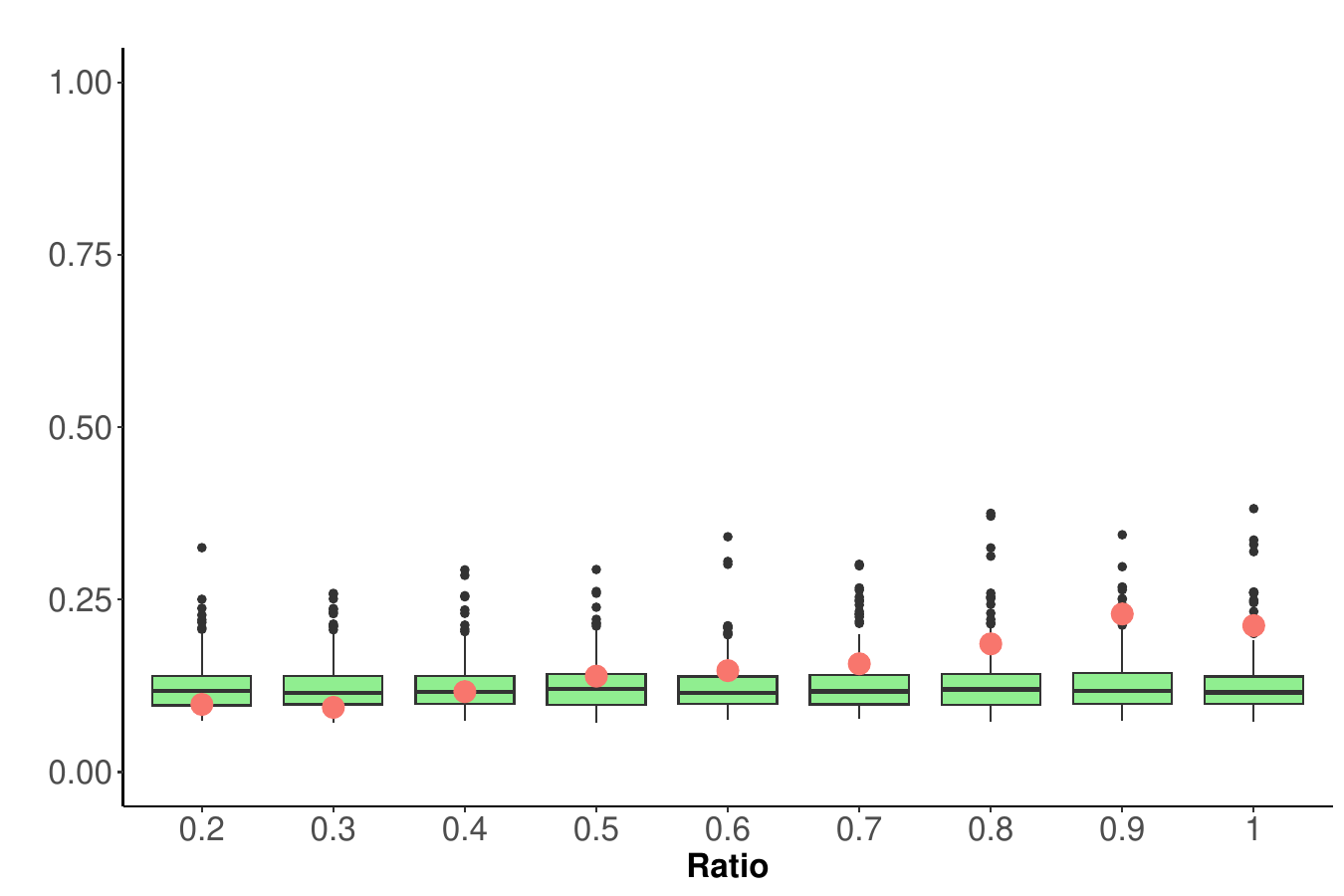}
    \caption*{b. $\hat{\gamma_{1}}$} 
  \end{minipage}
  \begin{minipage}{13.5 cm}{\footnotesize{Notes: These figures illustrate the standard error using the closed form solution provided in Theorem \ref{asymptotic} over 300 simulations. The number of nodes is kept at 200 for all the simulation scenarios in Panel (a). The distribution of standard errors is compared to the empirical standard deviation of the estimates highlighted in red. The number of nodes is increased to 500 for Panel (b) figure.}}
\end{minipage} 
\end{figure}

\subsection{Network homophily correction} We also design a simulation study to verify the performance of the proposed bias correction method along with latent homophily adjustment in terms of providing unbiased estimates of the network influence parameter. We generate the networks from a Stochastic Block Model (SBM) with an increasing number of nodes $n=\{50,  75, 100, 125, 150, 200$, $250, 300, 400, 500, 600\}$, a fixed number of latent factors at $d=2$, and the number of communities $k=4$. The matrix of probabilities $P$ is generated as $P=UU^T$, where the matrix $U_{n \times 2}$ is generated such that it has only 4 unique rows (with $k \geq d$). The resulting block matrix of probabilities is $\begin{pmatrix}
  0.53 &  0.19 &  0.18 & 0.45\\
 0.19 &  0.37 &  0.14 &  0.35 \\
 0.18 &  0.14 & 0.08 & 0.20 \\
0.45 & 0.35  & 0.20 & 0.50 \\
\end{pmatrix}$. The network edges $A_{ij}$ are generated independently from Bernoulli distribution with parameters $P_{ij}$. The covariate matrix $Z$ is correlated with the matrix $U$, and response $Y$ is generated from a multivariate normal distribution as follows:
\[
Y \sim N \left((I-\rho_0 L)^{-1}(U\beta_0 + Z \gamma_0), \sigma_0^2 (I-\rho_0 L)^{-1} (I-\rho_0 L)^{-1}\right).
\]

\begin{figure}[!h] 
  \caption{Mean Bias (+/- 1.96*SE) in network influence parameter $\hat{\rho}$ and accuracy of standard error computation}
  \centering
  \label{fig:simulation11}
  \begin{minipage}[b]{0.5\linewidth}
    \centering
    \includegraphics[width=\linewidth]{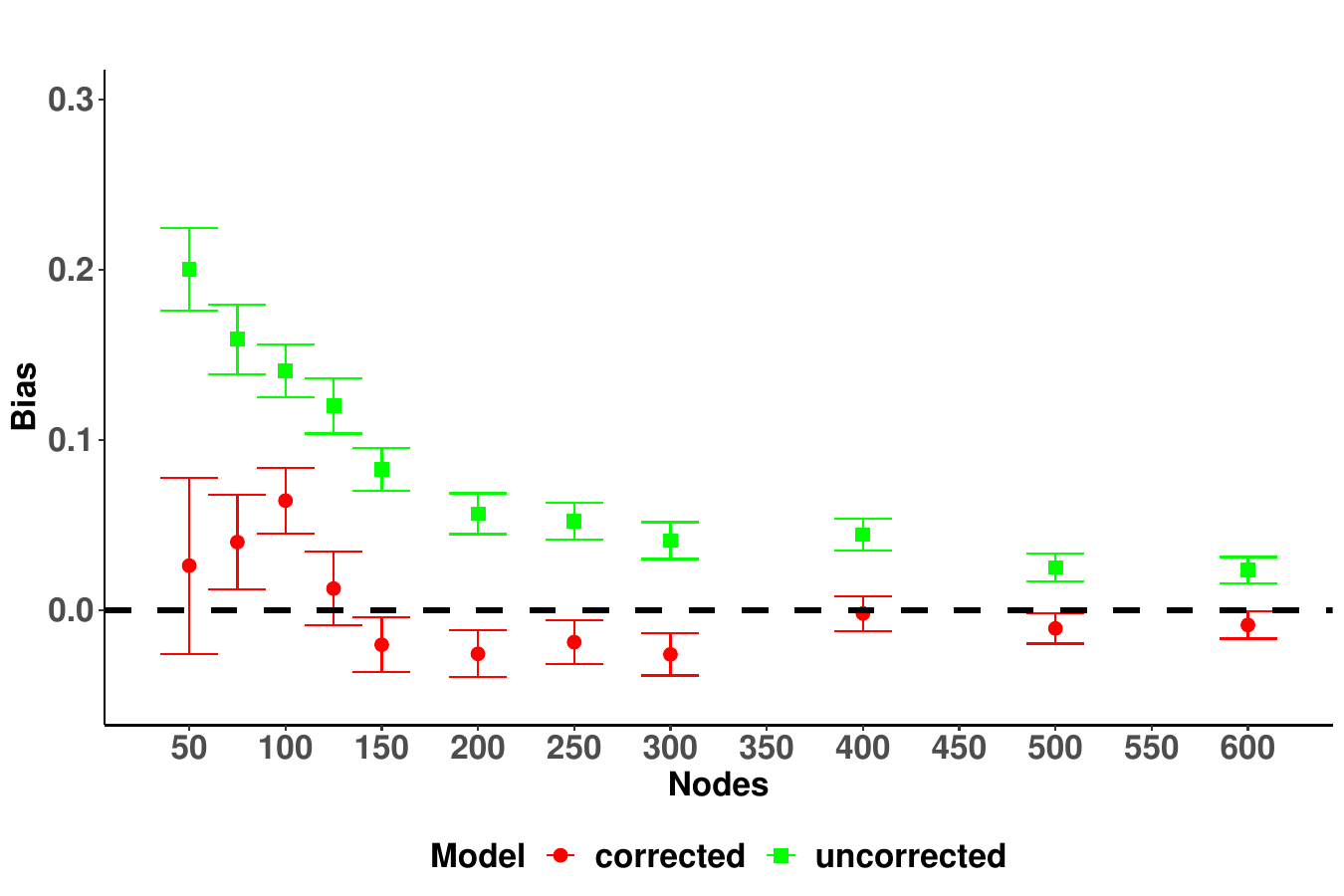} 
        \caption*{a. $\hat{\rho}$} 
  \end{minipage}
  \begin{minipage}{13.5 cm}{\footnotesize{Notes: Panel (a) shows a comparison of estimates of the network influence parameter $\rho$ from the model with latent factors but no bias correction (uncorrected), and the model with latent factors and bias correction (corrected). The points represent mean bias, and the error bars represent 1.96*SE for mean bias. }}
\end{minipage} 
\end{figure}

We set $\rho_0 = 0.4$, $\beta_0 = (1,2)$, $\sigma_0=0.8$ and $\gamma_0 = (0.2,-0.3)$. We compare two methods (a) latent factors estimated from SBM but without any bias correction (uncorrected model) and (b) estimated latent factor with bias correction (corrected model), in terms of bias of estimating the network influence parameter $\rho$. In the SAR context, $L$ is uncorrelated with $\tilde{X}$ that comprises $\tilde{U},Z$. However, in the network setting where the latent homophily factors are estimated from the network, they are highly correlated with $L$ and consequently with $LY$. Therefore, the extent of the bias in $\hat{\rho}$ is much more pronounced in the network homophily setting than in the SAR covariate setting. Here, we illustrate the extent of improvement our method provides over the uncorrected estimator for $\rho$. 

Panel (a) in Figure \ref{fig:simulation11} shows a comparison of the estimates of the network influence parameter from the two methods. It illustrates that the proposed measurement error bias-corrected estimator for the network effect parameter $\rho$ has less bias than the estimator with homophily control but no bias correction,  especially in small samples. The bias in the parameter estimate goes close to 0 more quickly with the measurement error correction. Therefore, the proposed bias correction methodology works well in correctly estimating peer effects in the presence of network homophily.

\begin{figure}[!h] 
  \caption{Mean Bias (+/- 1.96*SE) in $\hat{\beta},\hat{\gamma},\hat{\rho}$}
  \centering
  \label{3slssims}
  \begin{minipage}[b]{0.5\linewidth}
    \centering
    \includegraphics[width=\linewidth]{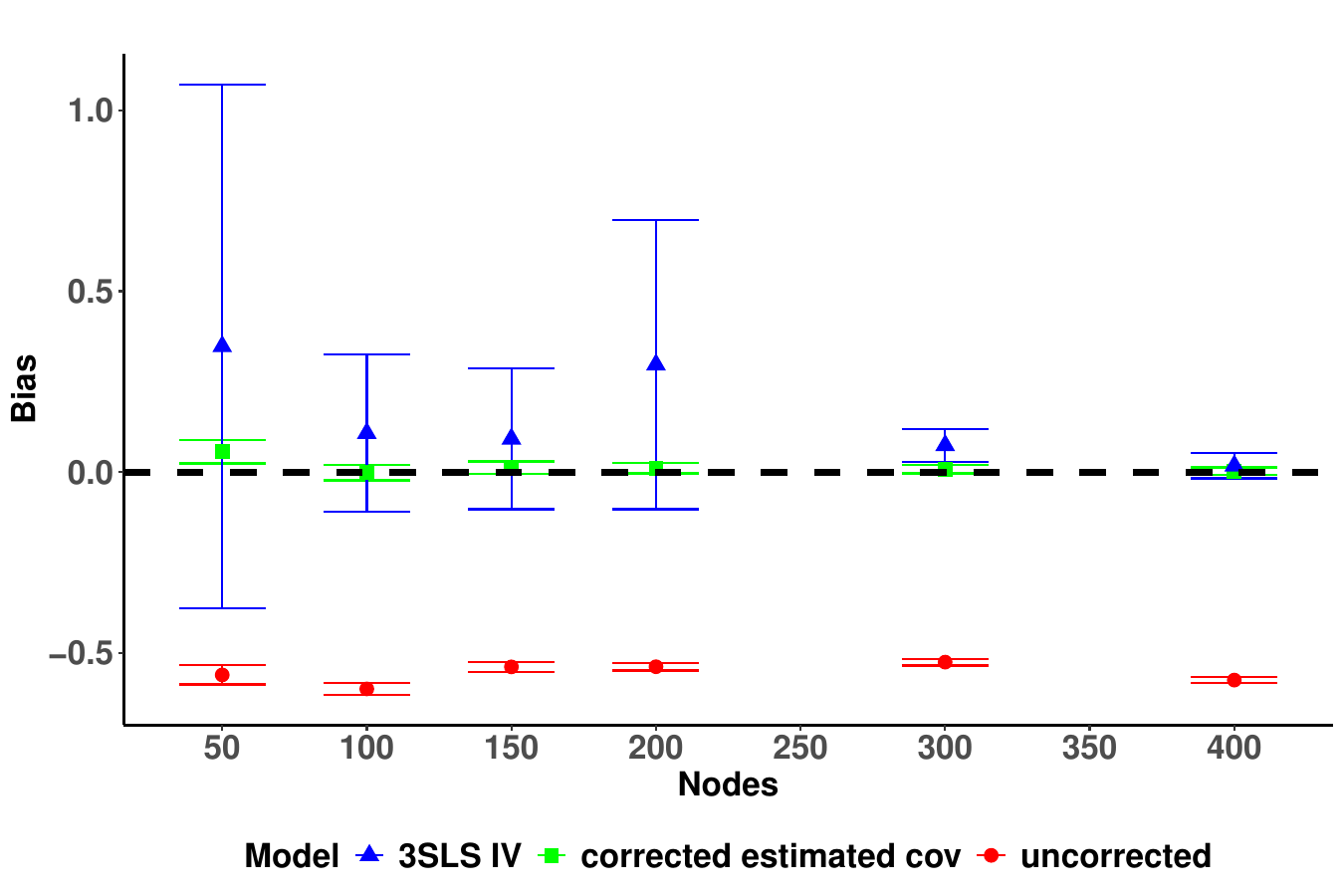} 
        \caption*{a. $\hat{\beta_{1}}$} 
  \end{minipage}%%
   \begin{minipage}[b]{0.5\linewidth}
    %\centering
    \includegraphics[width=\linewidth]{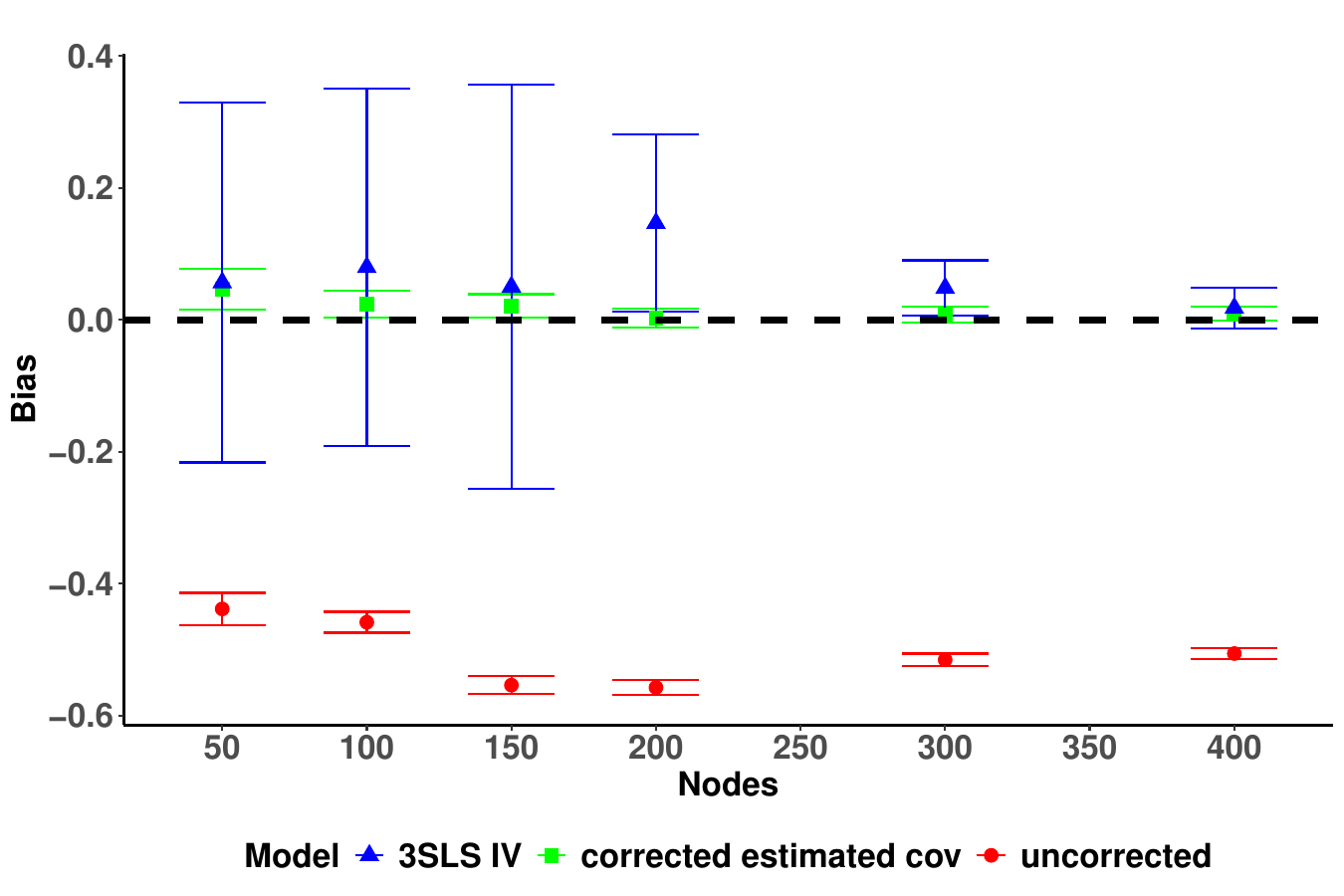}
    \caption*{b. $\hat{\beta_{2}}$} 
  \end{minipage}
  \begin{minipage}[b]{0.5\linewidth}
    \centering
    \includegraphics[width=\linewidth]{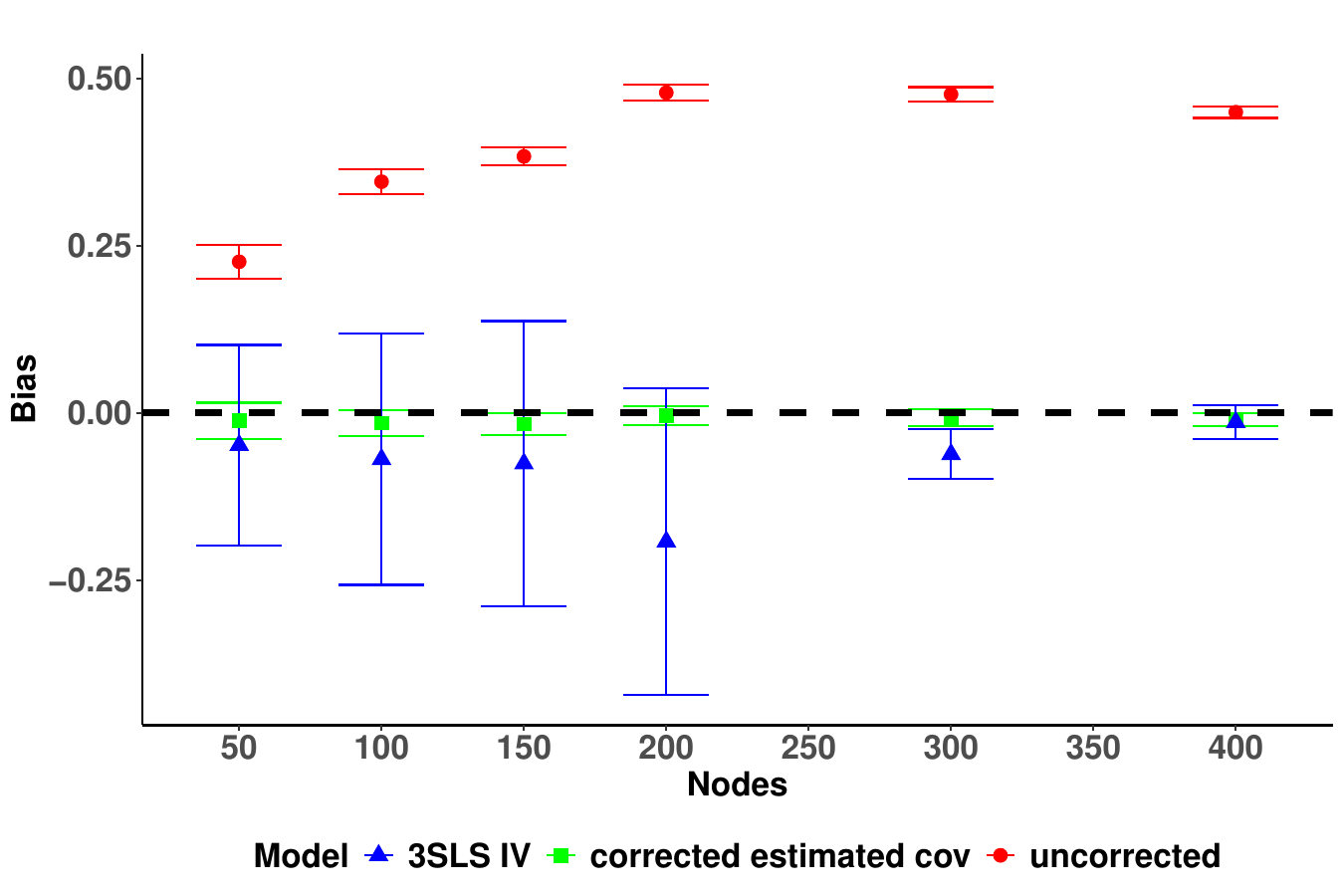} 
        \caption*{c. $\hat{\gamma_{1}}$} 
  \end{minipage}%%
   \begin{minipage}[b]{0.5\linewidth}
    %\centering
    \includegraphics[width=\linewidth]{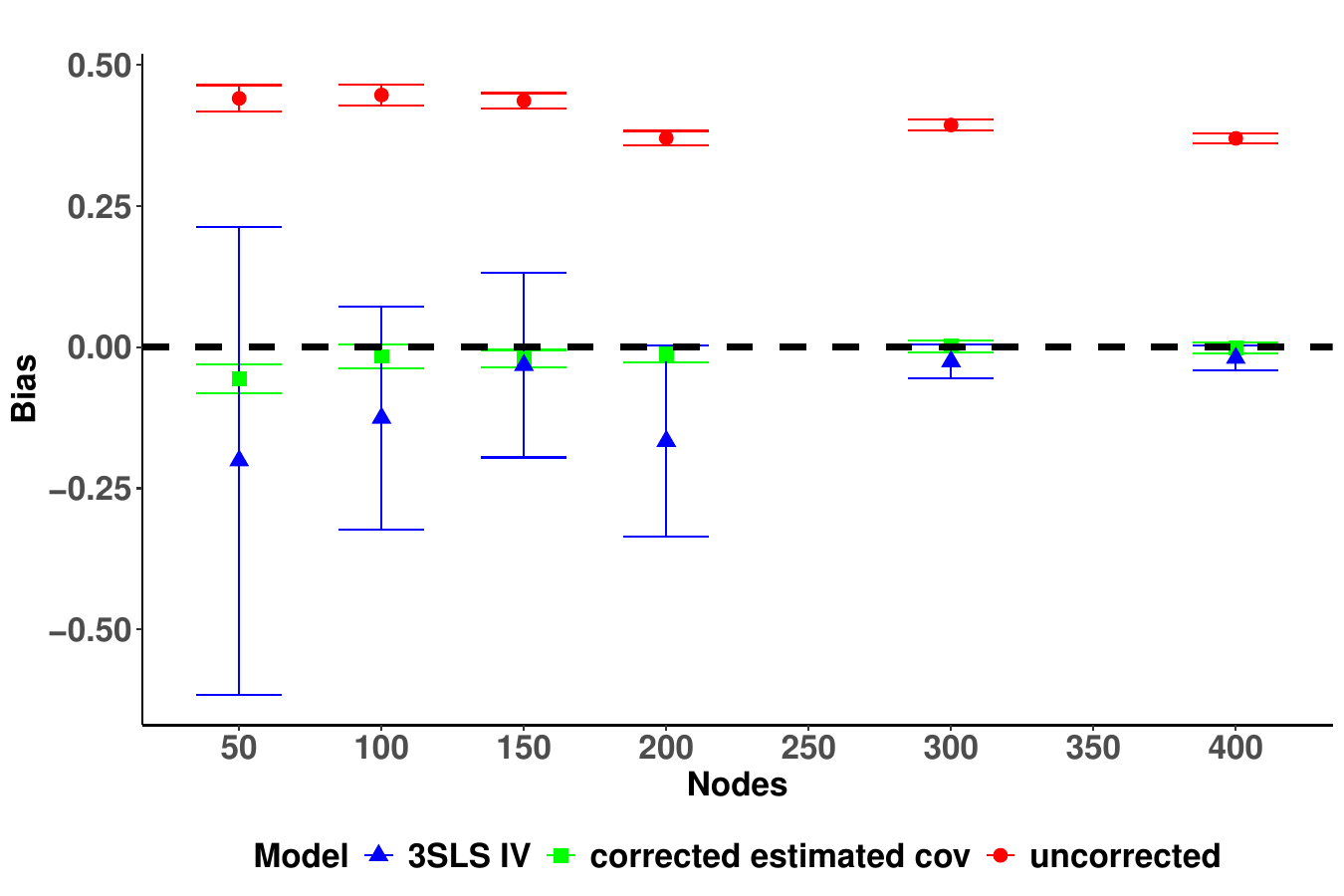}
    \caption*{d. $\hat{\gamma_{2}}$} 
  \end{minipage}
  \begin{center}
      \begin{minipage}[b]{0.5\linewidth}
    %\centering
    \includegraphics[width=\linewidth]{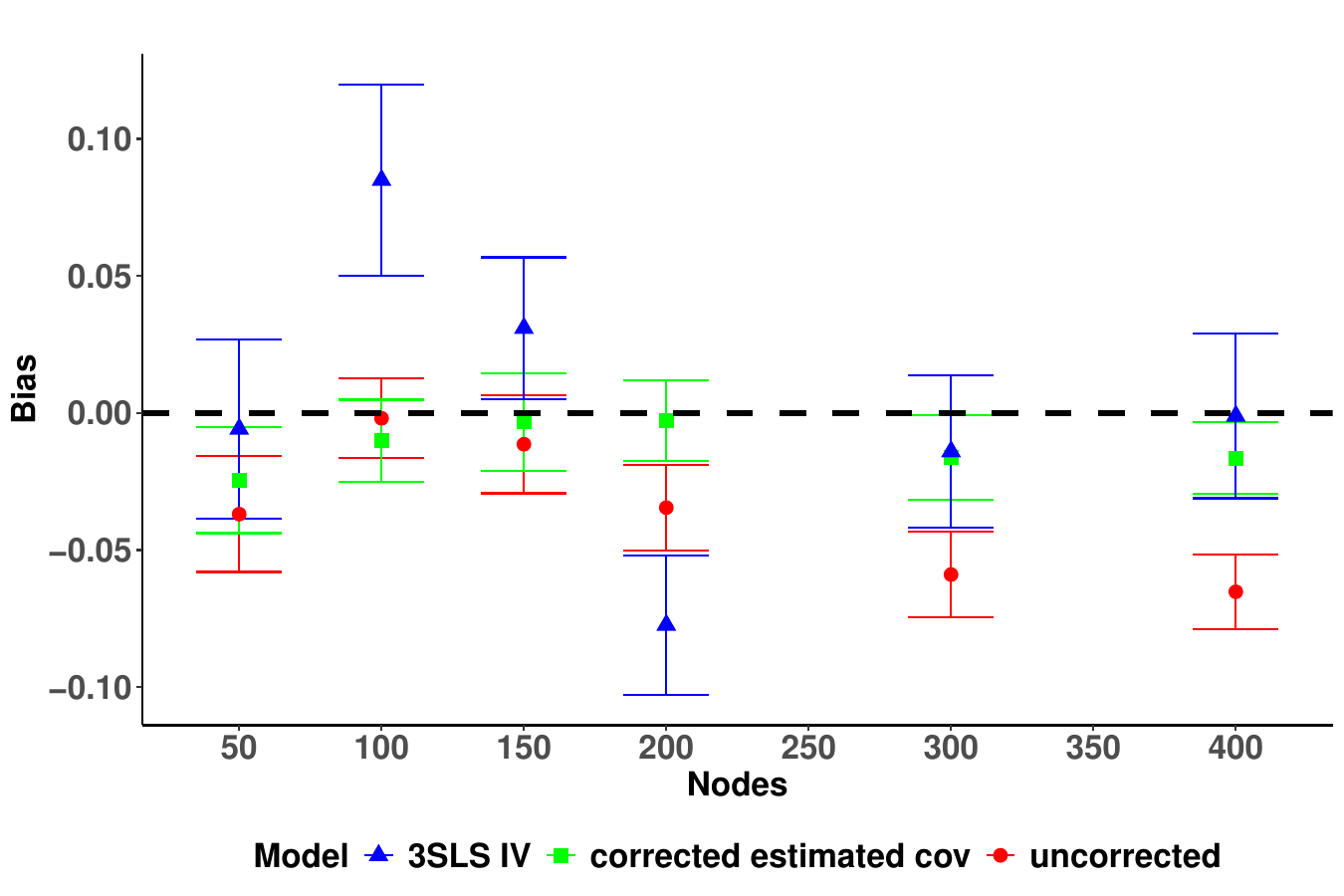}
    \caption*{e. $\hat{\rho}$} 
  \end{minipage}
  \end{center} 
  \begin{minipage}{13.5 cm}{\footnotesize{Notes: These figures compare the average estimation bias and 1.96*SE for the mean bias over 200 simulations for ME-QMLE in SAR model with estimated error covariance matrix and, uncorrected MLE of SAR, and the 3SLS instrumental variable method for SAR. }}
\end{minipage} 
\end{figure}
\subsection{Comparison with IV approach and estimated error covariance}
Next, we compare the proposed ME-QMLE with the 3SLS instrumental variables method proposed in \citep{luo2022estimation}. For the ME-QMLE we use the estimated covariance matrix of error which is fully data driven similar to the instrumental variable method. The setup is the same as the first simulation study on covariate measurement error, except we observe 4 replicate measurements $\tilde{U}$ from which we can estimate the covariance of the error $\Sigma_{\xi}$. As stated in section 4.1, we use the sample mean as the error-prone variable and estimate the covariance matrix of error using the formulae stated there. Next we also generate an instrumental variable that is correlated with $U$. We compare the performance of ME-QMLE with 3SLS with increasing sample size in Figure \ref{3slssims}. As the sample size increases both ME-QMLE and 3SLS estimators are able to provide unbiased parameter estimates for all parameters, while the uncorrected MLE continues to have a large bias. We note that the 3SLS estimator has a very large variance and is often heavily biased in small samples, especially for $n=50$ and $n=100$, while the ME-QMLE continues to have low variance and a bias close to 0. Therefore the ME-QMLE might have some advantages over the 3SLS for small samples. However, we also note that the performance of ME-QMLE depends on the availability of replicate measurements, while that of 3SLS depends on the availability of good instruments.

\section{Empirical Applications}
\label{coviddata}

\subsection{Conflict and Network}

In this section, analyze a publicly available conflict and social network data collected for the study in \cite{paluck2016changing} using the proposed estimators. The data comes from a randomized controlled trial in 56 schools to study the impact of a student-led intervention on conflict. As part of this study, pre-intervention, the authors collected a social network survey where the students in 5th, 6th, 7th, and 8th grade were asked to name up to 10 friends they interact with within the school. Additionally, the authors collected information on GPA and a set of rich covariates, capturing the socio-demographic information of the students.

The study sample consists of 24191 students out of which approximately 50\% were allocated to treatment and the remaining were part of the control group. To keep our analysis from being impacted by the intervention, we use the sample of the students in the control group. Table \ref{summarytableconflict} provides summary statistics on the variables of interest. The final sample comprises 7872 students in the control group for whom we could observe non-missing values for all the covariates and the outcome variable of interest. Our goal is to estimate peer influence on GPA. The methods developed in this paper allow us to disentangle peer effects from latent homophily. The estimation proceeds in two steps. In the first step, we use spectral embeddings on the friendship network to estimate latent homophily vectors. In the second step, we adjust for unobserved homophily using the estimated latent homophily in the outcome model to estimate the peer influence parameter.

The GPA varies between 0 and 4, with a mean of 3.19 and a standard deviation of 0.605. Panel (a) in Figure \ref{fig:gpa} shows the distribution of GPA. Panel (b) shows the density plots for GPA separately for high ($\geq$ mean of peer GPA) and low peer GPA (below the mean of peer GPA). We observe that the distribution of GPA is shifted to the right for those who have high GPA friends. The mean GPA for the students who have low GPA peers (high GPA peers) is 3.10 (3.27). In addition to peer GPA, we also control for several covariates in the regression, which have substantial explanatory power for the outcome variable of interest. About 25\% of the final sample have mothers with no college education. Panel (b) in Figure \ref{fig:gpacov} shows the density plots for GPA separately by mother's education. We see that there is a higher density of the high GPA students whose mothers had a college education. The survey also collected information on the family structure. We report the differences in the distribution of GPA separately for the students who reported to be living with a single parent. About 12\% of the students reported to be living with just their mother and 1.6\% reported to be living only with their father. We find that students reported to be staying with one parent have lower percentiles for GPA.

\begin{figure}[!htbp]
  \centering
   \caption{GPA and peer GPA}
    \label{fig:gpa}
\begin{minipage}[b]{0.5\linewidth}
    \centering
    \includegraphics[width=\linewidth]{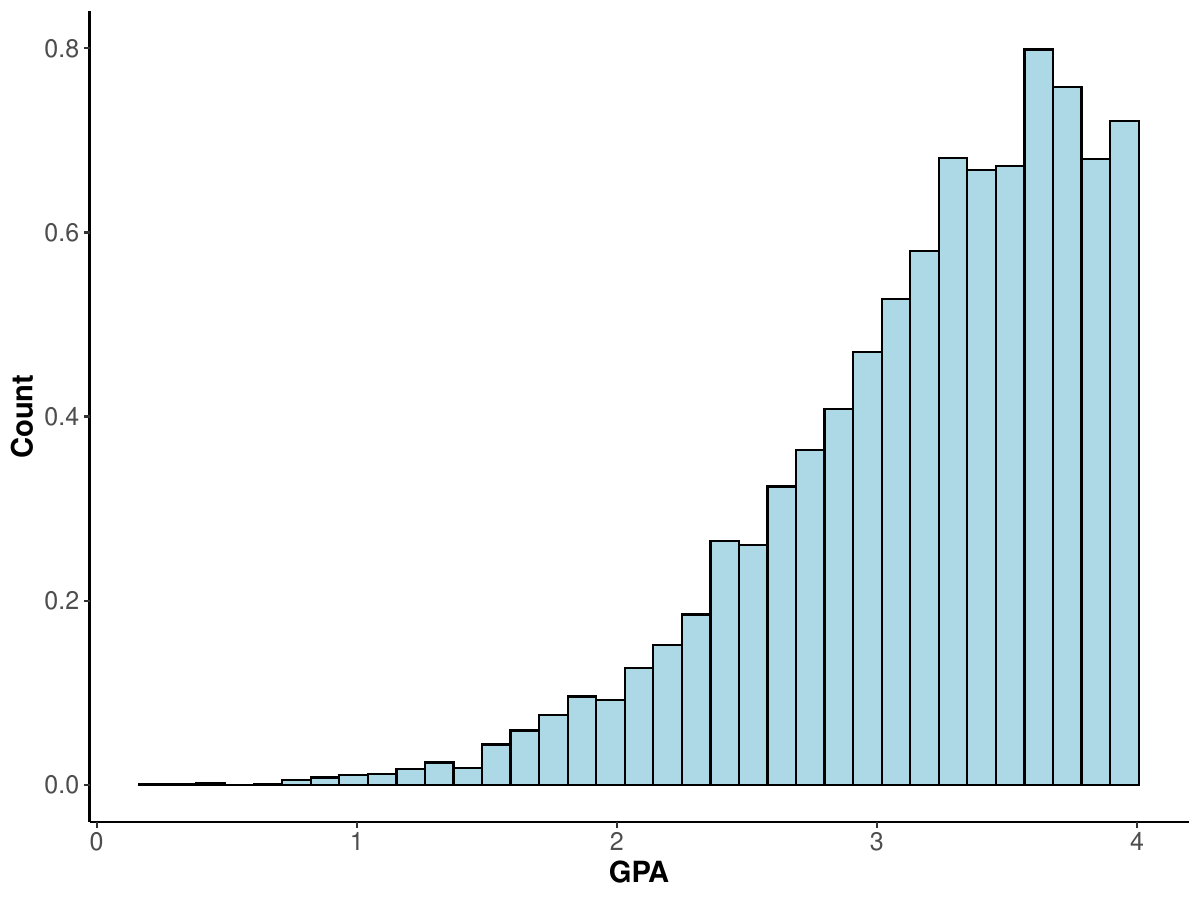} 
        \caption*{a. GPA} 
  \end{minipage}%%
   \begin{minipage}[b]{0.5\linewidth}
    %\centering
    \includegraphics[width=\linewidth]{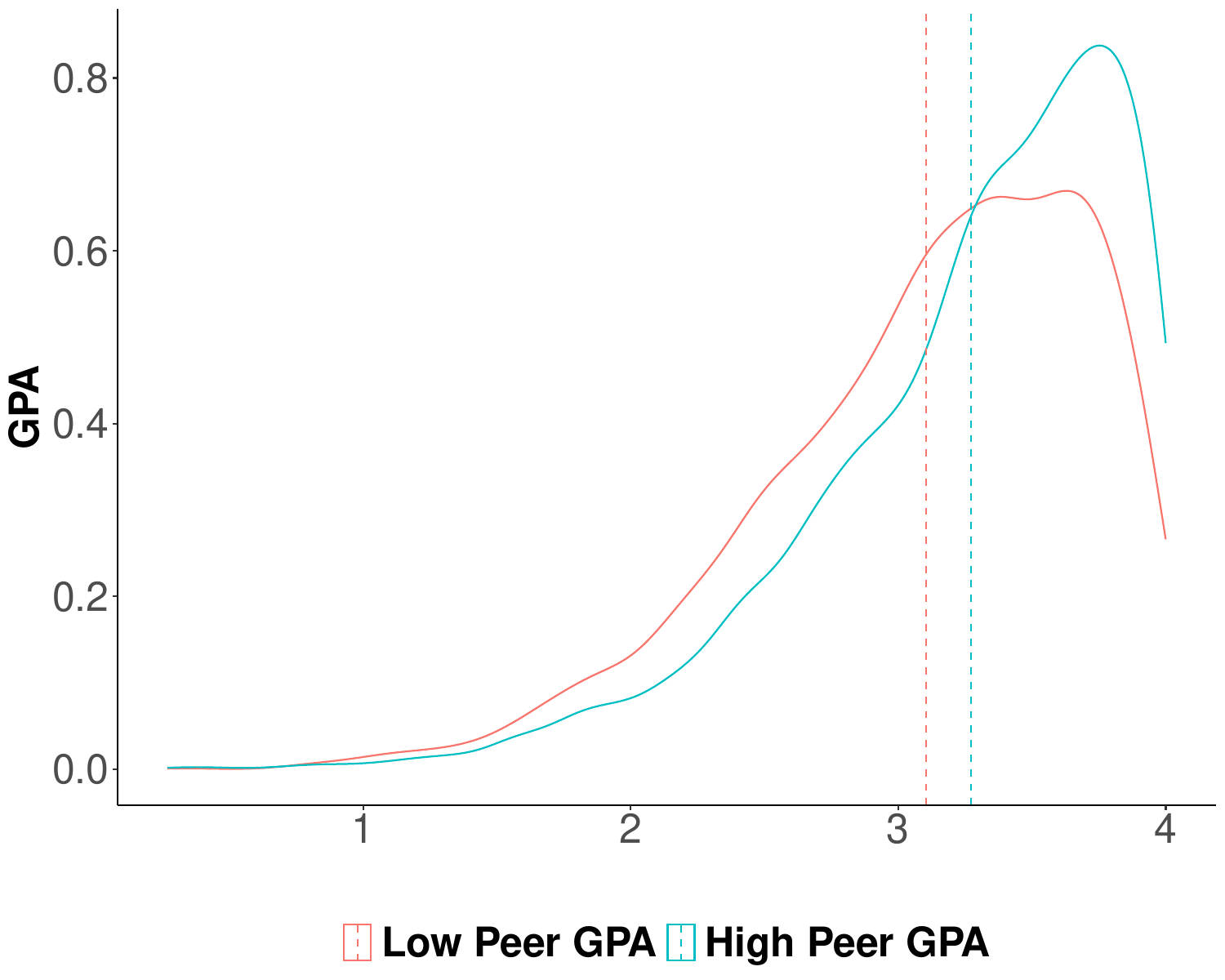}
    \caption*{b. Peer GPA} 
  \end{minipage}
 \begin{minipage}{13.5 cm}{\footnotesize{Notes: Panel (a) illustrates the distribution of GPA. Panel (b) shows the density plots for GPA separately for high and low peer GPA groups. The respective means are displayed by the dashed vertical lines.}}
\end{minipage} 
\end{figure}

\begin{figure}[!htbp]
  \centering
   \caption{GPA and other covariates}
    \label{fig:gpacov}
\begin{minipage}[b]{0.5\linewidth}
    \centering
    \includegraphics[width=\linewidth]{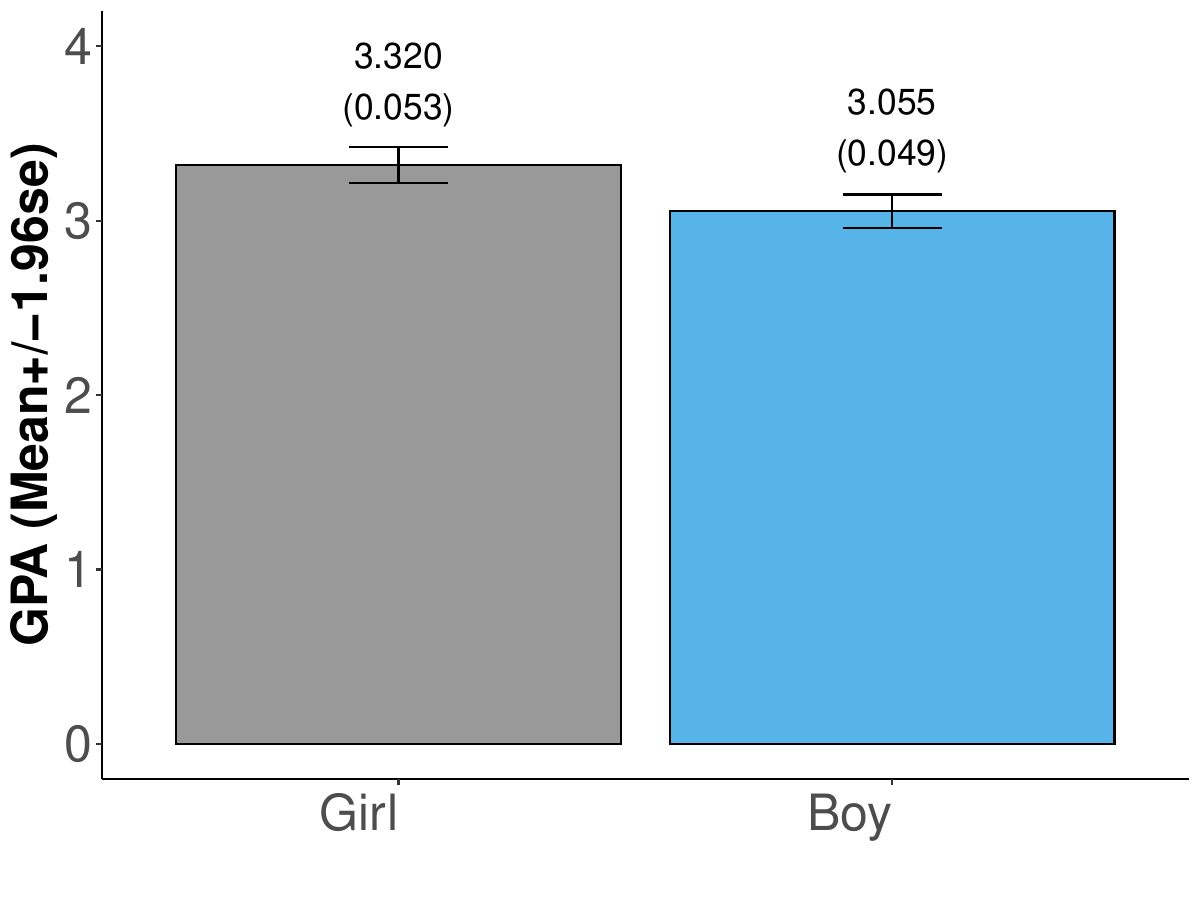} 
        \caption*{a. Gender} 
  \end{minipage}%%
   \begin{minipage}[b]{0.5\linewidth}
    %\centering
    \includegraphics[width=\linewidth]{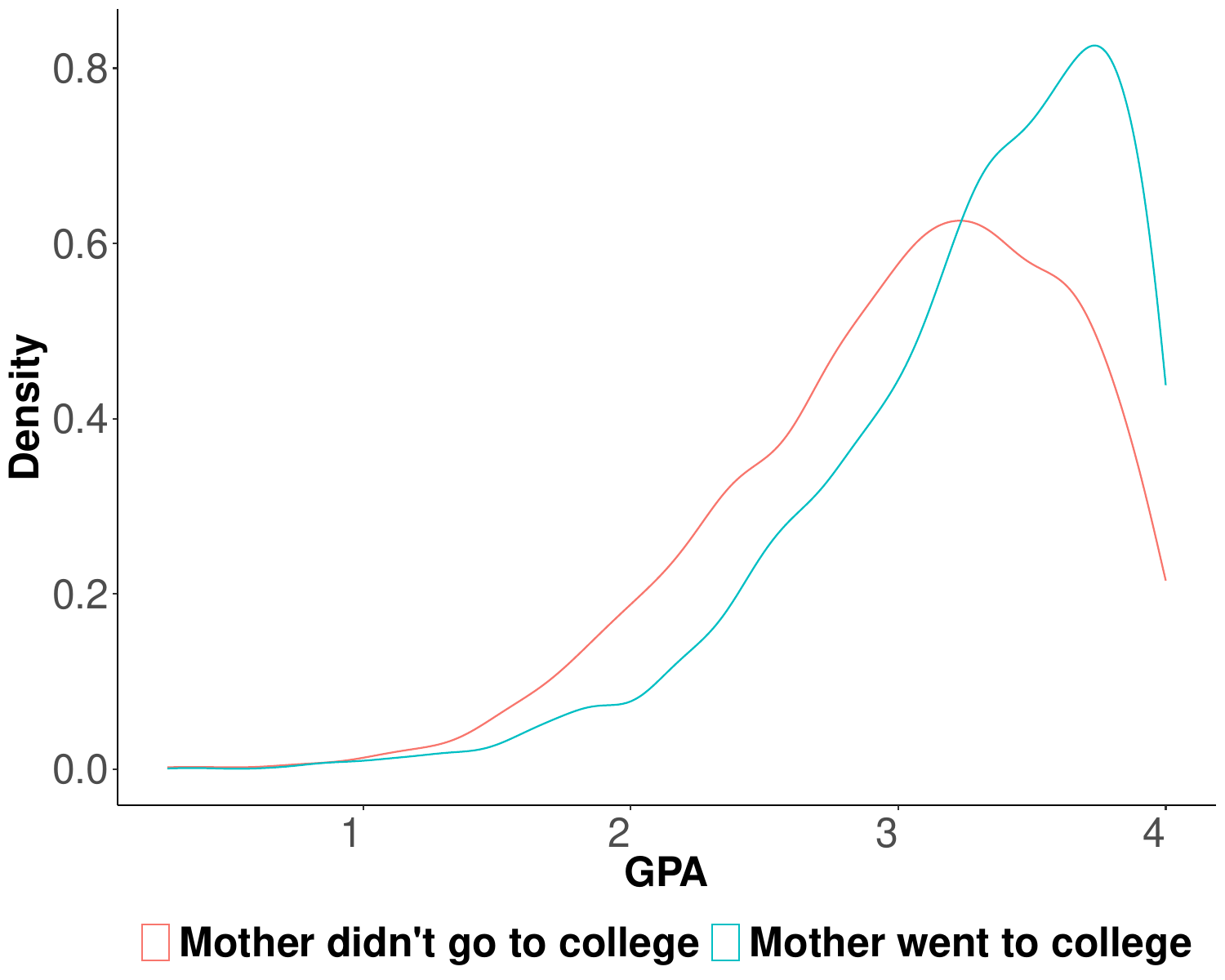}
    \caption*{b. Mother's Education} 
  \end{minipage}
  \begin{minipage}[b]{0.5\linewidth}
    \centering
    \includegraphics[width=\linewidth]{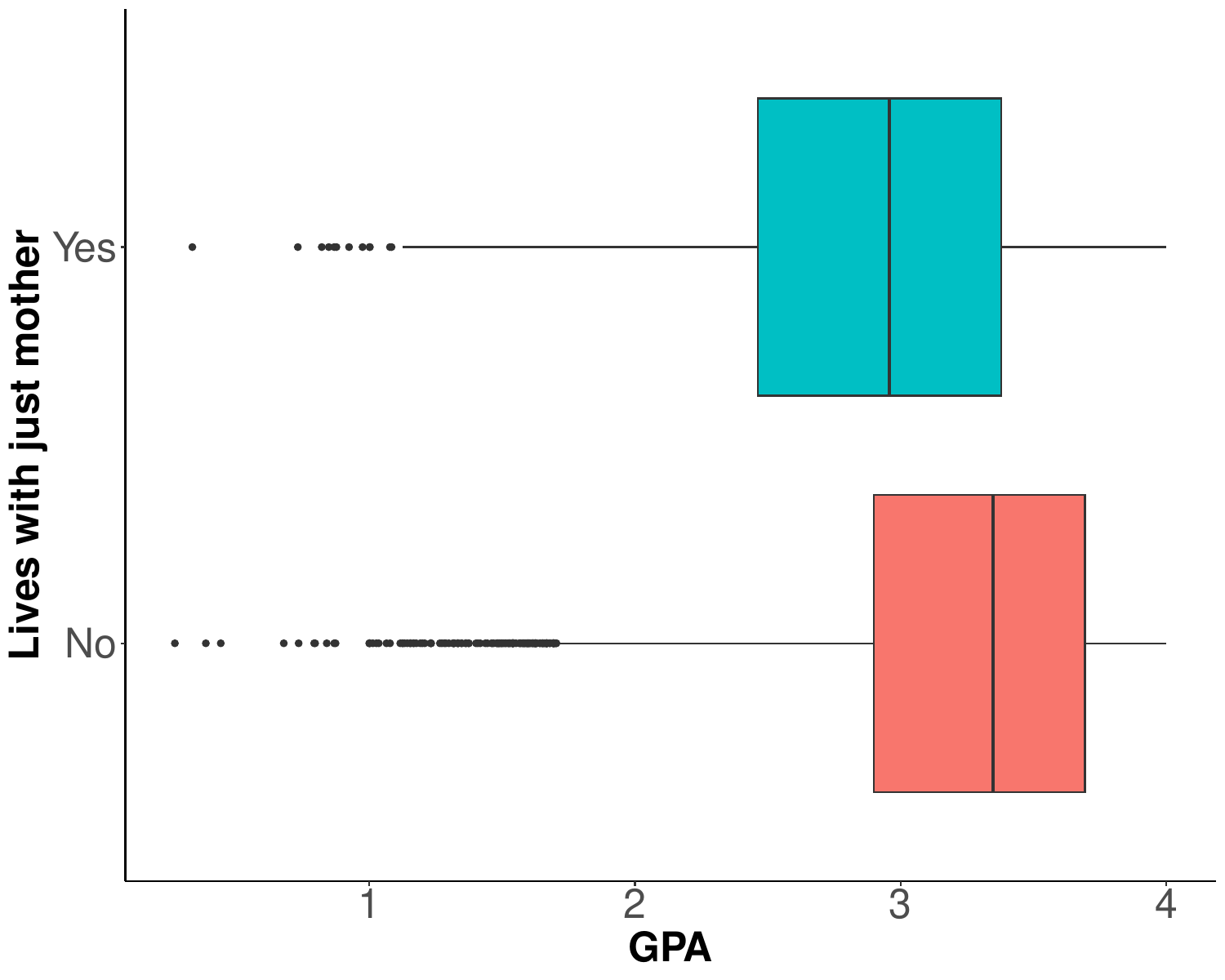} 
        \caption*{c. Single mother} 
  \end{minipage}%%
   \begin{minipage}[b]{0.5\linewidth}
    %\centering
    \includegraphics[width=\linewidth]{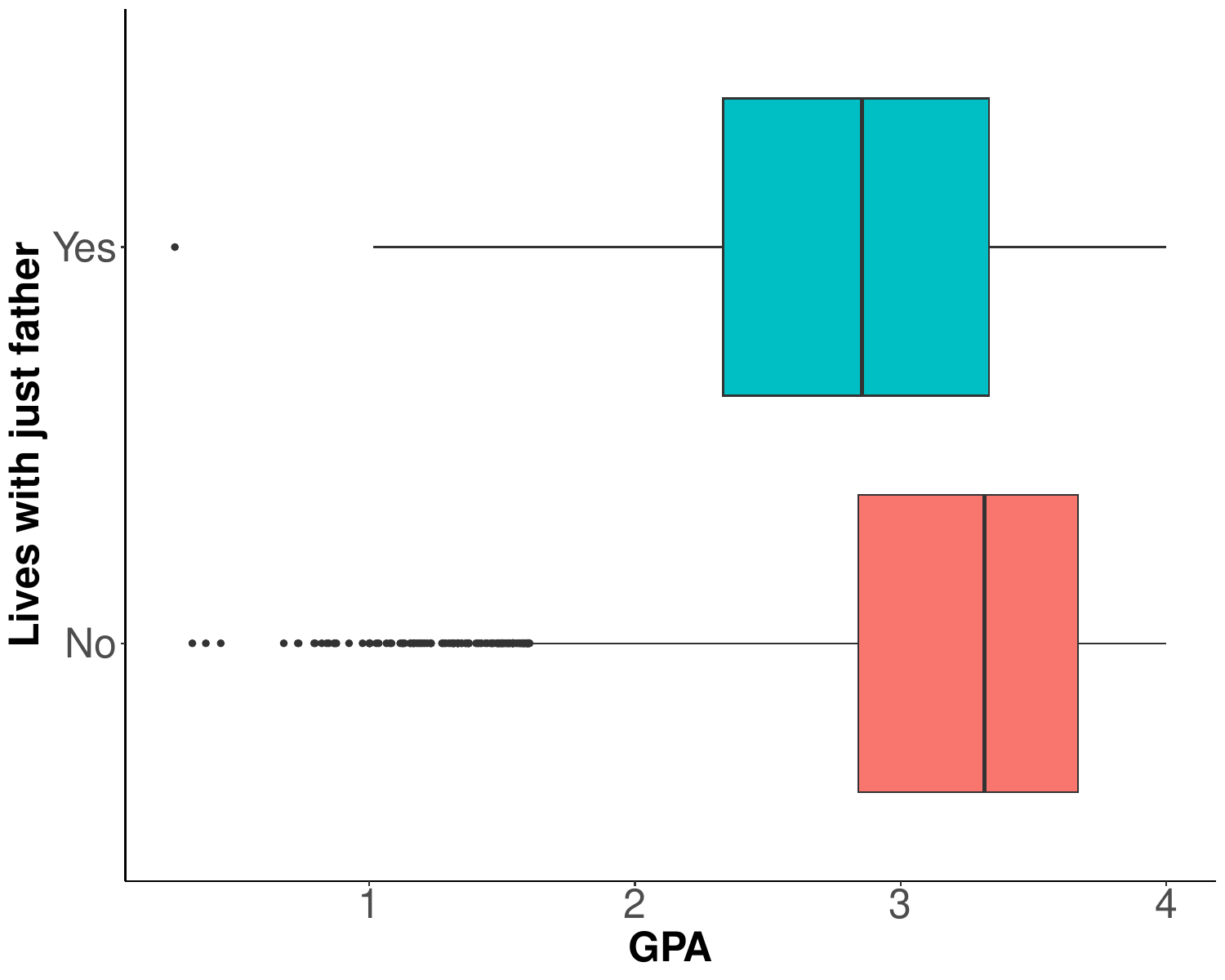}
    \caption*{d. Single father} 
  \end{minipage}
 \begin{minipage}{13.5 cm}{\footnotesize{Notes: Panel (a) illustrates the differences in average GPA by gender. 95\% CI is shown using error bars, and the mean and standard error are provided at the top of each bar. Panel (b) shows the density plots for GPA separately by mother's education. Panel (c) and (d) depict separate box plots for GPA for students who reported to be living with just their mother and father, respectively.}}
\end{minipage} 
\end{figure}

Table \ref{regressionestimatesgpa} provides the results on peer influence in student GPA. Column (1) provides the estimates of SAR when we do not account for latent homophily. Column (2) provides estimates for when we adjust for latent homophily and correct for measurement error using the methods developed in this paper. The selection of the dimension of latent homophily factors is done using network cross-validation methods developed in \cite{li2020network}. The peer influence parameter drops by 18.3\% once we adjust for the latent homophily factors. Although peer influence drops in magnitude after adjusting for latent homophily, the 95\% confidence interval does not include 0. This suggests that friends with high GPAs can positively impact GPA.

\begin{table}[!htbp]
\centering
\captionsetup{width=10cm}
\caption{Peer Effects in GPA}
\begin{adjustbox}{width=0.8 \columnwidth,center}
\label{regressionestimatesgpa}
\begin{tabular}{p{7 cm}p{4cm}p{4cm}}\hline\hline
Dependent Variable&\multicolumn{2}{c}{ Grade Point Average (GPA)}\\
     & SAR Unadjusted & SAR Adjusted \\
    &(1)&(2)\\
    \hline
Peer Influence                                             & 0.1452   & 0.1189   \\
                                                      & (0.0186) & (0.0162) \\
Age                                            & -0.1028  & -0.1022  \\
                                                      & (0.0128) & (0.0137) \\
6th grade                                    & -0.1152  & -0.1159  \\
                                                      & (0.0313) & (0.0265) \\
7th grade                                    & -0.0641  & -0.0633  \\
                                                      & (0.0382) & (0.0354) \\
8th grade                                    & -0.0292  & -0.0348  \\
                                                      & (0.0479) & (0.0471) \\
Gender (Boy)                                       & -0.2070  & -0.2098  \\
                                                      & (0.0124) & (0.0125) \\
Black ethnicity                             & -0.2642  & -0.2783  \\
                                                      & (0.0223) & (0.0257) \\
Hispanic                                    & -0.1426  & -0.1418  \\
                                                      & (0.0155) & (0.0169) \\
Mother went to college                       & 0.1344   & 0.1334   \\
                                                      & (0.0138) & (0.0146) \\
Yes: lives with just mom                  & -0.2282  & -0.2273  \\
                                                      & (0.0183) & (0.0205) \\
Yes: lives with just dad                  & -0.2806  & -0.2765  \\
                                                      & (0.0465) & (0.0564) \\
Yes: lives with parents but separated     & -0.1561  & -0.1537  \\
                                                      & (0.0206) & (0.0216) \\
Do participate in sports at school         & 0.0492   & 0.0501   \\
                                                      & (0.0125) & (0.0125) \\
Do participate in sports outside of school & 0.0919   & 0.0927   \\
                                                      & (0.0136) & (0.0138) \\
Do participate in theater/drama             & -0.0129  & -0.0161  \\
                                                      & (0.0190) & (0.0176) \\
Do participate in music                     & 0.0487   & 0.0525   \\
                                                      & (0.0131) & (0.0129) \\
Do participate in other arts                & -0.0072  & -0.0106  \\
                                                      & (0.0159) & (0.0155) \\
Do participate in other school clubs        & 0.1357   & 0.1358   \\
                                                      & (0.0144) & (0.0133) \\
Date people at this school                  & -0.1630  & -0.1597  \\
                                                      & (0.0146) & (0.0159) \\
Do lots of homework                         & 0.1661   & 0.1671   \\
                                                      & (0.0122) & (0.0122) \\
Do read books for fun                       & 0.1844   & 0.1829   \\
                                                      & (0.0132) & (0.0125) \\
Intercept                                             & 3.8878   & 3.9724   \\
                                                      & (0.1468) & (0.1507) \\\\
                                                      N&7872&7872\\
\hline\hline
\multicolumn{3}{c}{ \begin{minipage}{15cm}{\footnotesize{Notes: The table provides the parameter estimates and the standard errors in the parenthesis. Column (1) provides the estimates with no homophily correction and column (2) provides the estimates once we adjust for unobserved homophily and correct for the measurement error bias.}}
\end{minipage} } \\
\end{tabular}
\end{adjustbox}
\end{table}

In addition to peer GPA, we find that, on average, boys have lower GPAs than girls in the sample. The GPA in 6th, 7th, and 8th grade is lower than the baseline 5th grade. Students who live with just a single parent have lower GPAs on average than students living with both parents. Moreover, participating in sports in school, reading books, and doing lots of homework are positively correlated with GPA. In line with Figure \ref{fig:gpacov}, students whose mothers went to college have a higher GPA than those who did not go to college.

\subsection{Covid-19 Data}
\begin{table}[!htbp]
\centering
\captionsetup{width=12cm}
\caption{Mean Variable as Proxy for Target Covariate}
%\begin{adjustbox}{width=1.3\columnwidth,center}
\label{regressionestimates2}
\begin{tabular}{p{6 cm}p{2cm}p{2cm}p{2cm}}\hline\hline
Dependent Variable&\multicolumn{3}{c}{(COVID-19 Deaths/Total Population)*10000}\\
    Variable & OLS & Uncorrected SAR & Corrected SAR \\
    &(1)&(2)&(3)\\
    \hline
Intercept                  & -19.259179 & -28.727864 & -28.450157 \\
                               & (3.089590) & (2.982157) & (3.676250) \\
Median Age 2010                  & 0.354176   & 0.348559   & 0.342334   \\
                               & (0.028633) & (0.027629) & (0.030604) \\
Heart Disease Mortality          & 0.021810   & 0.016798   & 0.016639   \\
                               & (0.003820) & (0.003686) & (0.004488) \\
Stroke Mortality                & 0.068863   & 0.052170   & 0.053388   \\
                               & (0.017779) & (0.017157) & (0.018296) \\
Smokers Percentage            & -0.061639  & -0.032255  & -0.030069  \\
                               & (0.050445) & (0.048680) & (0.055668) \\
Resp. Mortality Rate 2014          & -0.064712  & -0.053260  & -0.053044  \\
                               & (0.009347) & (0.009019) & (0.010508) \\
Hospitals                    & -0.092748  & -0.101527  & -0.269912  \\
                               & (0.124484) & (0.120118) & (0.105682) \\
ICU beds                    & -0.012690  & -0.013819  & -0.020501  \\
                               & (0.004213) & (0.004065) & (0.003242) \\
Mask Rarely                    & -0.822454  & -1.215399  & -1.071396  \\
                               & (4.009183) & (3.868561) & (4.455684) \\
Mask Sometimes                 & -2.169702  & -2.778055  & -2.795821  \\
                               & (3.528897) & (3.405119) & (3.644980) \\
Mask Frequently                & 0.944854   & 0.257828   & 0.471734   \\
                               & (3.044751) & (2.937959) & (3.406656) \\
Mask Always                    & 2.850727   & 3.251691   & 2.790626   \\
                               & (2.469976) & (2.383362) & (2.881731) \\
Population Density per Sq Mile 2010 & 0.000643   & 0.000748   & 0.000777   \\
                               & (0.000176) & (0.000170) & (0.000187) \\
SVI Percentile                  & 6.133201   & 5.539482   & 5.636224   \\
                               & (0.575071) & (0.554968) & (0.574216) \\
Mean mortality                  & 0.001794   & 0.002017   & 0.003863   \\
                               & (0.000665) & (0.000642) & (0.000607) \\
Cases per 10,000 population                      & 0.015359   & 0.013952   & 0.013994   \\
                               & (0.000528) & (0.000509) & (0.000836)  \\\\
N &2807&2807&2807\\\hline\hline
\multicolumn{4}{c}{ \begin{minipage}{15 cm}{\footnotesize{Notes: The table provides the parameter estimates and the standard errors in the parenthesis. The mean mortality for old population corresponds to $\tilde{U}$ in this application problem, and all other covariates are assumed to be measured without error (Z).}}
\end{minipage} } \\
\end{tabular}
%\end{adjustbox}
\end{table}
This section illustrates the measurement error correction in covariates for SAR model on COVID-19 county-level deaths in 2020. For this analysis, we obtained the county-level abridged data file constructed in \cite{altieri2020curating}. The authors of \cite{altieri2020curating} provide the description of the data variables in this \href{https://github.com/Yu-Group/covid19-severity-prediction/blob/master/data/list_of_columns.md}{GitHub repo}. We work with a subset of covariates relevant to our application listed in the Appendix Table \ref{summarytable}. The unit of observation is county $i$, and we observe the latitude-longitude information for the county's population center. Using the geo-codes, we compute the spatial weight matrix of geodesic distances using the \texttt{\href{https://cran.r-project.org/web/packages/geodist/geodist.pdf}{geodist}} package in R.

As discussed earlier in Section \ref{dataexample}, we are interested in estimating the impact of old age mortality rate on COVID-19 deaths. However, this variable is not measured in the data. Instead, we observe 3 variables: the mortality rate for the age group 65-74 years, 75-84 years, and 85 and above. Consequently, we compute an average of these 3 variables and estimate the variance of $\xi$ using the formulae in Section \ref{dataexample}. Table \ref{regressionestimates2} provides the estimated coefficients for OLS for linear models, the uncorrected SAR, and the SAR model corrected for measurement error in Columns (1), (2), and (3), respectively. As discussed above, measurement error, even in one covariate, introduces bias not only for the coefficient corresponding to that covariate but also for the other covariates that are measured without error ($Z$). This can be seen when we compare Columns (2) and (3) in Table \ref{regressionestimates2}. In comparison to the linear model (OLS), we see changes in the coefficients for several covariates when we use a SAR model (whether uncorrected or corrected). This shows that accounting for spatial dependence through SAR might be beneficial. Second, between the uncorrected and corrected SAR, a few coefficients differ, including that of the target covariate ``mean mortality" for which we used replicate measurements.

\section{Conclusion}
In this paper, we develop new methods for bias correction for SAR models when there is measurement error in covariates. The need for this new methodology is motivated by two types of data applications. First, the researcher might face a situation where the target covariate is not available and instead replicate error-prone measurements are available or the target variable is not available for the whole data but a related proxy variable is observed for all units. In this scenario, with an estimate of the covariance between the proxy and target covariate, one can use the proposed ME-QMLE method to correct for the bias on the parameters of the outcome model. Second, this method can also be applied in network SAR models where the researcher controls for unobserved homophily in the outcome model by estimating latent factors from the network adjacency matrix. Since the latent factors are measured with error, our proposed methods can be used to correct the bias induced by this measurement error. 

The bias-corrected estimator has useful asymptotic properties. We prove that our estimator is consistent and follows an asymptotically normal distribution. The limiting covariance matrix involves unconditional expectation (with respect to both model error term and measurement error term) of the negative Hessian matrix and the unconditional variance of the score vector. We also provide a consistent estimator for the limiting covariance matrix in the presence of measurement error. We use these asymptotic results to derive closed-form solutions for the standard error of the estimator. Our simulation analysis verifies the bias correction property and the validity of the standard error estimates of the estimator in finite samples. Finally, we apply the new estimator to two datasets, one on estimating peer influence in GPA for middle school students in New Jersey, and on county-level COVID-19 deaths in 2020 in the United States. 

\noindent The proofs and additional data and tables are provided in the Appendix.

    % references

%\newpage
\bibliographystyle{plainnat}
\bibliography{netgen,coevol,hyperbol,psychology,vc,triads,triads1,Networks,GraphTheory,
cluster,hawkesref,sarpeer,msp}

\newpage

\section*{Appendix S: Proofs of Results}
    \renewcommand{\theequation}{S.\arabic{equation}}
    \renewcommand{\thesection}{S}
    \setcounter{equation}{0}
\setcounter{table}{0}
\renewcommand{\thetable}{S\arabic{table}}

\renewcommand{\thefigure}{S\arabic{figure}}
\setcounter{figure}{0}   

    \medskip
    % appendix (put online supplement in a separate file)
\section{Proofs, Additional Tables and Figures}

\subsection{Some propositions and Lemmas}
\label{lemmaproof}
The following results will be used repeatedly in the proofs.

\begin{prop}
    Under the assumptions A1-A6, we have
    \begin{equation}
\frac{1}{n}\begin{pmatrix}
 \tilde{X}_n^T X_n\delta_0\\
 \tilde{X}_n^T G_nX_n\delta_0\\
\tilde{X}_n^T\tilde{X}_n-\sum_{i}\Omega_{i}
\end{pmatrix} \equiv \begin{pmatrix}
 \frac{1}{n}\sum_i (X_n\delta_0)_i \tilde{X}_{in}^T\\
 \frac{1}{n}\sum_i (G_nX_n\delta_0)_i \tilde{X}_{in}^T\\
\frac{1}{n} \sum_i(\tilde{X}_{in}\tilde{X}_{in}^T-\Omega/n)
\end{pmatrix} \overset{p}{\to} 
\frac{1}{n} \begin{pmatrix}
X_n^T  X_n\delta_0\\
X_n^T  G_nX_n\delta_0 \\
X_n^TX_n
\end{pmatrix}.
\end{equation}
\label{wlln}
\end{prop}

\begin{proof}
    These results follow from the weak law of large numbers since the assumptions provide sufficient regularity conditions on $\tilde{X}_n$. 
\end{proof}

\begin{lem}
    Let $\tilde{X}_n$ be a $n \times d$ random matrix with  uniformly bounded entries for all $i,j$. Further assume  the elements of the matrix $ (\frac{1}{n} \tilde{X}_n^T \tilde{X}_n -\frac{\sum_i \Omega_i}{n})^{-1})$ are also uniformly bounded for all $i,j$. Then for the projectors $M_{2n}$ and $ I_n- M_{2n}$ where $M_{2n} = I_n - \tilde{X}_n(\tilde{X}_n^T\tilde{X}_n - \sum_i \Omega_i)^{-1}\tilde{X}_n^T $, both the row and column absolute sums are bounded uniformly, i.e., $ \sum_j |(M_{2n})_{ij}| <\infty$ and all $i$.
    \label{m2bound}
\end{lem}

\begin{proof}
    Let $B_n = (\frac{1}{n}(\tilde{X}_n^T\tilde{X}_n - \sum_i \Omega_i))^{-1}.$ Then, 
    \[
    I-M_{2n} = \frac{1}{n} \tilde{X}_n(\frac{\tilde{X}_n^T\tilde{X}_n - \sum_i \Omega_i}{n})^{-1}\tilde{X}_n^T = \frac{1}{n} \sum_{r} \sum_s B_{n,r,s}\tilde{x}_{n,r} \tilde{x}_{n,s}^T, 
    \]
    where $\tilde{x}_{n,r}$ denotes the $r$th column of $\tilde{X}_{n}$.
    Then we compute
    \begin{align*}
        \sum_j |(M_{2n})_{ij}| & \leq \frac{1}{n} \sum_j \sum_{r} \sum_s |B_{n,r,s}\tilde{x}_{n,ir} \tilde{x}_{n,js}|\\
        & \leq \frac{1}{n} \sum_j \sum_{r} \sum_s |B_{n,r,s}| |\tilde{x}_{n,ir}|| \tilde{x}_{n,js}]|.
    \end{align*}

    Now, from the given conditions of the lemma, we have $|\tilde{x}_{n,js}| \leq c_x$ and $|B_{n,r,s}| \leq c_b$ for some constants $c_x$ and $c_b$. Therefore the sum above is upper bounded by $d^2c_xc_b$, which is constant as a function of $n$. Hence the claimed result follows.
\end{proof}

\begin{lem}
    Under the assumptions A1-A6, we have the following results
\begin{enumerate}
    \item $S_n(\rho)S_n^{-1}$ is uniformly bounded in row and column sums. 
    \item The elements of the vectors $X_n \delta_0$ and $G_n X_n \delta_0$ are uniformly bounded in $n$, where $G_{n}=L_{n}S_{n}^{-1}$.
    \item  $M_{2n} S_n(\rho)S_n^{-1}$ and $(S_n^T)^{-1})S_n^T(\rho)M_{2n}S_n(\rho)S_n^{-1}$ are uniformly bounded in row and column sums.
\end{enumerate}
\label{uniformbound}
\end{lem}

\begin{proof}
    
For the first result,
$\|S_n(\rho)S_n^{-1}\|_{\infty} \leq \|S_n(\rho)\|_{\infty} \|S_n^{-1}\|_{\infty} <\infty$, for all $\rho$. The same holds for the column sum norm.
For the second result, 
since all elements of $X_n$ are uniformly bounded in $n$ and $\delta_0$ is finite, the elements of $X_n \delta_0$ are composed of sums of $d$ bounded elements (with $d$ not growing with $n$). Consequently, the elements of $X_n \delta_0$ are bounded for all $i$. For $G_nX_n\delta_0$ we note that,
$\|G_n X_n\|_{\infty} \leq \|L_n\|_{\infty} \|S_n^{-1}\|_{\infty} \|X_n\|_{\infty} <\infty$. Then, the same argument as above leads to the conclusion that elements of $G_nX_n\delta_0$ are bounded.

For the third result, we note that from Lemma \ref{m2bound}, $M_{2n}$ is uniformly bounded in row and column sums. Then, by the norm inequality of row and column sum norms, the uniform bound result follows.
\end{proof}
\subsection{Proof of Theorem 1}

\begin{proof} 

\label{proofthm1}

To show consistency of $\hat{\theta}_n$, we first analyze the consistency of $\hat{\rho}_n$ and then that of $\hat{\delta}_n$ and $\hat{\sigma}^2_n$. 

The true data-generating process is
\[
Y_n=  S_n^{-1} X_n\delta_0 + S_n^{-1}V_n, \quad V_n \sim N(0, \sigma^2 I_n).
\]
Recall the corrected loglikelihood function maximized to obtain the estimators is
\[
l_n^{*}(\theta) = -\frac{n}{2} \log ( 2 \pi \sigma^2) -\frac{1}{2\sigma^2}[ \tilde{V}_n(\rho)^T \tilde{V}_n(\rho) -  \delta^T(\sum_i \Omega_i) \delta] + \log |S_n(\rho)|,
\]
where 
\[
\tilde{V}_n(\rho) = S_n(\rho)Y_n - \tilde{X}_n \delta.
\]
Now define
\[
Q_n(\rho) = \max_{\beta, \sigma^2} \mathbb E[l_n^*(\theta)] = \max_{\beta, \sigma^2} \mathbb E^+ \mathbb E^* [l_n^*(\theta)],
\]
as the concentrated unconditional expectation of the corrected log-likelihood function and $\mathbb E^+$ is the expectation with respect to the random variable $V$ and $\mathbb E^{*}$ as the expectation with respect to the random variable $\xi$. Moreover, 
\begin{align}
\label{eqproof}
    \mathbb E^+ \mathbb E^* &[l_n^*(\theta)] = -\frac{n}{2} \log ( 2 \pi \sigma^2)-\frac{1}{2\sigma^2}\mathbb E^+ \mathbb E^* \bigg [(S_n(\rho)Y_n - \tilde{X}_n \delta)^{T}(S_n(\rho)Y_n - \tilde{X}_n \delta)\bigg ]-\nonumber\\
    &\frac{1}{2\sigma^2}[\delta^T(\sum_i \Omega_i) \delta]+\log |S_n(\rho)|\nonumber\\
    &= -\frac{n}{2} \log ( 2 \pi \sigma^2)-\frac{1}{2\sigma^2}\mathbb E^+ \mathbb E^* \bigg [(S_n(\rho)Y_n - ({X}_n+\eta) \delta)^{T}(S_n(\rho)Y_n - ({X}_n+\eta) \delta)\bigg ]-\nonumber\\
    &\frac{1}{2\sigma^2}[\delta^T(\sum_i \Omega_i) \delta]+\log |S_n(\rho)|\nonumber\nonumber\\
    &= -\frac{n}{2} \log ( 2 \pi \sigma^2)-\frac{1}{2\sigma^2}\mathbb E^+\mathbb E^*  \bigg [ Y_{n}^{T}S_{n}(\rho)^{T}S_{n}(\rho) Y_{n}-Y_{n}^{T}S_{n}(\rho)^{T}\tilde{X}_{n}\delta -\delta^{T}\tilde{X}_{n}^{T}S_{n}(\rho) Y_{n}+\delta^{T}\tilde{X}_{n}^{T}\tilde{X}_{n}\delta \bigg ]\nonumber\\
    &-\frac{1}{2\sigma^2}[\delta^T(\sum_i \Omega_i) \delta]+\log |S_n(\rho)|\nonumber\\
    &= -\frac{n}{2} \log ( 2 \pi \sigma^2)-\frac{1}{2\sigma^2}\mathbb E^+ \bigg [ Y_{n}^{T}S_{n}(\rho)^{T}S_{n}(\rho) Y_{n}-Y_{n}^{T}S_{n}(\rho)^{T}{X}_{n}\delta -\delta^{T}{X}_{n}^{T}S_{n}(\rho) Y_{n}+\nonumber\\
    &\delta^{T}\mathbb E^{*}({X}_{n}^{T}{X}_{n}+\eta^{T}\eta+X_{n}^{T}\eta+\eta^{T}X_{n})\delta \bigg ]-\frac{1}{2\sigma^2}[\delta^T(\sum_i \Omega_i) \delta]+\log |S_n(\rho)|\nonumber\\
    &= -\frac{n}{2} \log ( 2 \pi \sigma^2)-\frac{1}{2\sigma^2}\mathbb E^+ \bigg [ Y_{n}^{T}S_{n}(\rho)^{T}S_{n}(\rho) Y_{n}-Y_{n}^{T}S_{n}(\rho)^{T}{X}_{n}\delta -\delta^{T}{X}_{n}^{T}S_{n}(\rho) Y_{n}+\delta^{T}({X}_{n}^{T}{X}_{n})\delta)\bigg ]\nonumber\\
    &+\log |S_n(\rho)|
\end{align}

Now solving the optimization problem, we obtain the solutions for $\delta, \sigma^2$ for a given $\rho$ as follows.

\noindent\textbf{Gradient w.r.t $\delta$}:
\begin{align*}
     & -\frac{1}{2\sigma^{2}}\bigg[ -2X_{n}^{T}S_{n}(\rho)\mathbb E^{+}[Y_{n} ]+2X_{n}^{T} X_{n}\delta   \bigg]=0\\
     &\delta^{*}(\rho)  =(X_n^TX_n)^{-1}X_n^TS_n(\rho) \mathbb E^{+}[Y_n] \\
     &=  (X_n^TX_n)^{-1}X_n^TS_n(\rho) S_n^{-1}X_n \delta_0
\end{align*}

\noindent\textbf{Gradient w.r.t $\sigma^{2}(\rho)$}: We use the $\mathbb E^{+}$ as stated in the last step in equation \ref{eqproof}. 
\begin{align*}
\mathbb E^+ \mathbb E^* [l_n^*(\theta)]= &-\frac{n}{2} \log ( 2 \pi \sigma^2)-\frac{1}{2\sigma^2} \bigg [ \mathbb E^+[Y_{n}^{T}S_{n}(\rho)^{T}S_{n}(\rho) Y_{n}]-\mathbb E^+[Y_{n}^{T}]S_{n}(\rho)^{T}{X}_{n}\delta -\\&\delta^{T}{X}_{n}^{T}S_{n}(\rho) \mathbb E^+[Y_{n}]+\delta^{T}({X}_{n}^{T}{X}_{n})\delta)\bigg ]
    +\log |S_n(\rho)|\\
    &=\frac{n}{2} \log ( 2 \pi \sigma^2)-\frac{1}{2\sigma^2} \bigg [ \delta_{0}^{T}X_{n}^{T}(S_{n}^{-1})^{T}S_{n}(\rho)^{T}S_{n}(\rho)S_{n}^{-1}X_{n}\delta_{0}+\sigma_{0}^{2}tr((S_{n}^{-1})^{T}S_{n}(\rho)^{T}S_{n}(\rho)S_{n}^{-1})\\
    &-\delta_{0}^{T}X_{n}^{T}(S_{n}^{-1})^{T}S_{n}(\rho)^{T}X_{n}(X_{n}^{T}X_{n})^{-1}X_{n}^{T}S_{n}(\rho)S_{n}^{-1}X_{n}\delta_{0}\bigg]\\
    &=\frac{n}{2} \log ( 2 \pi \sigma^2)-\frac{1}{2\sigma^2} \bigg [ \delta_{0}^{T}X_{n}^{T}(S_{n}^{-1})^{T}S_{n}(\rho)^{T}M_{1n}S_{n}(\rho)S_{n}^{-1}X_{n}\delta_{0}+\\   &\sigma_{0}^{2}tr((S_{n}^{-1})^{T}S_{n}(\rho)^{T}S_{n}(\rho)S_{n}^{-1})\bigg]\\
    &=\frac{n}{2} \log ( 2 \pi \sigma^2)-\frac{1}{2\sigma^2} \bigg [ \delta_{0}^{T}X_{n}^{T}(G_{n}^{T}(\rho_{0}-\rho)+I)M_{1n}((\rho_{0}-\rho)G_{n}+I)X_{n}\delta_{0}+\\    &\sigma_{0}^{2}tr((S_{n}^{-1})^{T}S_{n}(\rho)^{T}S_{n}(\rho)S_{n}^{-1})\bigg]\\
    &=\frac{n}{2} \log ( 2 \pi \sigma^2)-\frac{1}{2\sigma^2} \bigg [ (\rho_{0}-\rho)^2\delta_{0}^{T}X_{n}^{T}(G_{n}^{T})M_{1n}(G_{n})X_{n}\delta_{0}+\sigma_{0}^{2}tr((S_{n}^{-1})^{T}S_{n}(\rho)^{T}S_{n}(\rho)S_{n}^{-1})\bigg]\\
\end{align*}
where $M_{1n} = (I_n - X_n (X_n^TX_n)^{-1}X_n^T)$. The penultimate line follows since $S_n(\rho)S_n^{-1} =(\rho_0 - \rho)G_n + I_n$ and the final line follows since $M_{1n}X_n\delta_0 = 0$ and $\delta_0^TX_n^TM_{1n}=0$. Now solving for optimal $\sigma^{2}(\rho)$
\begin{align*}
    {\sigma^2}(\rho)^{*} & = \frac{1}{n}\left ((\rho_0-\rho)^2(G_n X_n\delta_0)^TM_{1n}(G_nX_n
    \delta_0\right) + \frac{\sigma^2_0}{n} tr \left((S_n^{-1})^{T}S_n(\rho)^TS_n(\rho)S_n^{-1} \right), 
    \label{sigmasq}
\end{align*}

We note that, therefore, the function $Q_n(\rho)$ is identical to the function analyzed in \cite{lee2004asymptotic}. Therefore, the same arguments in \cite{lee2004asymptotic} prove the uniqueness of $\rho_0$ as the global maximizer of $Q_n(\rho)$ in the compact parameter space $R$. In particular, for any $\epsilon>0$, we have
\[
\lim \sup_{n \to \infty} \max_{\rho \in N_{\epsilon}(\rho_0)} |Q_n(\rho) - Q_n(\rho_0)|<0.
\]
Therefore, we only need to show the uniform convergence of the concentrated corrected log-likelihood $l_n^*(\rho)$ to $Q_n(\rho)$ over the compact parameter space $R$.
Now we have
\[
\frac{1}{n} (l^*_n(\rho) - Q_n(\rho)) = -\frac{1}{2}(\log (\hat{\sigma}_n^2) - \log (\sigma^{2}(\rho))),
\]
where 
\begin{equation}
   \hat{\sigma}^2_n(\rho) = \frac{1}{n} Y_n^T S_n^T(\rho) M_{2n} S_n(\rho) Y_n,
\end{equation}
with 
\[
M_{2n}= I_n - \tilde{X}_n(\tilde{X}_n^T\tilde{X}_n - \Omega)^{-1}
\tilde{X}^T_n.
\]
Since
\[
I + (\rho_0 -\rho) G_n = I + \rho_0 L_nS_n^{-1}
 - \rho L_n S_n^{-1} = I + (I-S_n) S_n^{-1}
 - \rho L_n S_n^{-1} = S_n^{-1} - \rho L_n S_n^{-1} = S_n(\rho) S_n^{-1},
\]
We can expand $\hat{\sigma}_n^2(\rho)$ to obtain
\begin{align*}
   \hat{\sigma}^2_n(\rho) &  = \frac{1}{n} (S_n^{-1}X_n \delta_0 + S_n^{-1}V_n)^T S_n^T(\rho) M_{2n} S_n(\rho) (S_n^{-1}X_n \delta_0 + S_n^{-1}V_n)\\
   & = \frac{1}{n} (X_n\delta_0 + (\rho_0-\rho)G_n X_n\delta_0 + S_n(\rho) S_n^{-1}V_n)^T M_{2n} (X_n\delta_0 + (\rho_0-\rho)G_n X_n\delta_0 + S_n(\rho) S_n^{-1}V_n) \\
  & =  \frac{1}{n} \bigg ((\rho_0-\rho)^2(G_nX_n\delta_0)^{T}M_{2n}(G_nX_n\delta_0 )   + (X_n \delta_0)^T M_{2n} (X_n \delta_0) \\
  & \quad + V_n^T(S_n^{-1})^{T} S^T_n(\rho) M_{2n}S_n (\rho) S_n^{-1}V_n  +  2 (\rho_0 -\rho) (G_nX_n \delta_0)^T M_{2n}S_n(\rho) S_n^{-1}V_n \\
  & \quad  + 2 (\rho_0 -\rho) (X_n\delta_0)^T M_{2n}(G_nX_n\delta_0) + 2(X_n\delta_0)^T M_{2n}S_n(\rho)S_n^{-1}V_n \bigg )
\end{align*}

The decomposition of $\hat{\sigma}^2(\rho)$ produces 6 terms that we analyze in groups of terms that can be analyzed with similar techniques.\\

\noindent\textbf{Terms 1, 2 and 5:}
We note the convergences for terms that do not contain $\rho$ from Proposition \ref{wlln}. Therefore, by the multivariate continuous mapping theorem
\begin{align*}
& (G_nX_n\delta_0 )^T M_{2n}(G_nX_n\delta_0 ) \\
& = (G_nX_n\delta_0 )^T (G_nX_n\delta_0 ) - (G_nX_n\delta_0 )^T  \tilde{X}_n(\tilde{X}_n^T\tilde{X}_n - \Omega)^{-1}
\tilde{X}_n^T G_nX_n\delta_0  \\
& \overset{p}{\to} (G_nX_n\delta_0 )^T (G_nX_n\delta_0 ) - (G_nX_n\delta_0 )^T  X_n(X_n^TX_n)^{-1}X_n^T G_nX_n\delta_0  \\
&=(G_nX_n\delta_0 )^T M_{1n}G_nX_n\delta_0  .
\end{align*}
In addition, $(\rho_0-\rho)^2$ is bounded by a constant since $\rho_0, \rho \in R$, which is a compact set. Therefore, the first term converges in probability to the first term in Equation \eqref{sigmasq} uniformly for all $\rho$.

For the second term, the multivariate continuous mapping theorem gives
\begin{align*}
 (X_n\delta_{0} )^T M_{2n}(X_n\delta_{0} ) &  = (X_n\delta_{0} )^T (X_n\delta_{0} ) - (X_n\delta_{0} )^T  \tilde{X}_n(\tilde{X}_n^T\tilde{X}_n - \Omega)^{-1}
\tilde{X}_n^T X_n\delta_{0}  \\
& \overset{p}{\to} (X_n\delta_{0} )^T (X_n\delta_{0} ) - (X_n\delta_{0} )^T  X_n(X_n^TX_n)^{-1}X_n^T X_n\delta_{0} \\
&=0.
\end{align*}
Finally, for the fifth term a similar argument shows
\begin{align*}
 (X_n\delta_0 )^T M_{2n}(G_n X_n\delta_0 ) &  = (X_n\delta_0 )^T (G_n X_n\delta_0 ) - \delta_0^T X_n^T \tilde{X}_n(\tilde{X}_n^T\tilde{X}_n - \Omega)^{-1}
\tilde{X}_n^T G_n X_n\delta_0  \\
& \overset{p}{\to} (X_n\delta_0 )^T (G_n X_n\delta_0 ) - \delta_0^T X_n^T X_n(X_n^TX_n)^{-1}X_n^T G_n X_n\delta_0 \\
&=0 .
\end{align*}

\noindent\textbf{Terms 4 and 6 :}  Notice for term 6,
\begin{align*}
 (X_n\delta_0)^T M_{2n}S_n(\rho)S_n^{-1}V_n & =  (X_n\delta_0)^T (S_n(\rho)S_n^{-1}V_n)  - (X_n\delta_0)^T \tilde{X}_n(\tilde{X}_n^T\tilde{X}_n - \Omega)^{-1}
\tilde{X}_n^T S_n(\rho)S_n^{-1}V_n. 
\end{align*}
For this term, we apply the Uniform WLLN (\cite{keener2010theoretical}, chapter 9) to the random function $g(V_i, \tilde{X}_n, \rho) = \frac{1}{n} (X_n\delta_0)^T (M_{2n}S_n(\rho)S_n^{-1}V_n) = \frac{1}{n} \sum_i (M_{2n}S_n(\rho)S_n^{-1}V_n)_i (X_n\delta_0)_i$.
Further from Lemma \ref{uniformbound}, the elements of $X_n \delta_0$ are uniformly bounded, and the absolute row sums of $M_{2n}S_n(\rho)S_n^{-1}$ are bounded . Therefore we  have for any $i$ and any $\rho$,
\begin{align*}
\mathbb E | (M_{2n}S_n(\rho)S_n^{-1}V_n)_i (X_n\delta_0)_i | & \leq \mathbb E [ | (M_{2n}S_n(\rho)S_n^{-1}V_n)_i|| (X_n\delta_0)_i |] \\
& \leq | (X_n\delta_0)_i | \mathbb E \sum_{j} | (M_{2n} S_n(\rho)S_n^{-1})_{ij} (V_n)_j| \\
& \leq | (X_n\delta_0)_1 |  \sum_{j} \mathbb E |(M_{2n} S_n(\rho)S_n^{-1})_{ij}| \mathbb E|(V_n)_j| < \infty.
\end{align*}
The last inequality follows since $M_{2n}$ and $V_n$ are independent by assumption and expectations of the absolute row and column sums of $M_{2n}S_n(\rho)S_n^{-1}$ are bounded for all $\rho$. Therefore, the expectation is finite uniformly for all $\rho \in R$. We can also compute the expectation as 
\begin{align*}
\mathbb E [ (M_{2n}S_n(\rho)S_n^{-1}V_n)_i (X_n\delta_0)_i ] & = \mathbb E [\sum_j (M_{2n}S_n(\rho)S_n^{-1})_{ij}(V_n)_j ](X_n\delta_0)_i\\
& = \sum_j \mathbb E[(M_{2n}S_n(\rho)S_n^{-1})_{ij}] \mathbb E[(V_n)_j ](X_n\delta_0)_i = 0
\end{align*}
Therefore, since $R$ is a compact set and $  \mathbb{E} [\sup_{\rho \in R} | (M_{2n} S_n(\rho)S_n^{-1}V_n)_i (X_n\delta_0)_i | ]< \infty$ for all $i$,  applying Uniform LLN (Theorem 9.2 in \cite{keener2010theoretical}), we conclude, 
\[
(X_n\delta_0)^T M_{2n}S_n(\rho)S_n^{-1}V_n \overset{p}{\to} 0 
\]
uniformly in  $\rho \in R$.

The convergence of Term 4 follows similarly by noting that the elements of $G_nX_n \delta_0$ are also uniformly bounded. 

\noindent\textbf{Term 3:}
 The 3rd term is a quadratic form in the error vector $V_n$ with a matrix $B_n(\tilde{X}_n, \rho) = (S_n^{T})^{-1} S^T_n(\rho) M_{2n}S_n (\rho) S_n^{-1}$ which itself is stochastic (due to presence of measurement error). In contrast, the commonly encountered quadratic forms  are of the form $y^T A y$, with $A$ being a nonrandom matrix. 
 We first note that from the results of Lemma \ref{uniformbound} the matrix $S_n^{T})^{-1} S^T_n(\rho) M_{2n}S_n (\rho) S_n^{-1}$ is bounded in row and column sums. Then using the relationship among matrix norms, namely, $\|A\|_2 \leq \sqrt{\|A\|_{\infty} \|A\|_1}$, the spectral norm of the matrix $B_n(\tilde{X}_n, \rho)$ is uniformly bounded in $\rho$:
 \[
 \|B_n(\tilde{X}_n,\rho)\|_2 \leq \|(S_n^{T})^{-1} S^T_n(\rho) M_{2n}S_n (\rho) S_n^{-1}\|_{\infty} \|(S_n^{T})^{-1} S^T_n(\rho) M_{2n}S_n (\rho) S_n^{-1}\|_{1}\leq C.
 \]
 Given the above result and the assumption that $E|V_i^{4 +\eta}| <\infty$, the result in (Lemma C.3, of \cite{dobriban2018high}) is applicable on the quadratic term (with random middle matrix) $V_n^TB_n(\tilde{X}_n,\rho)V_n$. Therefore, we have the following convergence result.
 \[
 \frac{1}{n} V_n^T(S_n^{T})^{-1} S^T_n(\rho) M_{2n}S_n (\rho) S_n^{-1}V_n - \frac{\sigma_0^2}{n} \tr((S_n^{T})^{-1} S^T_n(\rho) M_{2n}S_n (\rho) S_n^{-1}) \overset{p}{\to} 0.
 \]
Now note that 
\begin{align*}
\tr((S_n^{T})^{-1} S^T_n(\rho) M_{2n}S_n (\rho) S_n^{-1}) 
 & = \frac{1}{n} \sum_i \{(S_n (\rho) S_n^{-1})^T \tilde{X}_n)_i (\frac{  \tilde{X}^T\tilde{X} -\sum_i \Omega_i}{n})^{-1} (\tilde{X}_n^T S_n (\rho) S_n^{-1})\}_{i}.
\end{align*}

Now applying Uniform WLLN (since $(S_n (\rho) S_n^{-1})$ has bounded row sums) we obtain
\[
\frac{1}{n} \sum_j (S_n (\rho) S_n^{-1})_{ij}(\tilde{X}_n)_j \overset{p}{\to} \frac{1}{n} \sum_j (S_n (\rho) S_n^{-1})_{ij}(X_n)_j,
\]
uniformly for all $\rho$ and all $i$, and the middle term which is not dependent on $\rho$ converges as $ \frac{1}{n}  (\tilde{X}^T\tilde{X} -\Omega) \overset{p}{\to}  \frac{1}{n}  X^TX$. 
Therefore for a given $i$, we have
\[
\frac{1}{n} (S_n (\rho) S_n^{-1})_i^T M_{2n} (S_n (\rho) S_n^{-1})_i \overset{p}{\to} \frac{1}{n} (S_n (\rho) S_n^{-1})_i^T M_{1n} (S_n (\rho) S_n^{-1})_i
\]

Taking summation over $i$, we have the desired result for term 3.

Now that convergence of $\hat{\rho}_n$ to $\rho_0$ is established, the convergence of $\hat{\delta}_n$ can be established. We note
\begin{align*}
    \hat{\delta}_n & = (\tilde{X}_n^T \tilde{X}_n - \sum_i \Omega_i)^{-1} \tilde{X}_n^T (I + (\rho_0-\rho)G_n)(X_n\delta_0 + V_n)
\end{align*}
Now from Proposition \ref{wlln}, we have,
\begin{align*}
\frac{1}{n}(\tilde{X}_n^T \tilde{X}_n - \sum_i \Omega_i) & = \frac{1}{n} X_n^TX_n + o_p(1)\\
\frac{1}{n}(\tilde{X}_n^TX_n\delta_0 + (\rho_0 -\rho)\tilde{X}_n^TG_nX_n\delta_0 )& = \frac{1}{n}X_n^TX_n\delta_0 + o_p(1)\\
\frac{1}{n}(\tilde{X}_n^TV_n + (\rho_0 -\rho)\tilde{X}_n^TG_nV_n )& = o_p(1).
\end{align*}
Therefore
\begin{align*}
    \hat{\delta}_n & = (\frac{1}{n}(\tilde{X}_n^T \tilde{X}_n - \sum_i \Omega_i))^{-1} \frac{1}{n}\{\tilde{X}_n^T (I + (\rho_0-\rho)G_n)X_n\delta_0 + \tilde{X}_n^T (I + (\rho_0-\rho)G_n)V_n)\} \nonumber \\
    & = \left(\frac{1}{n} X_n^TX_n + o_p(1)\right)^{-1} \left(\frac{1}{n}X_n^TX_n\delta_0 + o_p(1) + o_p(1)\right) \\
    & = \delta_0 + o_p(1).
\end{align*}

\end{proof}

\subsection{Proof of Theorem 2}

\begin{proof}
\label{proofthm2}

From Taylor expansion with intermediate value theorem, we have
\[
\nabla_{\theta}l^*_n(\hat{\theta}_n,\tilde{X}_n) = \nabla_{\theta}l^*_n(\theta_0,\tilde{X}_n) + \nabla^2_{\theta}l^*_n(\tilde{\theta}_n,\tilde{X}_n)(\hat{\theta}_n - \theta_0),
\]
for some $\tilde{\theta}_n$ which is an intermediate vector between $\hat{\theta}_n$ and $\theta_0$ (therefore $\tilde{\theta}_n \overset{p}{\to} \theta_0$).
Since $\nabla_{\theta}l^*(\hat{\theta}_n,\tilde{X}_n) =0$, by definition of $\hat{\theta}_n$, this implies
\[\sqrt{n} (\hat{\theta}_n - \theta_0) = \left(-\frac{1}{n} \nabla^2_{\theta}l^*_n(\tilde{\theta}_n,\tilde{X}_n)\right)^{-1} \frac{1}{\sqrt{n}} \nabla_{\theta}l^*_n(\theta_0,\tilde{X}_n)
\]
We show the following two convergence in probability results by analyzing the convergence of each element of the corresponding matrices:
\begin{align*}
     \frac{1}{n} \nabla^2_{\theta}l_n^*(\tilde{\theta}_n,\tilde{X}_n) \overset{p}{\to} \frac{1}{n} \nabla^2_{\theta}l_n^*(\theta_0,\tilde{X}_n) \quad \quad \text{ and } \quad  \frac{1}{n} \nabla^2_{\theta}l^*_n(\theta_0,\tilde{X}_n) \overset{p}{\to} I(\theta_0, X_n).
\end{align*}

We start by noting the following two results:
\begin{align*}
    \tilde{V}_n(\theta_0)  = V_n - \eta_n \delta_0 \quad \text{ and    } \quad   L_nY_n  = L_n(S_n^{-1}X_n \delta_0 + S_n^{-1}V_n) = G_nX_n\delta_0 + G_n V_n.
\end{align*}

\subsection*{Convergence of $\frac{1}{n}\nabla^2_{\theta} l_n^* (\tilde{\theta}_n, \tilde{X}_n)$ to $ \frac{1}{n}\nabla^2_{\theta} l_n^* (\theta_0, X_n).$}

In this subsection, we will show the convergence of the observed Hessian matrix of the log-likelihood function evaluated at $\tilde{\theta}_n$, which is $\frac{1}{n}\nabla^2_{\theta} l_n^* (\tilde{\theta}_n, \tilde{X}_n)$. Note, this is a function of $\tilde{X}_n$ since we observed the erroneous version of $X_n$ in practice and constructed our corrected likelihood with it. Therefore, unlike the result in \cite{lee2004asymptotic}, we need to show two convergences - (1) the function converges to its value evaluated at $\theta_0$, and (2) the function with $\tilde{X}_n$ converges to a function with $X_n$. 

Since $\tilde{\theta}_n$ is an intermediate value between $\hat{\theta}_n$ and $\theta_0$, then we have $\tilde{\theta}_n - \theta_0 = o_p(1)$ from Theorem 1. Now we show that as a consequence, each element of the matrix $\frac{1}{n}\nabla^2_{\theta} l_n^* (\tilde{\theta}_n, \tilde{X}_n) - \frac{1}{n}\nabla^2_{\theta} l_n^* (\theta_0, X_n)$ is $o_p(1)$. 

We start with the following result:
\begin{align}
\frac{1}{n}\nabla^2_{\delta,\delta} l_n^* (\tilde{\theta}_n, \tilde{X}_n) - \frac{1}{n}\nabla^2_{\delta,\delta} l_n^* (\theta_0,\tilde{X}_n) = (\frac{1}{\sigma^2_0} - \frac{1}{\tilde{\sigma}^2_n}) [ -\frac{\tilde{X}_n^T\tilde{X}_n}{n}] & = o_p(1)(\frac{X_n^TX_n + \sum_i \Omega_i}{n} + o_p(1)) \nonumber \\
& = o_p(1)O_p(1)=o_p(1),
\end{align}
since $\frac{1}{\sigma^2_0} - \frac{1}{\tilde{\sigma}^2_n}=o_p(1)$ by continuous mapping theorem, and assumptions $\lim_{n \to \infty} \frac{1}{n}X_n^TX_n$ is finite and $\Omega_i$ are finite.

Next since $G_n$ is uniformly bounded in row sum, we have $\mathbb E[|(G_nV_n)_1 (\tilde{X}_n)^T_1| ]< \infty$. Then noting that $\mathbb E[(G_nV_n)_1 (\tilde{X}_n)^T_1] = 0 $, leads to
\[ 
\frac{1}{n} \tilde{X}_n^TG_nV_n \overset{p}{\to} 0.
\]
Further we have seen before from WLLN that $ \frac{1}{n}\tilde{X}_n^TG_nX_n\delta_0 \overset{p}{\to}  \frac{1}{n} X_n^TG_nX_n\delta_0 $.
Therefore,
\[
\frac{1}{n}\tilde{X}_n^TL_nY_n = \frac{1}{n} \tilde{X}_n^TG_nX_n\delta_0 + \frac{1}{n}\tilde{X}_n^TG_nV_n= \frac{1}{n}X_n^TG_nX_n\delta_0 + o_p(1) + o_p(1) = O_p(1).
\]

Consequently,
\begin{align}
\frac{1}{n}\nabla^2_{\delta,\rho} l_n^* (\tilde{\theta}_n) - \frac{1}{n}\nabla^2_{\delta,\rho} l_n^* (\theta_0) = (\frac{1}{\sigma^2_0} - \frac{1}{\tilde{\sigma}^2_n})\frac{1}{n}(\tilde{X}_n^TL_nY_n)=o_p(1)O_p(1)=o_p(1).
\end{align}

Next we analyze $\frac{1}{n}\nabla^2_{\delta,\sigma^2} l_n^* (\tilde{\theta}_n) - \frac{1}{n}\nabla^2_{\delta,\sigma^2} l_n^* (\theta_0) $.
We need the following intermediate results:
\begin{align*}
V_n(\tilde{\theta}_n) & = V_n(\theta_0) -  (\tilde{\rho}_n -\rho_0)L_nY_n  -\tilde{X}_n(\tilde{\delta}_n-\delta_0) \\
\frac{1}{n}\tilde{X}_n^TV_n(\theta_0) & = \frac{1}{n} \tilde{X}_n^TV_n + \frac{1}{n} X_n^T\eta_n\delta_0   -\frac{1}{n}\eta_n^T\eta_n\delta_0 = o_p(1) + o_p(1) + (\frac{\sum_i\Omega}{n} \delta_0 + o_p(1)) = O_p(1)\\
\frac{1}{n}\tilde{X}_n^T \tilde{X}_n &  = \frac{1}{n} X_n^TX_n + \frac{1}{n} \sum_i \Omega_i + o_p(1) = O_p(1). 
\end{align*}
Then 
\begin{align*}
\frac{1}{n}\tilde{X}_n^T(\frac{V_n(\tilde{\theta}_n)}{\tilde{\sigma}^4} - \frac{V_n(\theta_0)}{\sigma^4_0}) & =  (\frac{1}{\tilde{\sigma}^4_n} -\frac{1}{\sigma^4_0})\frac{1}{n}\tilde{X}_n^TV_n(\theta_0)  - \frac{1}{\tilde{\sigma}^2_n} \frac{1}{n}\tilde{X}_n^T\tilde{X}_n ((\tilde{\delta}_n-\delta_0) -  \frac{1}{\tilde{\sigma}^2_n} \frac{1}{n} \tilde{X}_n^TL_nY_n  (\tilde{\rho}_{n} -\rho_0) \\
& = o_p(1)O_p(1) + O_p(1)O_p(1)o_p(1) +O_p(1)O_p(1)o_p(1) = o_p(1),
\end{align*}
since $\frac{1}{\tilde{\sigma}_n^2} = \frac{1}{\sigma_0^2} + o_p(1)$, $\tilde{\rho}_n = \rho_0 +o_p(1)$, $\tilde{\delta}_n = \delta_0 +o_p(1)$ and ${\sigma}_0, \rho_0, \delta_0$ are finite values. 
\begin{align}
\frac{1}{n}\nabla^2_{\delta,\sigma^2} l_n^* (\tilde{\theta}_n) - \frac{1}{n}\nabla^2_{\delta,\sigma^2} l_n^* (\theta_0) & = \frac{1}{n}(\sum_i \Omega_i)(\frac{\tilde{\delta}_n}{\tilde{\sigma}_n^4}-\frac{\delta_0}{\sigma_0^4}) -
\frac{1}{n}\{\tilde{X}_n^T(\frac{V_n(\tilde{\theta}_n)}{\tilde{\sigma}^4} - \frac{V_n(\theta_0)}{\sigma^4_0})\} \nonumber \\
& = O_p(1)o_p(1) + o_p(1) =o_p(1). 
\end{align}

For $\frac{1}{n}\nabla^2_{\sigma^2,\sigma^2} l_n^* (\tilde{\theta}_n) - \frac{1}{n}\nabla^2_{\sigma^2,\sigma^2} l_n^* (\theta_0)$ we note the following results:
\begin{align*}
    \frac{1}{n} V_n(\tilde{\theta}_n)^TV_n(\tilde{\theta}_n) & = \frac{1}{n} V_n(\theta_0)^TV_n(\theta_0) + (\tilde{\delta}_n-\delta_0)^T  \frac{\tilde{X}_n^T\tilde{X}_n}{n}(\tilde{\delta}_n-\delta_0) + (\tilde{\rho} -\rho_0)^2 \frac{(L_nY_n)^T(L_nY_n)}{n} \\
    & \quad + 2 (\tilde{\rho} -\rho_0) (\tilde{\delta}_n-\delta_0)^T  \frac{\tilde{X}_n^TL_nY_n}{n} +  2  (\tilde{\delta}_n-\delta_0)^T  \frac{\tilde{X}_n^TV_n(\theta_0)}{n} + 2 (\tilde{\rho} -\rho_0) \frac{(L_nY_n)^TV_n(\theta_0)}{n} \\
    & = \frac{1}{n} V_n(\theta_0)^TV_n(\theta_0)  + o_p(1).
\end{align*}
Note, $\tilde{V}_n(\theta_0) = V_n - \eta_n\delta_0.$ Therefore,
\begin{align*}
\frac{1}{n}\tilde{V}_n(\theta_0)^T \tilde{V}_n(\theta_0) & = (V_n - \eta_n\delta_0)^T (V_n - \eta_n \delta_0)\\
& = V_n^TV_n + \delta_0^T\eta_n^T\eta_n \delta_0 + o_p(1) \\
& =\sigma_0^2 +o_p(1) + o_p(1) + \frac{1}{n} \delta_0^T(\sum_i\Omega_i) \delta_0 + o_p(1) = O_p(1). 
\end{align*}
Therefore,
\begin{align}
& \frac{1}{n}\nabla^2_{\sigma^2,\sigma^2} l_n^* (\tilde{\theta}_n) - \frac{1}{n}\nabla^2_{\sigma^2,\sigma^2} l_n^* (\theta_0) \nonumber \\ & \quad  = \frac{1}{2}(\frac{1}{\tilde{\sigma}_n^4}-\frac{1}{\sigma_0^4}) - \frac{1}{2}(\frac{1}{\tilde{\sigma}_n^6}-\frac{1}{\sigma_0^6}) [\frac{1}{n} V_n(\theta_0)^TV_n(\theta_0) +(\tilde{\delta}_n^T \sum_i (\Omega_i) \tilde{\delta}_n - \delta_0^T \sum_i (\Omega_i) \delta_0) ] \nonumber \\
& \quad =o_p(1) + o_p(1)[O_p(1) +o_p(1)] = o_p(1).
\end{align}

Next
\begin{align}
\frac{1}{n}\nabla^2_{\sigma^2,\rho} l_n^* (\tilde{\theta}_n) - \frac{1}{n}\nabla^2_{\sigma^2,\rho} l_n^* (\theta_0) & =\frac{1}{\tilde{\sigma}^4}\frac{1}{n}((\tilde{\delta}_n-\delta_0)^{T}\tilde{X}_n^{T}L_{n}Y_n+ \frac{1}{\tilde{\sigma}^4}\frac{1}{n}(Y_n^TL_n^TL_nY_n(\tilde{\rho}_n-\rho_0) \nonumber \\
& \quad - (\frac{1}{\tilde{\sigma}^4}-\frac{1}{\sigma_0^4})\frac{1}{n}Y_n^TL_n^TV_n(\theta_0)  \\
& =O_p(1)O_p(1)o_p(1) + O_p(1)O_p(1)o_p(1) +o_p(1)O_p(1)=o_p(1) \nonumber.
\end{align}

For the final result from the intermediate value theorem, we have
\[
\frac{1}{n}(\tr(G_n(\tilde{\rho_{n}})^2) - \tr(G_n(\rho_0)^2) = 2 \frac{1}{n}\tr(G_n(\bar{\rho})^3)(\tilde{\rho_{n}}-\rho_0) = \frac{1}{n}O(n/h_n)o_p(1)=o_p(1), 
\]
and $\frac{1}{n}(Y_nL_n)^T(Y_nL_n)=\frac{1}{n}O_p(n/h_n)= o_p(1)$ under the assumptions by results of Lemma A.3 of \cite{lee2004asymptotic}. Then
\begin{align}
\frac{1}{n}\nabla^2_{\rho,\rho} l_n^* (\tilde{\theta}_n) - \frac{1}{n}\nabla^2_{\rho,\rho} l_n^* (\theta_0) =o_p(1).
\end{align}

\subsection* {Convergence of $\frac{1}{n} \nabla^2_{\theta}l_n^*(\theta_0, \tilde{X}_n)$
to $I (\theta_0, X_n) = \mathbb E^+ \mathbb E^*[\nabla^2_{\theta}l_n^*(\theta_0, \tilde{X}_n)]$ } 

Most of the convergence results for the elements follow from the convergence results used in the earlier proofs. The first element,
\begin{align*}
\frac{1}{n}\nabla^2_{\delta,\delta} l_n^* (\theta_0, \tilde{X}_n) - E[\nabla^2_{\delta,\delta} l_n^* (\theta_0, \tilde{X}_n)] = \frac{1}{\sigma_0^2}[\frac{1}{n}(\tilde{X}_n^T\tilde{X}_n - \sum \Omega_i) - \frac{1}{n} X_n^TX_n) = o_p(1).
\end{align*}
The second element using the result in Lemma \ref{wlln},
\begin{align*}
\frac{1}{n}\nabla^2_{\delta,\rho} l_n^* (\theta_0, \tilde{X}_n) - E[\nabla^2_{\delta,\rho} l_n^* (\theta_0, \tilde{X}_n)] & = \frac{1}{\sigma_0^2}[\frac{1}{n}(\tilde{X}_n^TL_nY_n)  - \frac{1}{n} (X_n^TG_nX_n\delta_0) \\
& = \frac{1}{n} \tilde{X}_n^TG_nX_n\delta_0 - \frac{1}{n} (X_n^TG_nX_n\delta_0) +\frac{1}{n}\tilde{X}_n^TG_nV_n \\
& = o_p(1) + o_p(1)= o_p(1).
\end{align*}

For the third element, we have
\begin{align*}
\frac{1}{n}\nabla^2_{\delta,\sigma^2} l_n^* (\theta_0, \tilde{X}_n) - E[\nabla^2_{\delta,\sigma^2} l_n^* (\theta_0, \tilde{X}_n)] & = \frac{1}{\sigma_0^4}[\frac{1}{n}(\tilde{X}_n^T\tilde{V}_n(\theta_0) + \sum_i \Omega_i\delta_0)  \\
& = \frac{1}{\sigma_0^4}[\frac{1}{n} \tilde{X}_n^TV_n + \frac{1}{n} X_n^T\eta_n\delta_0   -\frac{1}{n}\eta_n^T\eta_n\delta_0  + \frac{\sum_i \Omega_i\delta_0}{n} \\
& = o_p(1) + o_p(1).
\end{align*}

For the fourth element, we have
\begin{align*}
\frac{1}{n}\nabla^2_{\sigma^2,\sigma^2} l_n^* (\theta_0, \tilde{X}_n) - E[\nabla^2_{\sigma^2,\sigma^2} l_n^* (\theta_0, \tilde{X}_n)] & = \frac{1}{\sigma_0^6}[\frac{1}{n}(\tilde{V}_n(\theta_0)^T\tilde{V}_n(\theta_0) - \delta_0^T(\sum_i \Omega_i)\delta_0) -\sigma_0^2]  \\
& = \frac{1}{\sigma_0^6}[\frac{1}{n}V_n^TV_n + \frac{1}{n}\delta_0^T\eta_n^T\eta_n \delta_0 + o_p(1)  - \frac{\delta_0^T(\sum_i \Omega_i)\delta_0)}{n} -\sigma_0^2] \\
& = \frac{1}{\sigma_0^6}[\sigma_0^2 + o_p(1) +\frac{\delta_0^T(\sum_i \Omega_i)\delta_0)}{n}+ o_p(1) - \frac{\delta_0^T(\sum_i \Omega_i)\delta_0)}{n} -\sigma_0^2] \\
& = o_p(1) + o_p(1).
\end{align*}
Now, for the fifth term, we have
\begin{align*}
\frac{1}{n}\nabla^2_{\sigma^2,\rho} & l_n^* (\theta_0, \tilde{X}_n) - E[\nabla^2_{\sigma^2,\rho} l_n^* (\theta_0, \tilde{X}_n)]  = \frac{1}{\sigma_0^4}[\frac{1}{n}\tilde{V}_n^T(G_nX_n\delta_0 + G_n V_n) -\frac{\sigma_0^2}{n} \tr(G_n)]  \\
& = \frac{1}{\sigma_0^4} [\frac{1}{n}V_n^TG_nX_n\delta_0 + V_n^TG_n V_n - \delta_0^T\eta_n^TG_nX_n\delta_0 - \delta_0^T\eta_n^TG_n V_n)-\frac{\sigma_0^2}{n} \tr(G_n)]\\
& = o_p(1).
\end{align*}
For the sixth and the final term, first notice
\begin{align*}
\frac{1}{n}Y_n^TL_n^TL_nY_n &= \frac{1}{n}(G_nX_n\delta_0 + G_nV_n)^T(G_nX_n\delta_0 + G_nV_n) \\
& = \frac{1}{n}( (G_nX_n\delta_0)^TG_nX_n\delta_0 + 2 (G_nX_n\delta_0)^TG_nV_n + (G_nV_n)^T(G_nV_n))\\
& = \frac{1}{n}H_n^TH_n +o_p(1) +\frac{\sigma_0^2}{n}\tr(G_n^TG_n) + o_p(1).
\end{align*}
\begin{align*}
\frac{1}{n}\nabla^2_{\rho,\rho} & l_n^* (\theta_0, \tilde{X}_n) - E[\nabla^2_{\rho,\rho} l_n^* (\theta_0, \tilde{X}_n)] \\
& = \frac{1}{\sigma_0^2}[\frac{1}{n}H_n^TH_n +o_p(1) + \frac{\sigma_0^2}{n}\tr(G_n^TG_n) + o_p(1) +\frac{\sigma_0^2}{n} \tr(G_nG_n) -\frac{1}{n}H_n^TH_n 
 - \frac{\sigma_0^2}{n} \tr(G_n^TG_n + G_nG_n)]  \\
& =o_p(1).
\end{align*}

\subsection*{Limiting distribution of the scaled corrected score vector}
The corrected score vector, i.e., the gradient of the corrected log-likelihood function $\nabla_{\theta} l_n^{*}(\theta)$ evaluated at the true parameter vector $\theta_0$ is as follows:
\begin{align*}
    \nabla_{\delta} l_n^* (\theta_0) &=  \frac{1}{\sigma_0^2}[X_n^T V_n + \eta_n^TV_n  - X_n^T\eta_n \delta_0 - \eta_n^T\eta_n\delta_0 + \sum_{i}\Omega_{i}\delta_0 ],\\
\nabla_{\rho} l^* (\theta_0)&  = \frac{1}{\sigma_0^2} [(G_nX_n\delta_0)^TV_n+ V_n^TG_nV_n -(G_nX_n\delta_0)^T\eta_n \delta_0 - V_n^TG_n \eta_n \delta_0 ] - \tr (G_n)\\
\nabla_{\sigma^2} l^* (\theta_0)& = -\frac{n}{2 \sigma_0^2} + \frac{1}{2\sigma_0^4} \{V_n^TV_n- V_n^T\eta_n\delta_0  +\delta_0^T\eta_n^T\eta_n\delta_0 - \delta_0^T (\sum_i \Omega_i) \delta_0\}.
\end{align*}

We first verify that the unconditional expectation of the gradient vector is 0. Note that $\eta_n$ is a $n \times d$ matrix whose rows $(\eta_n)_i$ are mean 0 independent random vectors. Since $(\eta_n)_i$ and $(V_n)_i$ are independent, we have $\mathbb E[\eta_n^TV_n]=0$ and $\mathbb E[V_n^TG_n\eta_n]=0$. Further, $\mathbb E[V_n^TG_nV_n] = \sigma^2 \tr(G_n)$. Finally, $\mathbb E[\eta_n^T \eta_n]= \sum_i \Omega_i$, by definition. With these results, we have,
\begin{align*}
    \mathbb E[\nabla_{\delta} l_n^* (\theta_0)] &=  \frac{1}{\sigma_0^2}[ -\mathbb E[ \eta_n^T\eta_n \delta_0 + \sum_{i}\Omega_{i}\delta_0 ] =0,\\
\mathbb E[\nabla_{\rho} l^* (\theta_0)]& = \frac{1}{\sigma_0^2} [(G_nX_n\delta_0)^T\mathbb E[V_n]+ \sigma_0^2 tr(G_n) -(G_nX_n\delta_0)^T\mathbb E[\eta_n] \delta_0 - \mathbb E[V_n^TG_n \eta_n] \delta_0 ] - tr (G_n) =0\\
\mathbb E[\nabla_{\sigma^2} l^* (\theta_0)] & = -\frac{n}{2 \sigma_0^2} + \frac{1}{2\sigma_0^4} \{n \sigma_0^2- \mathbb E[V_n^T\eta_n]\delta_0  +\delta_0^T\mathbb E[\eta_n^T\eta_n]\delta_0 -  \delta_0^T (\sum_i \Omega_i) \delta_0\}=0.
\end{align*}

The unconditional variance is given by,
\[
var[\frac{1}{\sqrt{n}}\nabla_{\theta}l_n^*(\theta_0)] = \mathbb E[( \frac{1}{\sqrt{n}} \nabla_{\theta}l_n^*(\theta_0))( \frac{1}{\sqrt{n}} \nabla_{\theta}l_n^*(\theta_0))^T] = \Sigma(\theta_0, X),
\]
is finite since the elements of the matrix are functions of moments of the bounded random vectors $(\eta_n)_i$, and second, third and fourth moments of $(V_n)_i$, all of which are finite by assumptions.

We note that the elements of the score vector consist of linear and quadratic forms on $V_n$ and $\eta_n$. The random variables $(\eta_n)_i, (V_n)_i, (\eta_n)_i(V_n)_i$ are independent and their means are 0. The matrix for quadratic form $G_n$ is uniformly bounded in row sums, the elements of the vector $G_nX_n\delta_0$ are uniformly bounded. Finally $\sup E|(\eta_n)_i|^{4+\epsilon}<\infty$, $\sup E|(V_n)_i|^{4+\epsilon}<\infty$ by assumptions. This also implies $\sup E|(\eta_n)_i(V_n)_i|^{4+\epsilon}<\infty$
therefore, the central limit theorem for linear and quadratic forms in \cite{kelejian1999generalized} and \cite{lee2004asymptotic} can be applied and we conclude
\[
\sqrt{n}(\frac{1}{n} \nabla_{\theta}l_n^*(\theta_0)) \overset{d}{\to} N(0, \Sigma(\theta_0,X))
\]

\end{proof}

Then, combining the three subparts, we arrive at the stated conclusion.

\subsection{Proof of Lemma 3}
The Lemma can be proved following the ideas in \cite{carroll2006measurement} Appendix B.6. We note that
\[\sqrt{n} \left(-\frac{1}{n} \nabla^2_{\theta}l^*_n(\tilde{\theta}_n,\tilde{X}_n)\right) 
 (\hat{\theta}_n - \theta_0) = \frac{1}{\sqrt{n}} \nabla_{\theta}l^*_n(\theta_0,\tilde{X}_n) + D_n(\theta,\Omega) (\hat{\omega} - \omega).
\]

First note that since $\tilde{X}_n$ and $\hat{\Omega}$ are independent, the RHS converges to a normal distribution limit with a limiting covariance matrix of
\[
\Sigma(\theta_0,X) + D(\theta_0, \Omega)C(\Omega) D(\theta_0,\Omega)^T,
\]
where $D(\theta_0,\Omega)$ and $C(\Omega)$ are as given in the statement of the lemma. In the previous theorem, we have already shown that $-\frac{1}{n} \nabla^2_{\theta}l^*_n(\tilde{\theta}_n,\tilde{X}_n) \to I(\theta_0,X)$ in probability. Therefore, combining the two results, we have the result claimed in the lemma.

\subsection{Miscellaneous results}

\subsubsection{$M_2^TM_2-K=M_2$}
\begin{proof}

\label{MTM}
\begin{align*}
M_2^TM_2 - K & = I - 2 \tilde{X}(\tilde{X}^T\tilde{X} - (\sum_i \Omega_i))^{-1} \tilde{X}^T + \tilde{X}(\tilde{X}^T\tilde{X} - (\sum_i \Omega_i))^{-1} \tilde{X}^T \tilde{X}(\tilde{X}^T\tilde{X} - (\sum_i \Omega_i))^{-1} \tilde{X}^T\\
& \quad - \tilde{X}(\tilde{X}^T\tilde{X} - (\sum_i \Omega_i))^{-1} (\sum_i \Omega_i) (\tilde{X}^T\tilde{X} - (\sum_i \Omega_i))^{-1}\tilde{X}^T \\
& = I - \tilde{X}(\tilde{X}^T\tilde{X} - (\sum_i \Omega_i))^{-1} \tilde{X}^T = M_2.
\end{align*}
\end{proof}
\subsubsection{First and second derivative of $l^*(\rho)$ with symmetric $L$}

\begin{proof}

\label{rhoderiv}
Let $\lambda_1, \ldots \lambda_n$ be the eigenvalues of the symmetric matrix $L$. The first and second derivatives of the corrected concentrated log-likelihood function $l^*(\rho)$ with respect to $\rho$ are given by:
\begin{align*}
\frac{\partial l^*}{\partial \rho} & =\frac{2}{n} \sum_{i=1}^n \frac{\lambda_i}{1-\rho\lambda_i} + 2\frac{\rho L^TY^TMYL-Y^TMYL}{\hat{\sigma}^2}\\
\frac{\partial^2 l^*}{\partial \rho} & =\frac{2}{n}*\sum_{i=1}^{n}\frac{\lambda_i^2}{(1-\rho \lambda_i)^2} + \frac{2L^TY^TMYL}{\hat{\sigma}^2}- \frac{4(\rho L^TY^TMYL-Y^TMYL)^2}{(\hat{\sigma}^2)^2}
\end{align*}
\end{proof}

\subsubsection{The gradient and Hessian of log-likelihood function}
\begin{proof}

\label{gradhess}

\begin{align*}
    \nabla_{\delta} l_n^* (\theta) &= \frac{1}{\sigma^2} [\tilde{X}_n^{T}(S_n(\rho)Y_n- \tilde{X}_n \delta) + \sum_i (\Omega_i) \delta],\\
\nabla_{\rho} l_n^* (\theta)&=\frac{1}{2\sigma^2}[2(L_{n}Y_{n})^T(\tilde{V}_n(\rho)) ] +  \frac{1}{|S_n(\rho)|}|S_n(\rho)|\tr (-S_n(\rho)^{-1}L_n) \\
 & \quad  \quad =\frac{1}{\sigma^2} [(L_{n}Y_{n})^T(\tilde{V}_n(\rho)) ] - \tr(G_n(\rho)),\\
\nabla_{\sigma^2} l_n^* (\theta)& = -\frac{n}{2 \sigma^2} + \frac{1}{2\sigma^4} \{(S_n(\rho)Y_n- \tilde{X}_n \delta)^T(S_n(\rho)Y_n- \tilde{X}_n \delta)-  \delta^T (\sum_i \Omega_i)\delta\},
\end{align*}

The elements of the Hessian matrix are given by:
\begin{align*}
    \nabla^2_{\delta,\delta} l_n^* (\theta) &= -\frac{1}{\sigma^2} [ -\tilde{X}_n^T\tilde{X}_n  +\sum_i \Omega_i],\\
\nabla_{\delta,\rho} l_n^* (\theta)&=-\frac{1}{\sigma^2}[\tilde{X}_n^TL_nY_n]\\
\nabla_{\delta,\sigma^2}l_n^* (\theta)& = - \frac{1}{\sigma^4} \{\tilde{X}_n^T(S_n(\rho)Y_n -\tilde{X}_n\delta) + (\sum_i \Omega_i) \delta\}\\
\nabla_{\sigma^2,\sigma^2}l_n^* (\theta)& =\frac{n}{2\sigma^4} -\frac{1}{\sigma^6}\{(S_n(\rho)Y_n -\tilde{X}_n\delta)^T(S_n(\rho)Y_n -\tilde{X}_n\delta) +  \delta^T (\sum_i \Omega_i)\delta\}\\
\nabla_{\sigma^2,\rho}l_n^* (\theta)& =-\frac{1}{\sigma^4}(S_n(\rho)Y_n -\tilde{X}_n\delta)^TL_nY_n\\
\nabla_{\rho,\rho}l_n^* (\theta)& =-\frac{1}{\sigma^2}(L_nY_n)^TL_nY_n - tr(\frac{\partial S_{n}(\rho)^{-1}L_{n}}{\partial \rho})\\
&=-\frac{1}{\sigma^2}(L_nY_n)^TL_nY_n + \tr(S_{n}(\rho)^{-1}\frac{\partial S_{n}(\rho)}{\partial \rho}S_{n}(\rho)^{-1}L_{n})\\
&=-\frac{1}{\sigma^2}(L_nY_n)^TL_nY_n - \tr(G_n(\rho)G_n(\rho)).
\end{align*}

Evaluating the negatives of these elements at the true parameter value $\theta_0 = \{\delta_0,\rho_0,\sigma_0^2\}$, we obtain the \textit{corrected observed} Fisher information matrix $\hat{I}(\theta_0, \tilde{X}_n)$ as follows:

\begin{align*}
     \frac{1}{n}\begin{pmatrix}
    \frac{1}{\sigma_0^2}(\tilde{X}_n^T\tilde{X}_n - \sum_i \Omega_i) & \frac{1}{\sigma_0^2}\tilde{X}_n^TL_nY_n& \frac{1}{\sigma_0^4} (\tilde{X}_n^T\tilde{V}_n(\theta_0) + \sum_i \Omega_i \delta_0)
    \\
\frac{1}{\sigma_0^2}Y_n^TL_n^T\tilde{X}_n  &  \frac{1}{\sigma_0^2} (L_nY_n)^T(L_nY_n) +  tr (G_n G_n) &  \frac{1}{\sigma_0^4} (\tilde{V}_n(\theta_0)^TL_nY_n \\
    \frac{1}{\sigma_0^4} (\tilde{V}_n(\theta_0)^T\tilde{X}_n +\delta_0^T \sum_i \Omega_i ) & \frac{1}{\sigma_0^4} (\tilde{V}_n(\theta_0)^TL_nY_n & \frac{n}{2\sigma_0^4}-\frac{1}{\sigma_0^6}(\tilde{V}_n(\theta_0)^T\tilde{V}_n(\theta_0) + \delta_0^T \sum_i \Omega_i \beta_0)
    \end{pmatrix}
\end{align*}
This matrix can also be used as an empirical estimator for the matrix $I(\theta_0,X_n)$.
\end{proof}

\newpage

\subsection{Additional Tables and Figures}

\begin{figure}[h] 
  \caption{Comparing the standard error with standard deviation of the estimates}
  \centering
  \label{incratioseappendix}
   \begin{minipage}[b]{0.5\linewidth}
    %\centering
    \includegraphics[width=\linewidth]{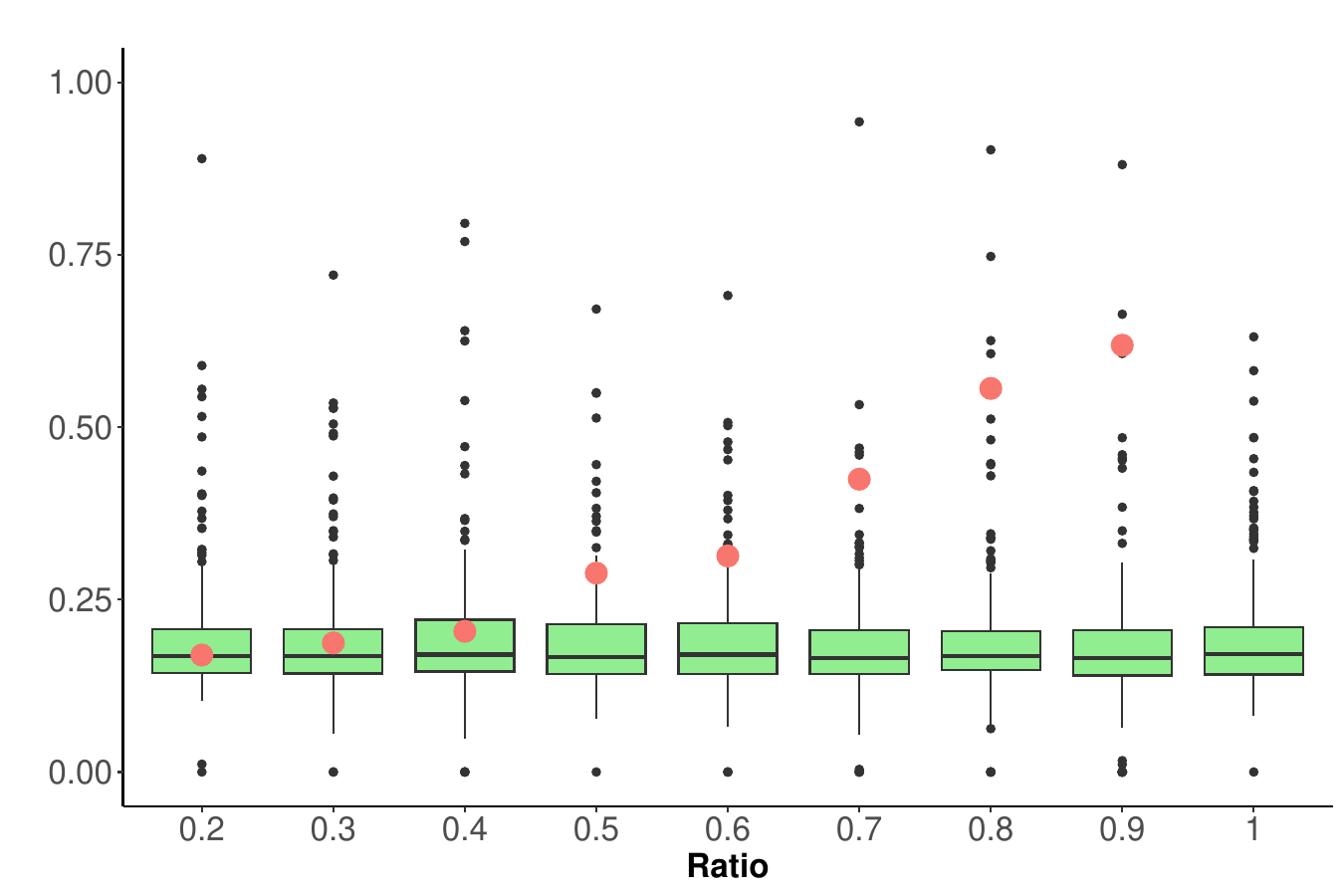}
    \caption*{a. $\hat{\beta_{1}}$} 
  \end{minipage}%%
     \begin{minipage}[b]{0.5\linewidth}
    %\centering
    \includegraphics[width=\linewidth]{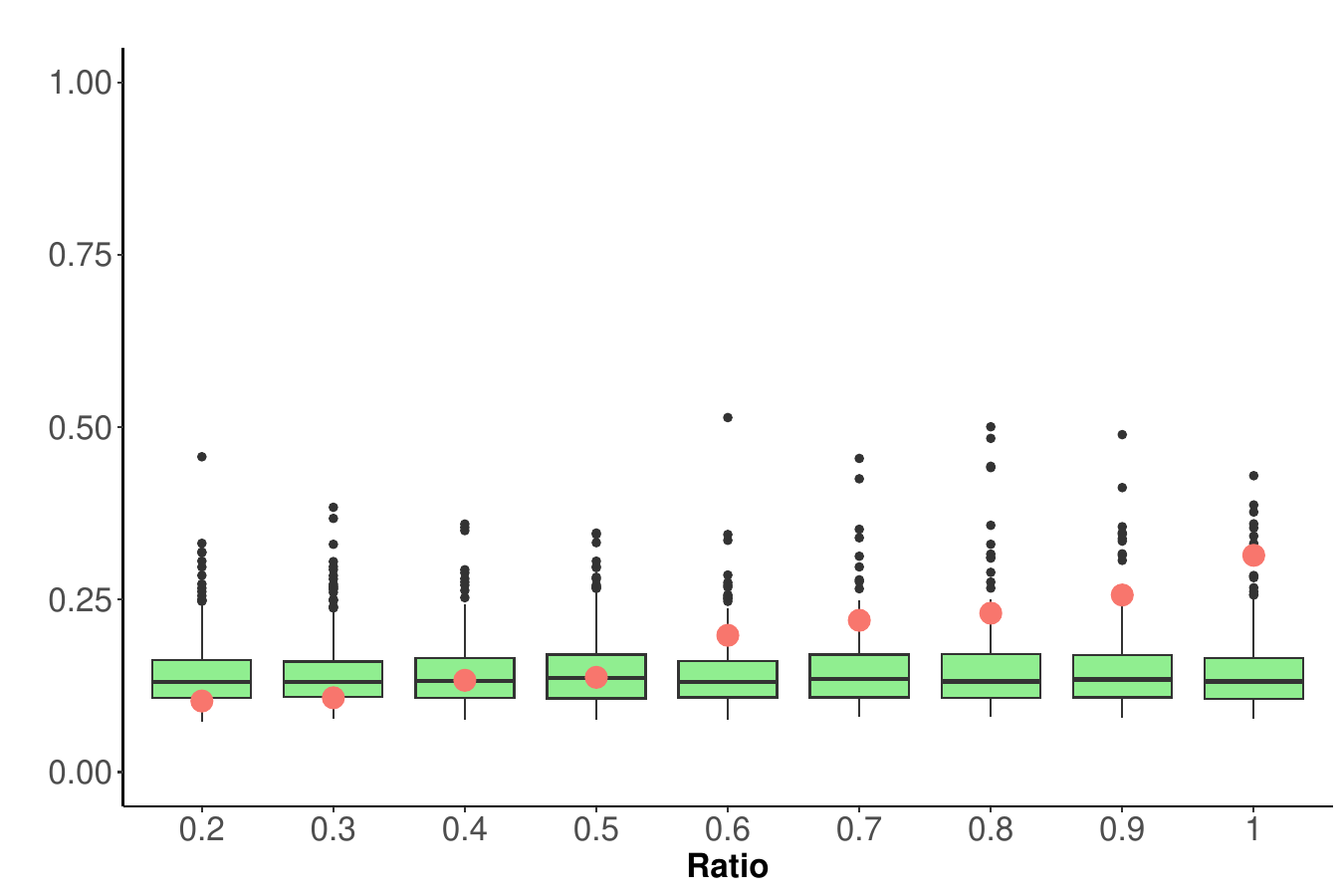}
    \caption*{b. $\hat{\beta_{1}}$} 
  \end{minipage}
     \begin{minipage}[b]{0.5\linewidth}
    %\centering
    \includegraphics[width=\linewidth]{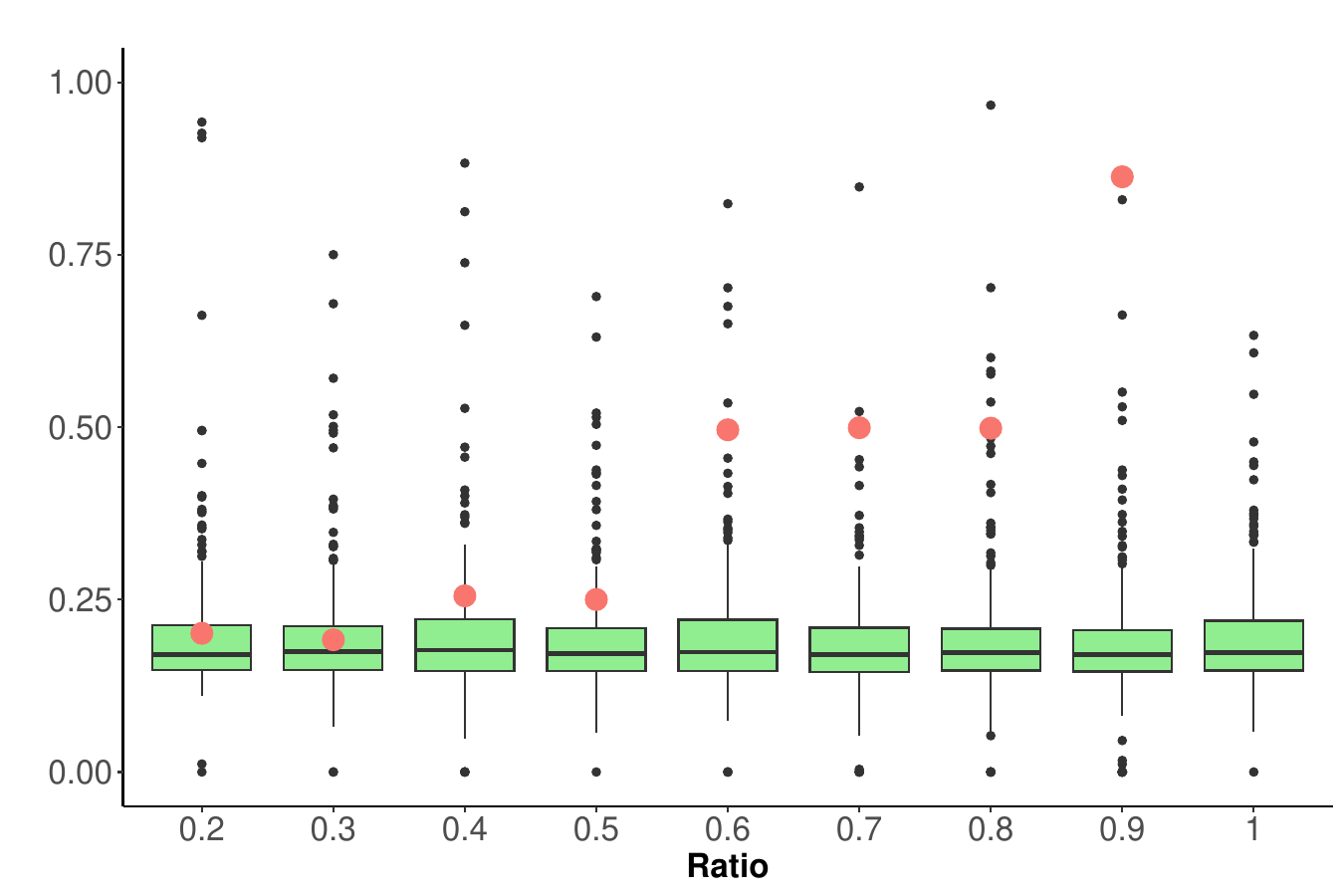}
    \caption*{c. $\hat{\beta_{2}}$} 
  \end{minipage}%%
     \begin{minipage}[b]{0.5\linewidth}
    %\centering
    \includegraphics[width=\linewidth]{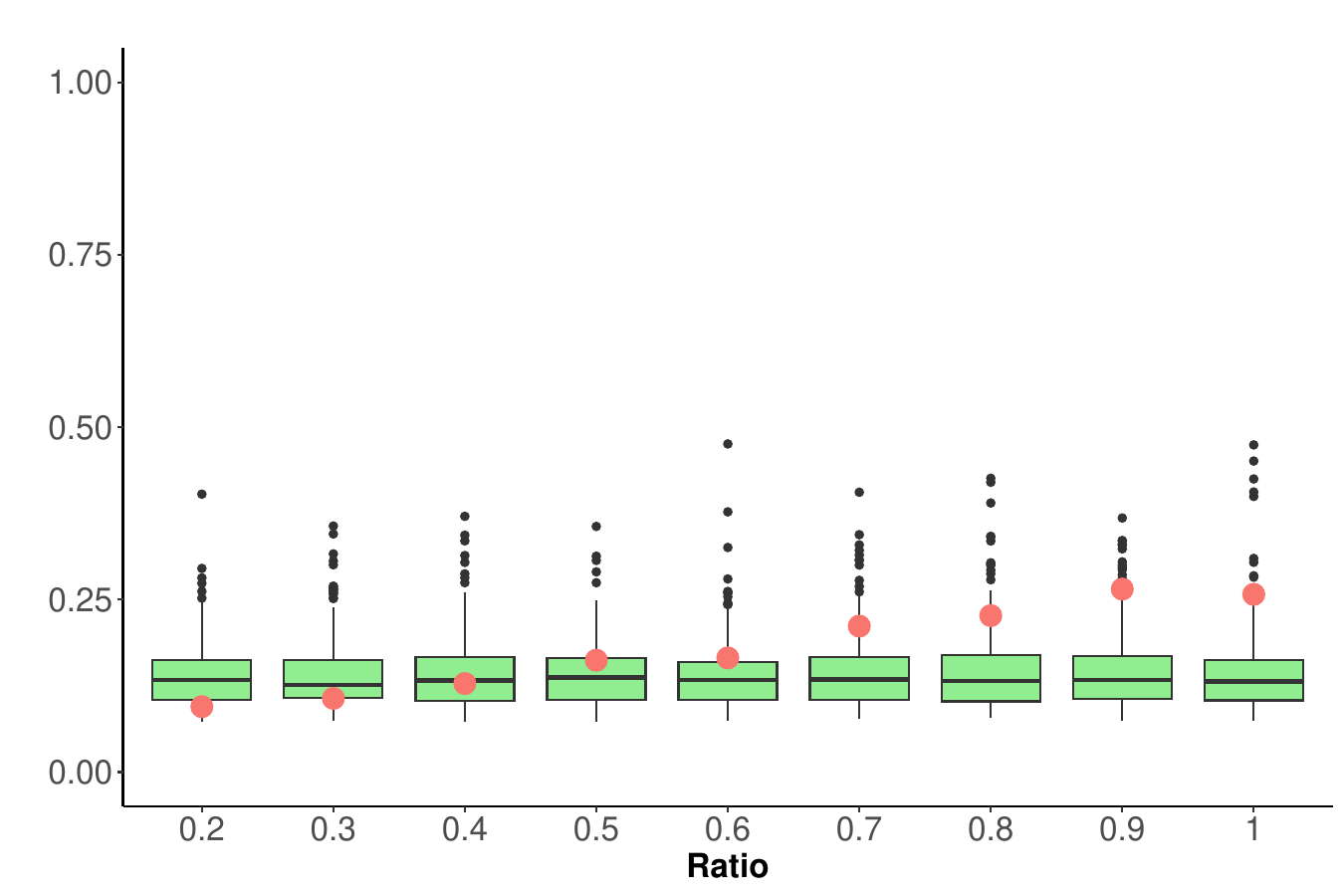}
    \caption*{d. $\hat{\beta_{2}}$} 
  \end{minipage}
     \begin{minipage}[b]{0.5\linewidth}
    %\centering
    \includegraphics[width=\linewidth]{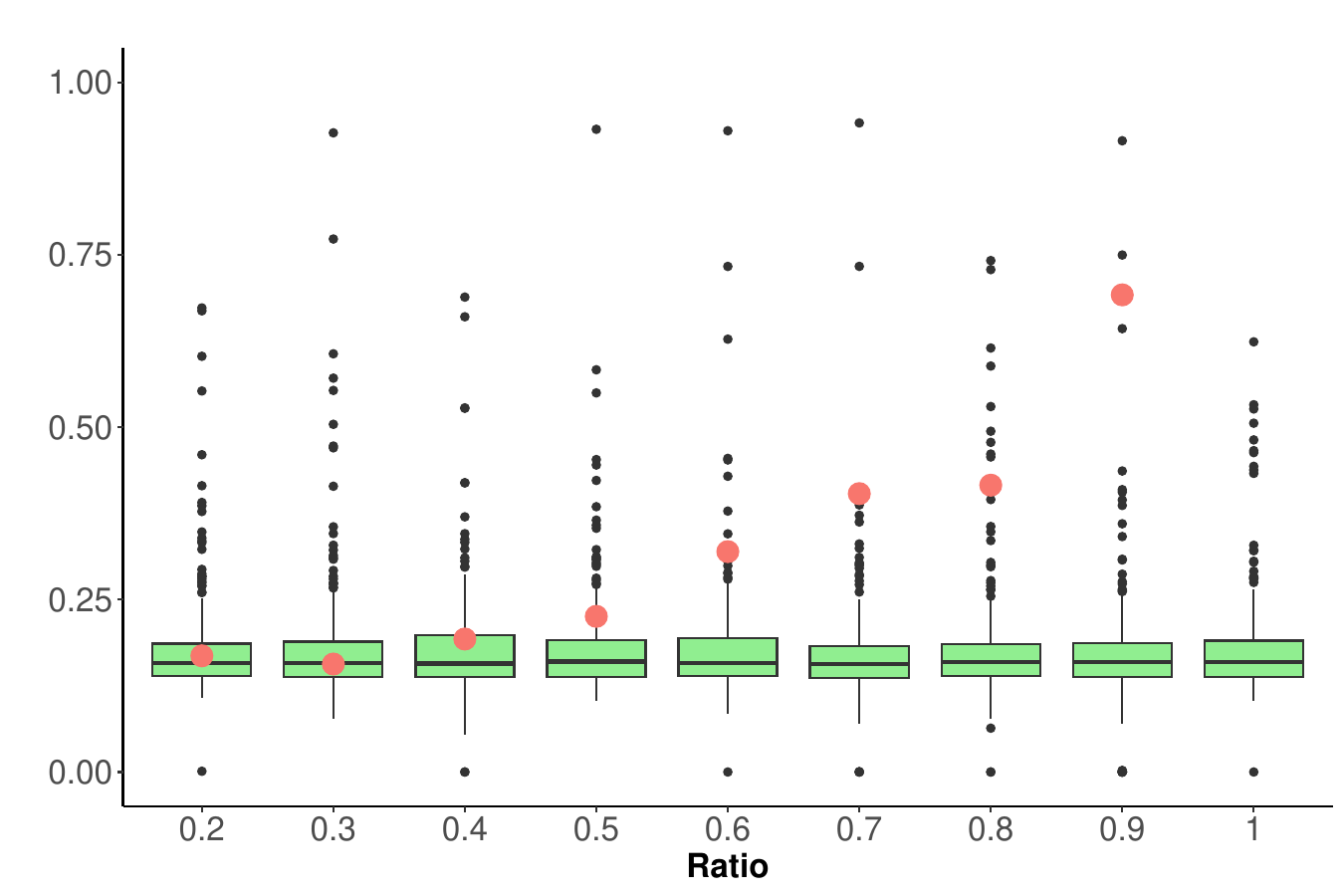}
    \caption*{e. $\hat{\gamma_{2}}$} 
  \end{minipage}%%
     \begin{minipage}[b]{0.5\linewidth}
    %\centering
    \includegraphics[width=\linewidth]{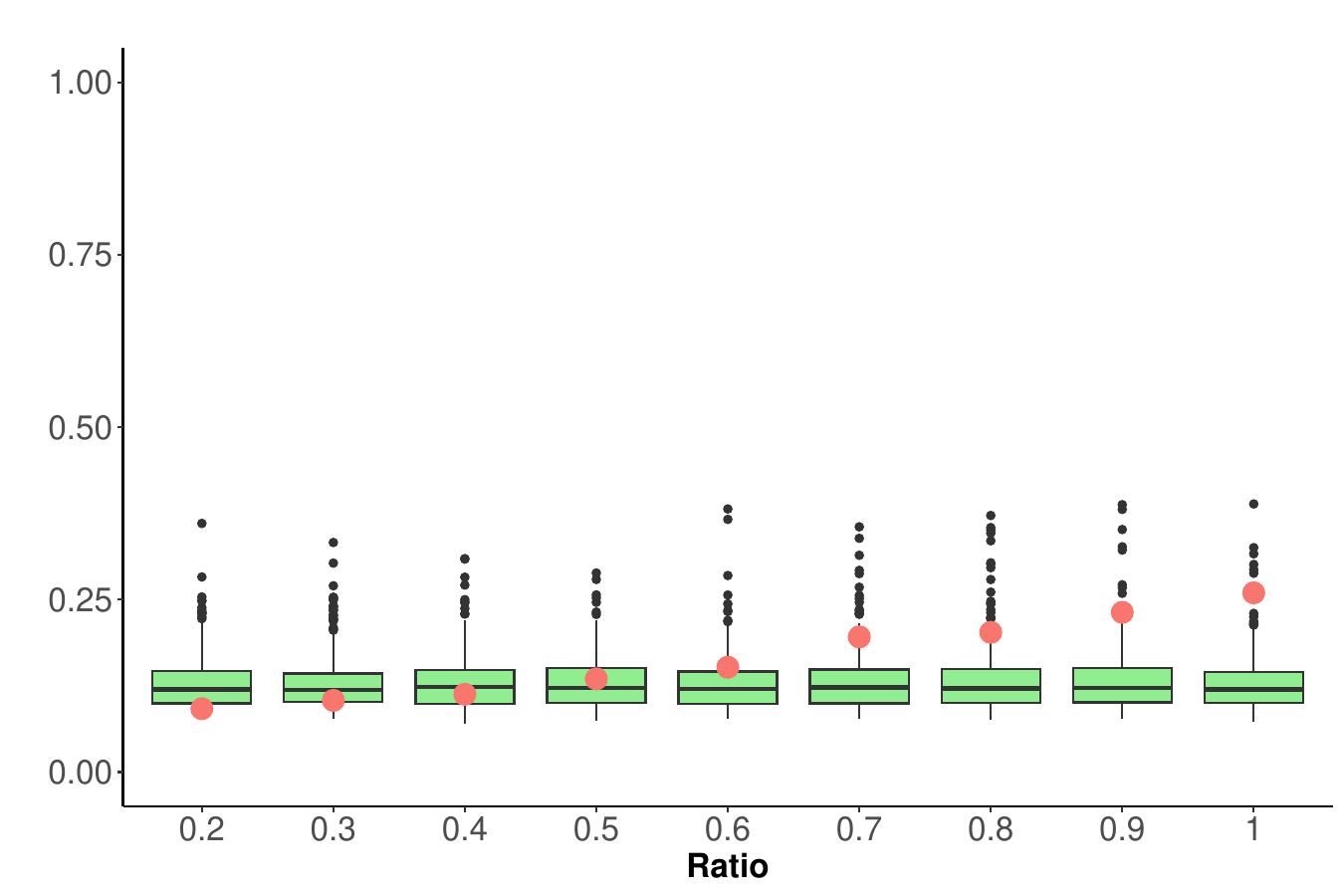}
    \caption*{f. $\hat{\gamma_{2}}$} 
  \end{minipage}
  \begin{minipage}{13.5 cm}{\footnotesize{Notes: We compare the standard error using the solution in Theorem \ref{asymptotic} over 300 simulations to the empirical standard deviation of the estimates highlighted in red. The number of nodes is kept at 200 (500) for all the simulation scenarios in the left (right) panel.}}
\end{minipage} 
\end{figure}

\begin{table}[h] \centering \renewcommand*{\arraystretch}{0.8}\caption{Summary Statistics}
\label{summarytableconflict}
\begin{tabular}{p{6cm}p{1cm}p{1cm}p{1cm}p{1cm}p{1cm}}
\hline
\hline
Variable & \multicolumn{1}{c}{N} & \multicolumn{1}{c}{Mean} & \multicolumn{1}{c}{St. Dev.} & \multicolumn{1}{c}{Min} & \multicolumn{1}{c}{Max} \\ 
\hline \\[-1.8ex] 
GPA & 7,872 & 3.188 & 0.605 & 0.268 & 4.000 \\ 
Age & 7,872 & 12.038 & 1.015 & 5 & 19 \\ 
5th grade & 7,872 & 0.049 & 0.216 & 0 & 1 \\ 
6th grade & 7,872 & 0.264 & 0.441 & 0 & 1 \\ 
7th grade & 7,872 & 0.339 & 0.473 & 0 & 1 \\ 
8th grade & 7,872 & 0.348 & 0.476 & 0 & 1 \\ 
Gender (Boy) & 7,872 & 0.497 & 0.500 & 0 & 1 \\ 
White ethnicity & 7,872 & 0.654 & 0.476 & 0 & 1 \\ 
Black ethnicity & 7,872 & 0.075 & 0.264 & 0 & 1 \\ 
Hispanic & 7,872 & 0.184 & 0.388 & 0 & 1 \\ 
Mother went to college & 7,872 & 0.749 & 0.434 & 0 & 1 \\ 
Yes: lives with just mom & 7,872 & 0.122 & 0.327 & 0 & 1 \\ 
Yes: lives with just dad & 7,872 & 0.016 & 0.125 & 0 & 1 \\ 
Yes: lives with other adults & 7,872 & 0.024 & 0.151 & 0 & 1 \\ 
Yes: lives with both parents & 7,872 & 0.764 & 0.425 & 0 & 1 \\ 
Yes: lives with parents but separated & 7,872 & 0.088 & 0.283 & 0 & 1 \\ 
Do participate in sports at school & 7,872 & 0.370 & 0.483 & 0 & 1 \\ 
Do participate in sports outside of school & 7,872 & 0.723 & 0.447 & 0 & 1 \\ 
Do participate in theater/drama & 7,872 & 0.119 & 0.323 & 0 & 1 \\ 
Do participate in music & 7,872 & 0.352 & 0.478 & 0 & 1 \\ 
Do participate in other arts & 7,872 & 0.177 & 0.381 & 0 & 1 \\ 
Do participate in other school clubs & 7,872 & 0.234 & 0.424 & 0 & 1 \\ 
Date people at this school & 7,872 & 0.216 & 0.412 & 0 & 1 \\ 
Do lots of homework & 7,872 & 0.451 & 0.498 & 0 & 1 \\ 
Do read books for fun & 7,872 & 0.336 & 0.472 & 0 & 1 \\ 
\hline
\hline
\multicolumn{6}{c}{ \begin{minipage}{13 cm}{\footnotesize{\textit{{Notes}}: The data-set corresponds to the education data on conflict collected by \cite{paluck2016changing}. We focus our analysis on the control group. The final cleaned sample consists of no missing data for any of the above covariates.}}
\end{minipage}} \\
\end{tabular}
\end{table}

\begin{table}[h] \centering \renewcommand*{\arraystretch}{1.2}\caption{Summary Statistics}
\begin{adjustbox}{width=0.95\columnwidth,center}
\label{summarytable}
\begin{tabular}{p{4cm}p{1cm}p{1cm}p{1cm}p{1cm}p{1cm}p{1cm}}
\hline
\hline
Variable & N & Mean & Median & Sd & Min & Max \\ 
\hline
Median Age 2010 & 2807 & 40 & 40 & 4.7 & 22 & 63 \\ 
Census Population 2010 & 2807 & 106653 & 31255 & 321544 & 2882 & 9818605 \\ 
3-Yr Diabetes 2015-17 & 1415 & 51 & 23 & 105 & 10 & 2515 \\ 
Diabetes Percentage & 2807 & 11 & 10 & 3.7 & 1.5 & 33 \\ 
Heart Disease Mortality & 2807 & 188 & 182 & 45 & 56 & 603 \\ 
Stroke Mortality & 2807 & 41 & 40 & 8.3 & 12 & 100 \\ 
Smoker's Percentage & 2807 & 18 & 17 & 3.4 & 5.9 & 38 \\ 
Resp Mortality Rate 2014 & 2807 & 65 & 63 & 17 & 18 & 161 \\ 
Hospitals & 2807 & 1.6 & 1 & 2.7 & 0 & 76 \\ 
ICU beds & 2807 & 26 & 4 & 87 & 0 & 2126 \\ 
Mask Never & 2807 & 0.076 & 0.066 & 0.056 & 0 & 0.43 \\ 
Mask Rarely & 2807 & 0.08 & 0.071 & 0.054 & 0 & 0.38 \\ 
Mask Sometimes & 2807 & 0.12 & 0.11 & 0.058 & 0.003 & 0.42 \\ 
Mask Frequently & 2807 & 0.2 & 0.2 & 0.062 & 0.047 & 0.55 \\ 
Mask Always & 2807 & 0.52 & 0.51 & 0.15 & 0.12 & 0.89 \\ 
Population Density per SqMile 2010 & 2807 & 228 & 53 & 784 & 0.3 & 17179 \\ 
SVI Percentile & 2807 & 0.52 & 0.52 & 0.28 & 0 & 1 \\ 
3-Yr Mortality Age 65-74Years 2015-17 & 2807 & 179 & 68 & 424 & 10 & 10827 \\ 
3-Yr Mortality Age 75-84 Years 2015-17 & 2807 & 224 & 84 & 537 & 10 & 14063 \\ 
3-Yr MortalityAge 85 Years 2015-17 & 2807 & 300 & 96 & 792 & 10 & 21296 \\ 
Deaths & 2807 & 113 & 32 & 379 & 0 & 10345 \\ 
Cases & 2807 & 6873 & 1951 & 23463 & 48 & 770915\\ 
\hline
\hline
\multicolumn{7}{c}{ \begin{minipage}{13 cm}{\footnotesize{\textit{{Notes}}: The data-set has been obtained from the data repository constructed for the paper \cite{altieri2020curating}. The description of the data variables can be obtained from the GitHub repo provided in this \href{https://github.com/Yu-Group/covid19-severity-prediction/blob/master/data/list_of_columns.md}{link}. This table shows the final set of variables used for the empirical illustration. The deaths and cases correspond to the cumulative COVID deaths and COVID cases as of December 31, 2020. }}
\end{minipage}} \\
\end{tabular}
\end{adjustbox}
\end{table}

\end{document}